\documentclass[%
 reprint,twocolumn,
 groupedaddress,
 superscriptaddress,
 amsmath,amssymb,
 aps,
 prx,
]{revtex4-2}

\usepackage{graphicx}% Include figure files
\graphicspath{
    {.} % document root dir
    {figures/}
}
\usepackage{dcolumn}% Align table columns on decimal point
\usepackage{hhline}
\usepackage{bbold}
\usepackage{bm}% bold math
\usepackage{hyperref}% add hypertext capabilities
\usepackage{xcolor}
\usepackage{amsthm}
\usepackage{CJKutf8}

\newcommand{\eqnref}[1]{Eq.\,\eqref{#1}}
\newcommand{\figref}[1]{Fig.\,\ref{#1}}
\newcommand{\tabref}[1]{Tab.\,\ref{#1}}
\newcommand{\secref}[1]{Sec.\,\ref{#1}}
\newcommand{\appref}[1]{Appendix\,\ref{#1}}

\newcommand{\ud}{\mathrm{d}}

\def\bea{\begin{eqnarray}}
\def\eea{\end{eqnarray}}
\def\be{\begin{equation}}
\def\ee{\end{equation}}
\def\beal{\begin{aligned}}
\def\eeal{\end{aligned}}
\def\bew{\begin{widetext}}
\def\eew{\end{widetext}}
\def\beit{\begin{itemize}}
\def\eeit{\end{itemize}}
\def\bea{\begin{array}}
\def\eea{\end{array}}

\def\k{\kappa}

\newtheorem{AHT}{Theorem}

\begin{document}

\begin{CJK}{UTF8}{gbsn}

\title{Deep reinforcement learning for quantum Hamiltonian engineering}

\author{Pai Peng (彭湃)}
\email[]{paipeng@mit.edu}
\affiliation{Department of Electrical Engineering and Computer Science, Massachusetts Institute of Technology, Cambridge, MA 02139}
\author{Xiaoyang Huang}
\thanks{P.P. and X. H contributed equally to this work.}
\affiliation{
Research Laboratory of Electronics, Massachusetts Institute of Technology, Cambridge, Massachusetts 02139, USA
}
\author{Chao Yin}
\affiliation{
Research Laboratory of Electronics, Massachusetts Institute of Technology, Cambridge, Massachusetts 02139, USA
}
\author{Linta Joseph}
\affiliation{Department of Physics and Astronomy, Dartmouth College, Hanover, NH 03755, USA}
\author{Chandrasekhar Ramanathan}
\affiliation{Department of Physics and Astronomy, Dartmouth College, Hanover, NH 03755, USA}
\author{Paola Cappellaro}\email[]{pcappell@mit.edu}
\affiliation{Department of Nuclear Science and Engineering, Massachusetts Institute of Technology, Cambridge, MA 02139}
\affiliation{
Research Laboratory of Electronics, Massachusetts Institute of Technology, Cambridge, Massachusetts 02139, USA
}

\date{\today}% It is always \today, today,

\begin{abstract}
Engineering desired Hamiltonian in quantum many-body systems is essential for applications such as quantum simulation, computation and sensing. 
Conventional quantum Hamiltonian engineering sequences are designed using human intuition based on perturbation theory, which may not describe the optimal solution and is unable to accommodate complex experimental imperfections. 
Here we  numerically search for Hamiltonian engineering sequences using deep reinforcement learning (DRL) techniques and experimentally demonstrate that they outperform celebrated sequences on a solid-state nuclear magnetic resonance quantum simulator. As an example, we aim at decoupling  strongly-interacting spin-1/2 systems. We train DRL agents in the presence of different experimental imperfections and verify robustness of the output sequences both in simulations and experiments. Surprisingly, many of the learned sequences exhibit a common pattern that had not been discovered before, to our knowledge, but has an meaningful  analytical description. We can thus  restrict the searching space based on this control pattern, allowing to search for  longer sequences,  ultimately leading to sequences that are robust against 
dominant imperfections in our experiments. Our results not only demonstrate a general method for quantum Hamiltonian engineering, but also highlight the importance of combining black-box artificial intelligence with understanding of physical system in order to realize experimentally feasible applications.
\end{abstract}

\maketitle
\end{CJK}
\section{Introduction}
Controllable quantum many-body systems hold great promise  not only to expand our understanding of fundamental physics, such as information scrambling and non-equilibrium phases of matter, but also to yield revolutionary technologies in computation, simulation and sensing. 
A core task of quantum control is to combine elementary control units to engineer  desired quantum Hamiltonians. Although it is relatively easy to derive the (approximate)  Hamiltonian resulting from a given control sequence, the inverse problem of designing optimal control sequences for a target Hamiltonian is highly challenging. The problem was tackled in Nuclear Magnetic Resonance (NMR) through the development of average Hamiltonian theory (AHT)~\cite{Haeberlen68}. 
Many of the celebrated sequences in NMR are designed from intuition and experience, based on low-order expansions in AHT~\cite{Waugh68,Cory90,PhysRevLett.55.1923,RevModPhys.76.1037,Tycko90,Boutis03,Mansfield73,Rhim73,TakegoshiCPL85,Tycko99}.
Unfortunately, simply relying on  intuition makes it difficult to find  generalizations and  capture higher-order effects and control imperfections that might be crucial in experiments. 
Conventional numerical optimization methods, such as gradient ascent pulse engineering~\cite{Khaneja05} and chopped random basis~\cite{Doria11}, can yield optimal solutions, but are fundamentally limited to low-entanglement subspace if applied to a non-integrable many-body system~\cite{Lloyd14l} and are most efficient for smooth control landscape without too many local optima. Therefore, they are more suitable for optimizing individual (continuous) pulse shapes or a short composite pulse, rather than a pulse sequence containing tens of pulses. While phase and amplitude modulated continuous decoupling are amenable to  gradient ascent optimization~\cite{SAKELLARIOU2000253}, their experimental implementation has been more limited than pulsed methods. 
%extends the field of time-suspended pulse engineering. Though it maybe more favorable in designing them using classical gradient method, continuous engineering is experimentally hard to achieve.\XYH{Not sure.}
%directly optimize over control parameters via calculating gradient, are often limited to small parameter space and smooth control landscape without too many local optima.

Recently, artificial intelligence, in particular reinforcement learning (RL) with deep neural networks, has surpassed human intelligence in many complex tasks such as Go~\cite{silver2016} and StarCraft II~\cite{Vinyals19}. As a subfield of Machine Learning (ML), RL differs from (un)-supervised learning by learning through exploration and exploitation based on the reward of the result. 
In stark contrast to conventional optimization methods, RL is a \textit{model-free} method, which only requires minimum knowledge to find the reward. This  matches closely the task of Hamiltonian engineering where human intuition into the optimal pulse sequence is limited and might  be biased. 
Deep neural networks (DNN) provide a versatile and powerful way to reparametrize a large search space. Unlike linear optimization, DNN are capable of doing both linear and non-linear mathematical manipulation to turn the input into the output (for RL and DNN, see recent reviews~\cite{sutton2018reinforcement,lecun2015}). 
In the quantum physics context, RL has been shown to provide successful strategies for quantum state preparation~\cite{Bukov18, Zhang19, Chen14, Chen19a, Albarran-Arriagada18, Mackeprang20}, quantum gate design~\cite{Niu19, Dalgaard20, Daraeizadeh20}, quantum communication~\cite{Wallnofer20}, quantum error correction~\cite{Fosel18, Nautrup19, Sweke21}, quantum state transfer~\cite{Zhang18}, and quantum sensing~\cite{Schuff20}. 
% and approximating ground states~\cite{Sotnikov20} % this was experiment... also seems different
\begin{figure*}[!htbp]
\includegraphics[width=0.42\linewidth]{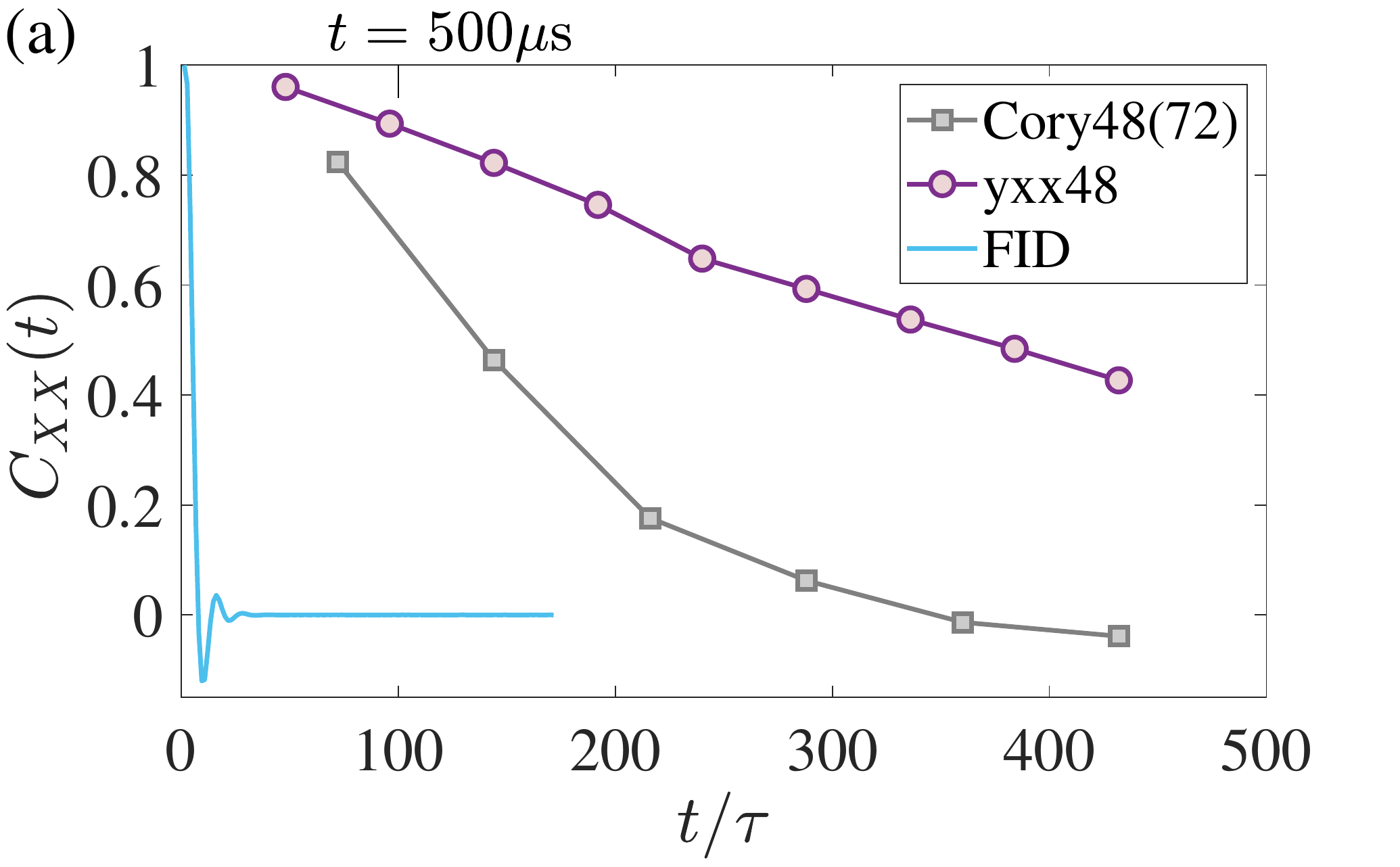}
\includegraphics[width=0.54\linewidth]{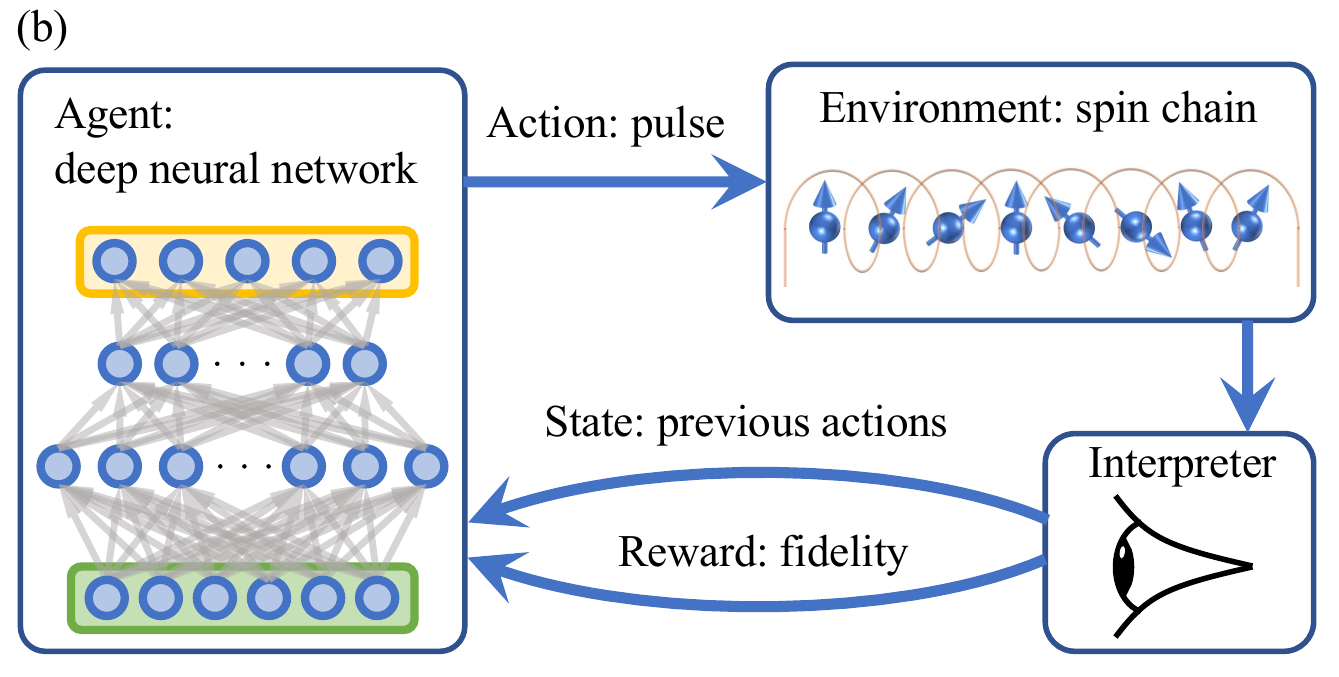}
\includegraphics[width=0.98\linewidth]{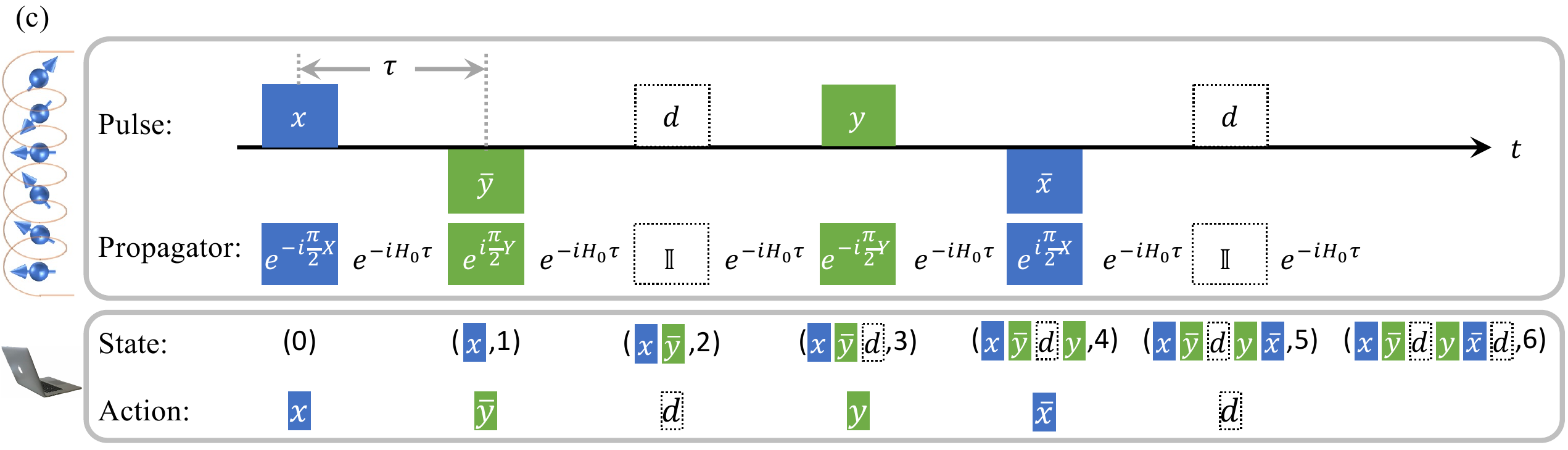}
\caption{
\label{fig:model}
(a) Decay of the x-correlation, $C_{XX}$, under free evolution (no pulses, blue curve), Cory48 pulse sequence (grey curve) and yxx48 obtained from RL (purple curve).
(b) High-level RL protocol for Hamiltonian engineering. The agent, realized as a deep neural network, takes an action based on the  current state. This action  applies the  corresponding control operation to the spin chain. The interpreter updates the state based on the chosen action, and feeds it to the agent to take next action. When the final time is reached, the interpreter calculates the reward, which is then used to optimize the agent.
(c) States and actions, and corresponding quantum operations, illustrated using the WAHUHA sequence~\cite{Waugh68}, which can be discretized into 6 times steps. $x,\bar{x},y,\bar{y}$ in the colored blocks denote $\pi/2$ pulses along $x,-x,y,-y$, respectively. $d$ in the dashed blocks denotes no pulse action (delay).
}
\end{figure*}
Although  RL has  in principle demonstrated its advantages for quantum applications via numerical studies, its practical implementation in experiments is still challenging due to non-ideal conditions arising from noise and control imperfections.

Here we apply RL with DNN [deep reinforcement learning (DRL)] to quantum Hamiltonian engineering and experimentally demonstrate its advantage  in a non-integrable system. 
We focus on the task of decoupling a spin-1/2 system with dipolar interaction (i.e. the target Hamiltonian is zero), which is directly useful for quantum memories~\cite{Ladd05}; our method can be further applied to other quantum engineering scenarios by simply replacing the reward function. 
As finding the optimal control in a non-integrable quantum many-body system is impractical~\cite{Lloyd14l}, we restrict the control space by allowing the machine learner to choose one of five actions at a time  (no pulse or a $\pi/2$ pulse along one of 4 axes), with a fixed delay time between actions, until the maximum time is reached. 
We then numerically calculate the unitary propagator of the resulting pulse sequence, and use the fidelity with respect to target propagator as the reward. 
The control is restricted to experimentally feasible operations, while still encompassing a wide range of target Hamiltonians that can exhibit integrable, ergodic, localized or prethermal behaviors~\cite{Wei18, Wei19, Peng21}.
%Compared with conventional digitizing method where each action only generates a tiny evolution, our scheme is experimentally realistic and has successfully yielded ergodic, many-body localizated and prethermal Hamiltonian~\cite{Wei18, Wei19, Peng19x}. 
The restriction leads to a complicated and nonconvex control landscape~\cite{Moore12c}, that would not be amenable to conventional optimization. We thus utilize DNN to reparametrize the control space and implement a state-of-the-art gradient-free method to optimize the neural networks~\cite{Such17}.

We not only apply the DRL to the idealized scenario, but also incorporate imperfections, such as pulse frequency offset, on-site disordered field, pulse angle error and finite pulse width, to mimic realistic experiments. We test the performance of the DRL pulse sequences using solid-state nuclear spin systems, and the sequences indeed show the expected robustness even in experiments. Surprisingly, although it is generally believed that symmetric sequences have better performance~\cite{Mansfield71,Haeberlen76}, many of the high-reward sequences found by DRL are not symmetric. Instead, they obey a common ``\textit{yxx} pattern'' which has not been found before to the best of our knowledge. We analytically explain the advantage of the \textit{yxx} pattern using 
AHT. Furthermore, the restriction to pulse sequences exhibiting the \textit{yxx} pattern significantly reduces the search space, thus enabling to find longer and more powerful sequences. As a result, we discover sequences that are robust against all relevant imperfections and outperform the celebrated Cory48 decoupling sequence~\cite{Cory90} in experiments [see Fig.~\ref{fig:model}(c)]. Our work demonstrates that some long-established knowledge may not be optimal for quantum Hamiltonian engineering, while pure black-box DRL is also resource consuming. It is beneficial to combine both human knowledge and artificial intelligence for practical applications.

This paper is organized as follows. In \secref{sec:RL} we explain how to model the Hamiltonian engineering task as a RL problem and our learning algorithm. \secref{sec:background} introduces our experimental system and average Hamiltonian theory. The learned sequences together with their experimental tests are presented in \secref{sec:seq}, before drawing our conclusions.

\section{Reinforcement learning}\label{sec:RL}

\subsection{Modeling}\label{sec:model}

\begin{figure*}[!htbp]
\includegraphics[width=0.9\linewidth]{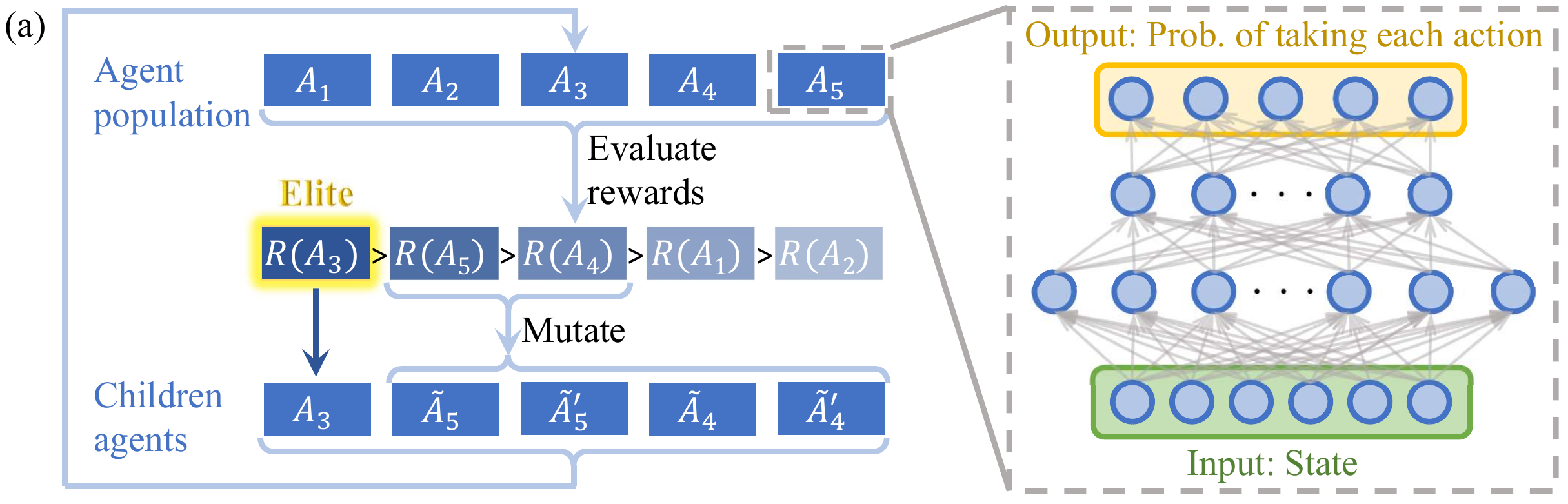}
\includegraphics[width=0.48\linewidth]{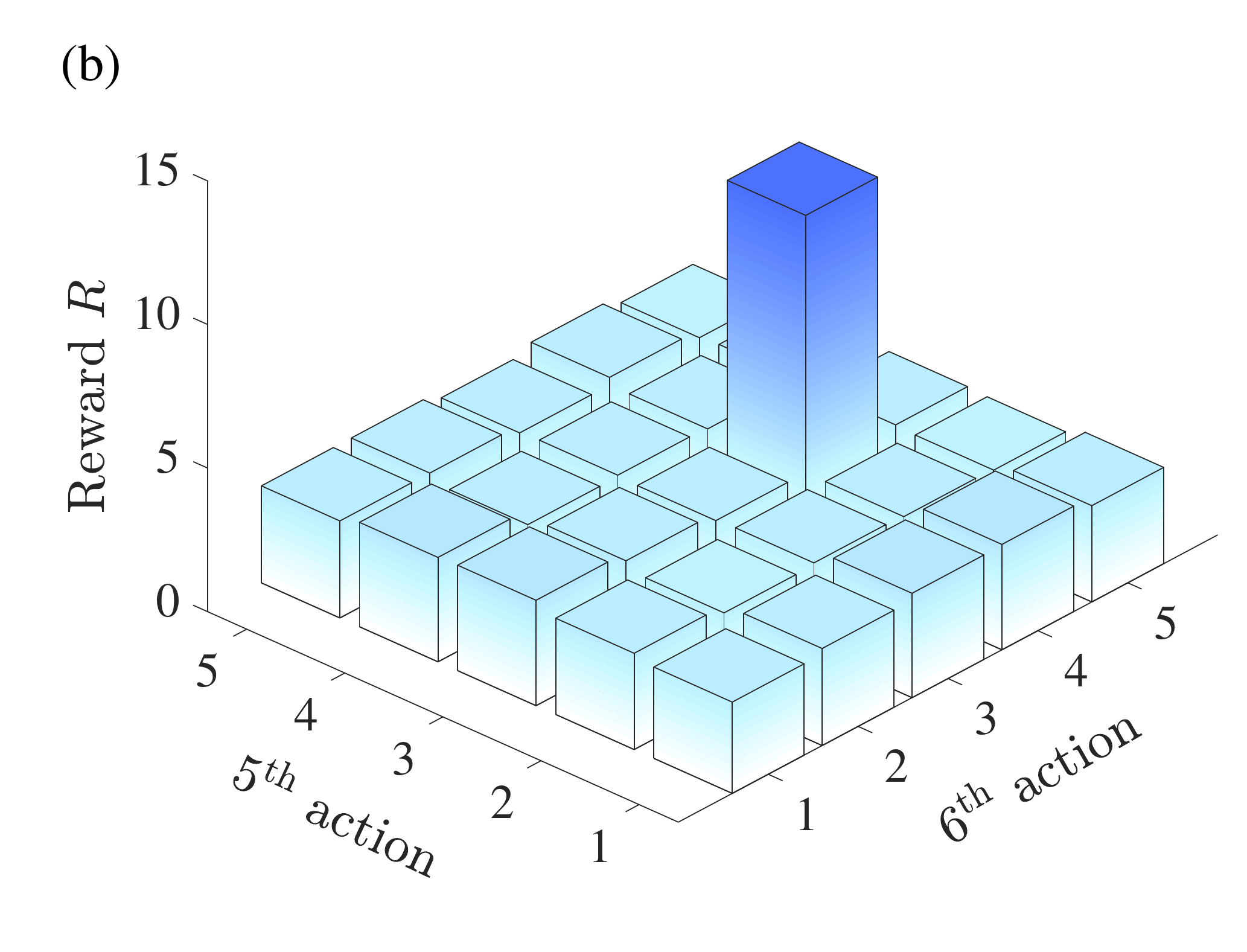}
\includegraphics[width=0.48\linewidth]{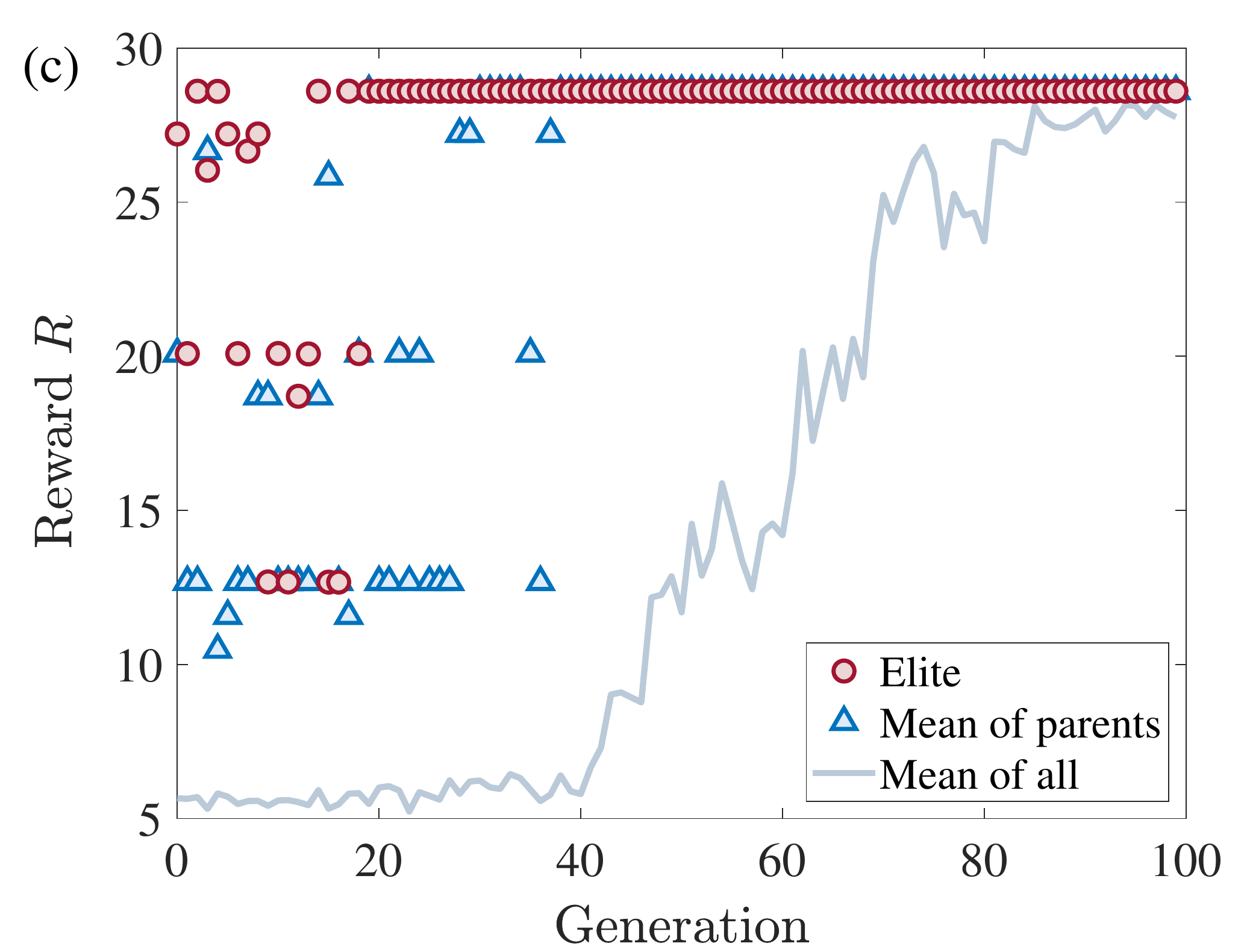}
\caption{
(a) Schematic of the RL algorithm. We keep a population of agents (in this cartoon the population size is $N_p=5$, while in actual implementation the parameters are specified in \tabref{tab:seq}). Each agent is a DNN that generates the probability of taking each action based on the current state. 
We evaluate the reward of all agents and choose the top $P$ agents as parents (here $P=3$, see  \tabref{tab:seq} for the actual implementation parameters). The top agent (Elite) is moved to the next generation without any change, while the other parents are slightly mutated to generate the next generation of agents.
(b) Illustration of the (reward) control landscape for a 6-pulse dynamical decoupling sequence in the ideal scenario. For simplicity, here we only show the case where the first four pulses are fixed to be $\{d, x, -y, d\}$, with $d$ denoting the no pulse action. x and y-axis represent the 5$^\text{th}$ and 6$^\text{th}$ action.
(c) Learning curve of DRL for 12-pulse dynamical decoupling in the ideal scenario. Here $N_p=201$, $P=21$.
}
\label{fig:algorithm}
\end{figure*}
We consider quantum Hamiltonian engineering in a spin-1/2 system where the internal Hamiltonian $H_0$ is the  secular dipolar interaction along the z-axis
\begin{equation}\label{eqn:h0}
    H_0\equiv D_z= \frac{1}{2}\sum_{j<k}^N J_{jk}\left(3S_z^j S_z^k - \vec{S}^j\!\cdot\!\vec{S}^k \right),
\end{equation}
%$H_0= \sum_{j<k}^L \frac{J}{2|k-j|^{3}}\left(3S_z^j S_z^k - \vec{S}^j\!\cdot\!\vec{S}^k \right)$ 
where $\vec{S}_j=(S_x^j,S_y^j,S_z^j)^T$  are spin-1/2 operators of the \(j\)-th spin $(j=1,\cdots,N)$ and 
$J_{jk}$ the coupling strength between spins $j$ and $k$. Later we will also use $D_x$ and $D_y$ defined in the similar way. In the training process we consider  a 1D spin chain with nearest coupling only, i.e. $J_{jk}=J\delta_{j+1,k}$ ($\delta_{i,j}$ is the Kronecker delta). This is convenient because 1D systems can be efficiently simulated on a classical computer; still, experimental validations are performed in 3D systems. 
%For decoupling, the dimension of the system is not important.  
We aim at decoupling the spins, that is, the target Hamiltonian is zero. As we will show, for decoupling purpose, the dimensionality does not play a crucial role; instead the performance is mostly determined by the symmetry of the Hamiltonian.

To decouple the interaction, it is sufficient to apply global rotations to the system, which can be easily implemented in experiments. We further restrict our control to $\pi/2$ global rotations along the $x,y,-x$ or $-y$ directions. These controls are available to almost every quantum platform and are known to be sufficient for decoupling, since they are the building blocks of  many celebrated decoupling sequences such as WAHUHA~\cite{Waugh68} and Cory48~\cite{Cory90}. Experimentally, the controls cannot be applied arbitrarily fast due to  pulse ring-down and apparatus dead-time. Instead, there is usually a minimal delay time $\tau$ in between  pulses.

The spin chain system and the  control rules constitute the environment of RL. We then need to set up the RL agent that interacts with the environment. RL works by building agents that choose a (sub)optimal action at any given time based on the current state (which collects all the previous actions). The action changes the status of the environment and updates the state, and the agents are optimized based on a reward determined by the environment, as shown in Fig.~\ref{fig:model}(b). 
We first discretize the time in steps of $\tau$ -- we consider only $t=m\tau$ with $m=0,1,2\cdots$. At each discrete time $t=m\tau$, the agent chooses an action $A_m$ from a set of five possible actions (no pulse or  a $\pi/2$ pulse along the $\pm x,\pm y$ directions). The state $S(m)$ is simply defined as a list containing all previous actions and the integer $m$, $S(m)=[A_0,A_1,\cdots,A_{m-1},m]$. 
% The inclusion of step $m$ into the state is not frequently used, but it indicates a properly ordered sequence here. 
As an example, we show the state and action at each time step of the WAHUHA sequence in Fig.~\ref{fig:model}(c). The process ends when the final time $M\tau$ is reached. At this point the environment has received $M$ actions $[A_0,A_1,\cdots,A_{M-1}]$ and undergone a unitary evolution with propagator
\begin{equation}
    U({\{A_m\}})
    =e^{-iH_0\tau}U_{A_{M-1}}e^{-iH_0\tau}U_{A_{M-2}}\cdots e^{-iH_0\tau}U_{A_0},
\end{equation}
where $U_{A_m}$ describes the evolution induced by the action $A_m$. If $A_m$ is ``no pulse'', then $U_{A_m}=\mathbb{1}$; if $A_m$ is a $\pi/2$ pulse, e.g. along $x$, then $U_{A_m}=e^{-i(\pi/2)X}$ with $X=\sum_j S_x^j$ being the collective spin-$x$ operator. Later we will also use $Y=\sum_j S_y^j$ and $Z=\sum_j S_z^j$. 
To compare pulse sequences of different lengths $M$, it is useful to consider $U_\tau({\{A_m\}})=\left[U({\{A_m\}})\right]^{1/M}$ the effective propagator for $t=\tau$.
How close is the engineered propagator $U_\tau({\{A_m\}})$ to the target propagator $U_{tgt}$ is characterized by the fidelity $F(\{A_m\})=\mathrm{Tr}|U_\tau({\{A_m\}})U_{tgt}^\dagger|/{2^N}\in[0,1]$, where $N$ is total number of spins. 
For the decoupling task $U_{tgt}=\mathbb{1}$. 
Since many good pulse sequences have near unity fidelity, we define the reward function as 
$R(\{A_m\})=-\ln\left[1-F(\{A_m\})\right]$ to emphasize the small infidelity. 
The fidelity also provides a lower bound for observable correlations~\cite{Yin21,Peng21,Heyl19}, which can be directly measured in experiments~\cite{Wei19}, and it is thus a good metric to assess the sequence~\cite{Bukov18}.

Imperfections can be easily incorporated into this model by changing the free evolution propagator or the pulse propagator. In this work we consider the following imperfections that are most evident in our experiments and frequently encountered in other systems: pulse frequency offset, on-site disorder, pulse angle error and finite pulse width. 
A frequency offset exists when the pulse frequency differs from spin resonance frequency. This can be modeled by adding the offset Hamiltonian $\Delta Z$ to $H_0$, where $\Delta$ is the amount of frequency offset. 
On-site disorder describes differences in the  frequency of each spin; the  deviation from the (nominal) mean frequency is a random variable. On-site disorder can be modeled by including the disorder Hamiltonian $\sum_j w_j S_z^j$ in $H_0$, with $w_j$ uniformly distributed in $[-W,W]$. 
The disorder Hamiltonian is very similar to the frequency offset, so a sequence that is robust against frequency offset is also typically robust against disorder. 
Therefore, in the training process we consider the frequency offset only, and in the test process we verify the two are indeed closely related. 
Angle errors happen when the rotation angle due to the pulse deviates from $\pi/2$ by an amount $\epsilon$. We assume this deviation is the same for all pulses, and thus it can modeled by changing all pulse propagators in the same way. 
For example, a $x$ pulse with an angle error is described by the propagator $e^{-i(\pi/2)X(1+\epsilon)}$. 
In experiments,  pulses are not instantaneous (delta-pulses) but   have a finite width. During the pulse time $t_w$ the spins  interact with each other, yielding a   propagator  $e^{-i(\pi X/2+H_0t_w)}$ for the $x$ pulse, similarly for $y$ and $z$ pulses. Beyond what we consider above, many other imperfections can be included by just modifying the reward, as long as the imperfection can be efficiently modeled. 

We note that our optimization setup  differs from the ones commonly used in numerical pulse engineering, where the   time is discretized into tiny time steps and each action only applies a small evolution to the system~\cite{Khaneja05,Bukov18,Niu19}. Instead in our protocol, neither the free evolution time $\tau$ nor the pulse rotation angle  need to be small. The advantage of this scheme is two folds: (i) $\pi/2$ pulses are usually available and well calibrated, while  modulating the control drive over short timescale poses challenges on the hardware and is more difficult and inefficient to calibrate; (ii) by taking a larger step per action, our method is more suitable to finding long pulse sequences, while previous methods are typically used to optimize single gates. These advantages are accompanied by a  worse control landscape [see an example in \figref{fig:algorithm}(b)]. However, the deep neural network and gradient-free optimization method successfully solve this issue, as we show in the next subsection.

\subsection{Algorithm}\label{sec:algorithm}

We first explain how the agent works [see the dashed box in Fig.~\ref{fig:algorithm}(a)]. The agent is a DNN that takes the state as input and generates the next action as the output, as introduced in \secref{sec:model}. The agent DNN contains two hidden linear layers, with the number of neurons in each layer proportional to the input and output size. We use rectified linear unit (ReLU)~\cite{vinod2010} as the activation function. At each step $m$, the agent takes the state $S(m)$ as the input  and generates 5 positive numbers corresponding to the probability of taking the 5 actions. Then the action is chosen randomly according to the probabilities, and the state now become $S(m+1)$. We apply the above procedure starting from $m=0$ until the maximum step $M$ is reached, then we get a output sequence from the agent.

Here we explain how we optimize the agents. The process is illustrated in Fig.~\ref{fig:algorithm}~(a). 
We start with $N_p$ agents, and for each agent generate 3 sequences (note the selection of actions is a random process so the 3 sequences may not be the be same \footnote{We also varied the number of sequences and found it does not affect the performance of RL}), and select the highest reward among the 3 sequences as the reward  of the agent. The reward is obtained on a 3-spin system and we verify that going to larger systems does not change our results.
Sorting the population of agents by the reward function in  descending order, we apply the truncation selection to choose the top $P$ individuals as the parents. 
Among the parents, we further select the most promising parent, the so-called \textit{Elite}, from all the parents by regenerating a few sequences (typically 5) and comparing their rewards.
The Elite will be included in the children generation without any change. 
Every parent agent other than the Elite will be mutated by adding a random Gaussian noise multiplying mutation power $\mu$ to all the DNN parameters to generate $(N_p-1)/(P-1)$ children agents. 
The mutation process plays the role of ``exploration'' (search in a large space) in RL. Too much  exploration (large $\mu$) will result in excessive randomness, making the process closer to a pure random search; too little exploration instead (small $\mu$) might leave  the RL stuck into a local minimum. 
In practice, we decrease $\mu$ during the learning process (so that in the beginning we explore a large space and later we search near the good agents) following the function
\begin{equation}
    \mu(g)=\mu_0(1-g/G),
\end{equation}
where $g=1,...,G$ denotes the agent generation. In this paper, we empirically set $\mu_0=0.05$ and $G=100$. We repeat the process until the maximum number of generations $G$ is reached.

With the truncation selection and mutation, we are able to balance  \textit{exploration} (search in a large space) and  \textit{exploitation} (focus on the promising area). One example of the learning curve of DRL is shown in \figref{fig:algorithm}(c). After 20 generations, the Elite DNN starts showing a convergent reward towards the optimal one (global minimum for this case); after 40 generations, all the parent agents (mean of parents) begin to converge; the entire agent population converges after around 90 generation. 
When approaching the end of learning, $\mu$ becomes small,  meaning little exploration but great exploitation. With little random noise, children agents are able to reproduce the optimal reward consistently, indicating the convergence of the algorithm.

\section{Experimental and theoretical background}
\label{sec:background}
\subsection{Experimental system}
\label{sec:system}
We use a solid-state NMR quantum simulator to experimentally test the performance of RL pulse sequences in realistic conditions.
Most of the experimental results presented in this work are obtained from a single crystal of CaF$_2$, where the $^{19}$F nuclear spins-1/2 form simple cubic structure. 
The sample is placed in a strong magnetic field (7~T) at room temperature. The nuclear spins interact via the secular dipolar interaction 
$D_z=\frac{1}{2}\sum_{j<k}^N J_{jk}\left(3S_z^j S_z^k - \vec{S}^j\!\cdot\!\vec{S}^k \right)$ with $J_{jk}=\hbar\gamma_F^2\frac{3\cos(\theta_{jk})^2-1}{|\vec r_{jk}|^3}$, where $\gamma_{F}$ is the gyromagnetic ratio of $^{19}$F nuclei, $\vec{r}_{jk}$ is the displacement between spins $j$ and $k$, $\theta_{jk}$ is the angle between $\vec r_{jk}$ and the magnetic field (aligned with the $z$-axis). The maximum possible $J_{jk}$ is 65.8~krad/s for CaF$_2$~\cite{SM-RL}. 
The relaxation time $T_1$ of our sample is $T_1\approx 14$~s, much longer than the time scale we explore here. The collective spin rotations are realized by on-resonance RF pulses with a $t_w=1.02~\mu$s $\pi/2$ pulse width. 
We can also artificially introduce and tune errors in addition to intrinsic imperfections. We introduce angle error by 
setting the pulse width to $1.02(1+\epsilon)~\mu$s. We can also use off-resonance pulses to introduce a frequency offset.

\begin{table*}[bht]
\begin{tabular}{|c|c|c|}
\hline
\textbf{Name} & \textbf{Sequence} & \textbf{Training parameter}\\ \hline Ideal6 & y, x, x, y, -x, -x & $\Delta=0, \epsilon=0, t_w=0, N_p=201, P=11$\\ \hline
Offset48 & \begin{tabular}[c]{@{}c@{}}$-y, -y, -y, -x, -y, -y, -y, -y, -x, -x, -x, -y,$\\ $-x, -x, -y, -y, -y, -x, -x, -x, -y, -y, -y, -y,$\\ $-x, -x, -y, -x, -x, -x, -x, -x, -x, -y, -y, -y,$\\ $-x, -x, -x, -x, -x, -x, -y, -x, -x, -x, -y, -x$\end{tabular} & $\Delta=0,\pm3J,\pm5J, \epsilon=0, t_w=0, N_p=3001, P=31$ \\ \hline
Angle12 & $-y, x, -x, y, -x, -x, -y, x, -x, y, x, x$  & \begin{tabular}[c]{@{}c@{}}$\Delta=0, \epsilon=0.05, t_w=0$\\ and $\Delta=0, \epsilon=0, t_w=0.1\tau, N_p=801, P=21$\end{tabular} \\ \hline
PW12 & $-y, -x, -x, x, x, y, -x, -x, -y, y, x, x$ & $\Delta=0, \epsilon=0, t_w=0.1\tau, N_p=801, P=21$\\ \hhline{|=|=|=|}
yxx48 & \begin{tabular}[c]{@{}c@{}} $y, -x, -x, y, -x, -x, -y, x, x, y, -x, -x,$\\ $-y, x, x, -y, x, x, y, -x, -x, y, -x, -x,$ \\$-y, x, x, y, -x, -x, -y, x, x, -y, x, x,$ \\$y, -x, -x, -y, x, x, y, -x, -x, -y, x, x$\end{tabular} & \begin{tabular}[c]{@{}c@{}}$\Delta=J, \epsilon=0.05, t_w=0, N_p=801, P=21,$\\ with \textit{yxx} restriction \end{tabular}\\ \hline
%$\Delta=J, \epsilon=0.05, t_w=0, N=801, P=21,$\\ with \textit{yxx} restriction \hline
yxx24 & \begin{tabular}[c]{@{}c@{}}$-y, x, -x, y, -x, -x, y, -x, x, -y, x, x,$\\$y, -x, x, -y, x, x, -y, x, -x, y, -x, -x $ \end{tabular} & Built from Angle12 \\ \hline 
\end{tabular}
\caption{\label{tab:seq}
Representative DRL pulse sequences under different training conditions. Angle12 appears in two training conditions. yxx24 is build from Angle12 using AHT analysis (see~\appref{appsec:yxx24intuition}).}
\end{table*}

At room temperature and in a strong magnetic field along the z axis, the initial state of an ensemble of $^{19}$F nuclear spins is described by the density matrix $\rho(0)\!\approx\!(\mathbb{1}\!-\!\epsilon' Z)/2^N$, with $N$ being the number of spins and $\epsilon'\! \sim\! 10^{-5}$. The identity part of the density matrix does not contribute to the NMR signal, so we only care about the deviation from it, $\delta\rho=4Z/N$, which has been normalized such that $\text{Tr}(\delta\rho Z)/2^{N}=1$. 
NMR experiments measure the collective magnetization along the $x$ axis, i.e. the signal is $\mathrm{Tr}(\delta\rho(t)X)/2^{N}$. 
If we regard the density matrix $\delta\rho$ as an observable, this signal is mathematically equivalent to an infinite-temperature correlation $\mathrm{Tr}(\delta\rho(t)X)/2^{N}\equiv\langle\delta\rho(t)X\rangle_{\beta=0}$. 
Using collective RF pulses, we can rotate the initial state and the observable to be $X,Y$ or $Z$. Therefore, we can measure the three autocorrelations $C_{XX}(t)=4\langle X(t)X\rangle_{\beta=0}/N$ and $C_{YY}$(t), $C_{ZZ}(t)$ defined in a similar way. 
Although in principle to get the propagator fidelity we have to measure autocorrelations of all observables, in the Supplementary Material we show that the geometric average of these three autocorrelations, $C_{avg}\equiv(C_{XX}C_{YY}C_{ZZ})^{1/3}$ already approximates the behavior of the propagator fidelity.

To   experimentally investigate  on-site disorder [Fig.~\ref{fig:yxx_t}(a)], we work with $^{19}$F nuclear spins in fluorapatite (FAp)~\cite{VanderLugt64}. The $^{31}$P nuclear spins-1/2 in the crystal are randomly polarized, giving rise to a disorder Hamiltonian $H_{dis}=\sum_j h_j S_z^j$, with $h_j$ being a random variable representing the disordered field at $j^\mathrm{th}$ $^{19}$F nucleus. Interaction between $^{19}$F nuclear spins is also given by the secular dipolar interaction as in CaF$_2$ but with a lower maximum possible strength $32.7$~krad/s. The $^{19}$F nuclei form a quasi-1D structure, as the interaction along the z-direction is $\sim40$ times stronger than along the other two directions. Although the quasi-1D nature is not important in the context of this work, it is useful for quantum simulation~\cite{Wei18,Wei19,Peng21,Yin21}. The relaxation time for the FAp crystal is $T_1\approx 0.8$s, shorter than for the CaF$_2$ sample, but still much longer than the duration of a single experiment.

\subsection{Average Hamiltonian theory}

AHT~\cite{Haeberlen68} is useful in understanding the performance of different pulse sequences, so we briefly review it here.  A quantum system under a time-dependent control can be generally described by the Hamiltonian $H(t)=H_0+H_c(t)$, with $H_0$ the intrinsic Hamiltonian and $H_c$ the control Hamiltonian. For the pulsed control case, $H_c(t)$ is piece-wise constant and nonzero only within the pulse width. 
AHT starts by defining the \textit{toggling frame}, an interacting frame that rotates with $H_c$, i.e., $U_{c}(t)=\mathcal{T}[e^{-i\int_0^t H_c(t') dt'}]$, where $\mathcal{T}$ is the time-ordered operator. In the toggling frame the Hamiltonian is $H_{tog}(t)=U_{c}^\dagger H_0 U_{c}$. 
At $t=0$ the toggling frame coincides with the lab frame. If after the pulse sequence the toggling frame rotates back to the lab frame (as it is the case for all decoupling sequences), then the toggling frame propagator $U_{tog}(M\tau)= \mathcal{T}[e^{-i\int_0^{N\tau} H_{tog}(t') dt'}]$ coincides with the lab frame propagator $U(\{A_M\})$. 
Although the toggling frame Hamiltonian is still time-dependent, it does not contain strong pulses and can be  effectively approximated by a  time-independent local Hamiltonian (the average Hamiltonian) $H_A$ satisfying $U_{tog}(N\tau)=e^{-iH_A N\tau}$ \footnote{The Floquet-Magnus expansion does not converge in a many-body quantum system, and thus one has to truncate the series in Eq.~\ref{eq:magnus} and leave a small time-dependent and/or non-local Hamiltonian~\cite{Abanin15}. Effects of the truncation and the small time-dependent Hamiltonian are only evident at very long time scale, therefore ignored in this paper.}.
$H_A$ can be found perturbatively using the Floquet-Magnus expansion~\cite{Magnus54,Blanes09} 
\begin{equation}
\label{eq:magnus}
\begin{aligned}
    H_A=&\frac{1}{M\tau}\int_0^{M\tau} H_{tog}(t) dt\\ &-\frac{i}{2M\tau}\int_0^{M\tau} dt_1\int_0^{t_1}dt_2 [H_{tog}(t_1),H_{tog}(t_2)]\\
    &+O[(M\tau)^2],
\end{aligned}
\end{equation}
where the right-hand side of the first line is the zeroth-order average Hamiltonian as it scales as $(M\tau)^0$ and the second line is the first-order average Hamiltonian. If the average Hamiltonian is zero to certain order in 1D, it remains zero for higher dimensions. For example, the zeroth-order average Hamiltonian for the WAHUHA sequence is $2(D_x+D_y+D_z)=0$ regardless of dimensionality. The dimensionality can affect the fidelity by changing the magnitude of the leading nonzero higher-order Hamiltonian.

\section{RL pulse sequences}
\label{sec:seq}
We apply  DRL to different scenarios and generate various pulse sequences. Some representative ones are shown in Table~\ref{tab:seq}. We first tackle the case where we introduce only one imperfection at a time, and we later consider the case where several imperfections are present. 
\subsection{Single imperfections}
\begin{figure*}[!htbp]
\centering
\includegraphics[width=0.32\textwidth]{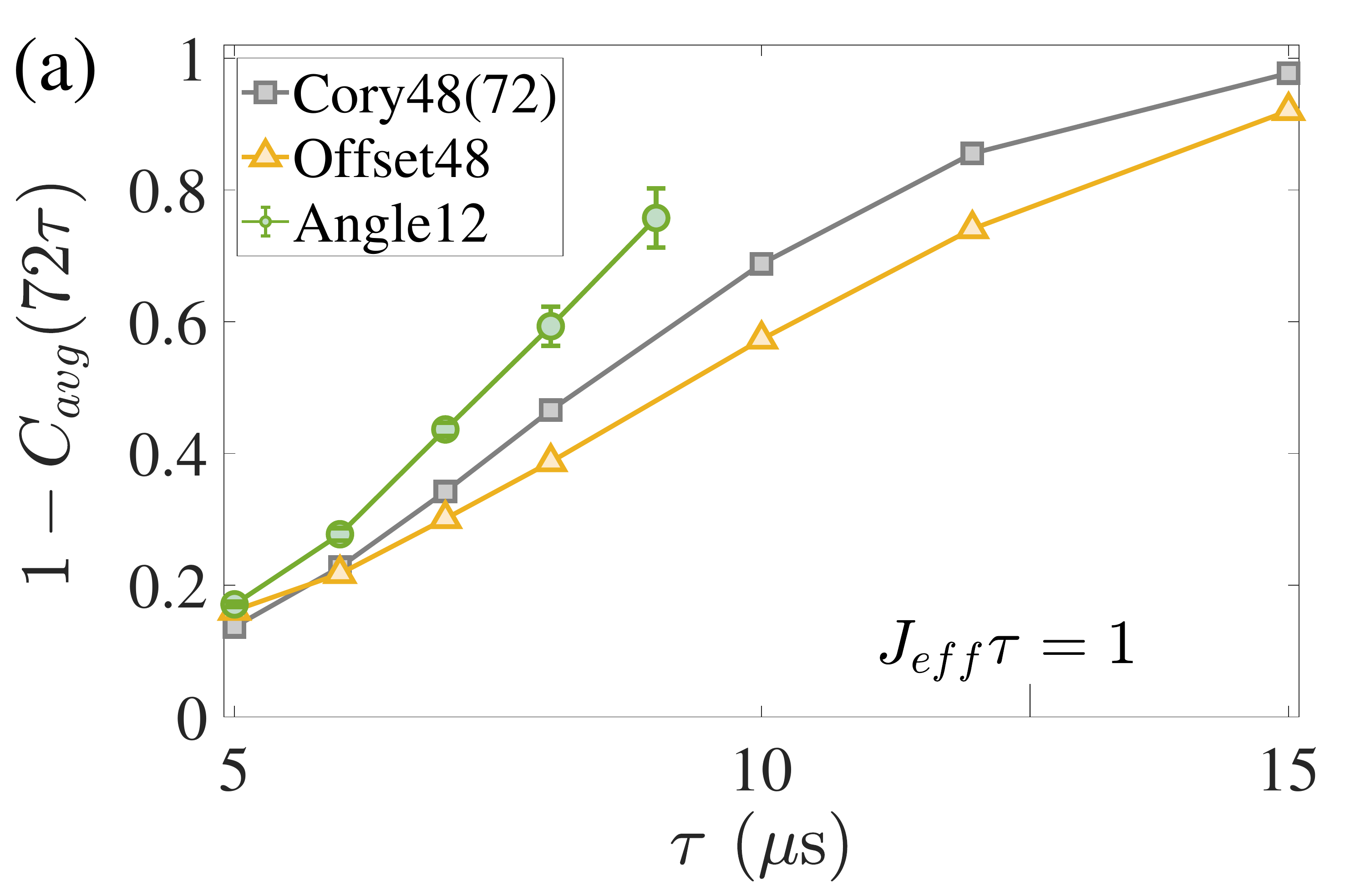}
\includegraphics[width=0.32\textwidth]{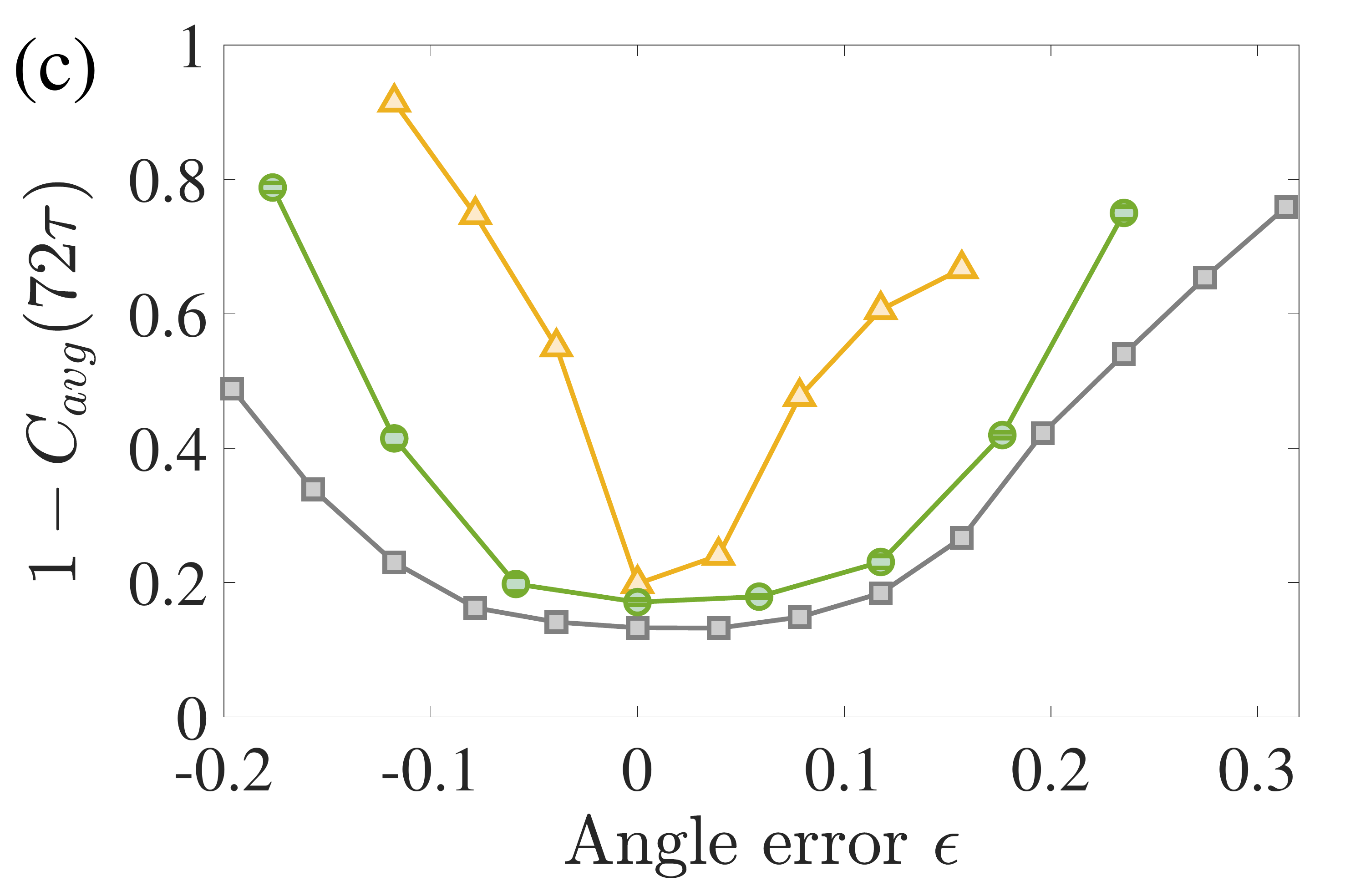}
\includegraphics[width=0.32\textwidth]{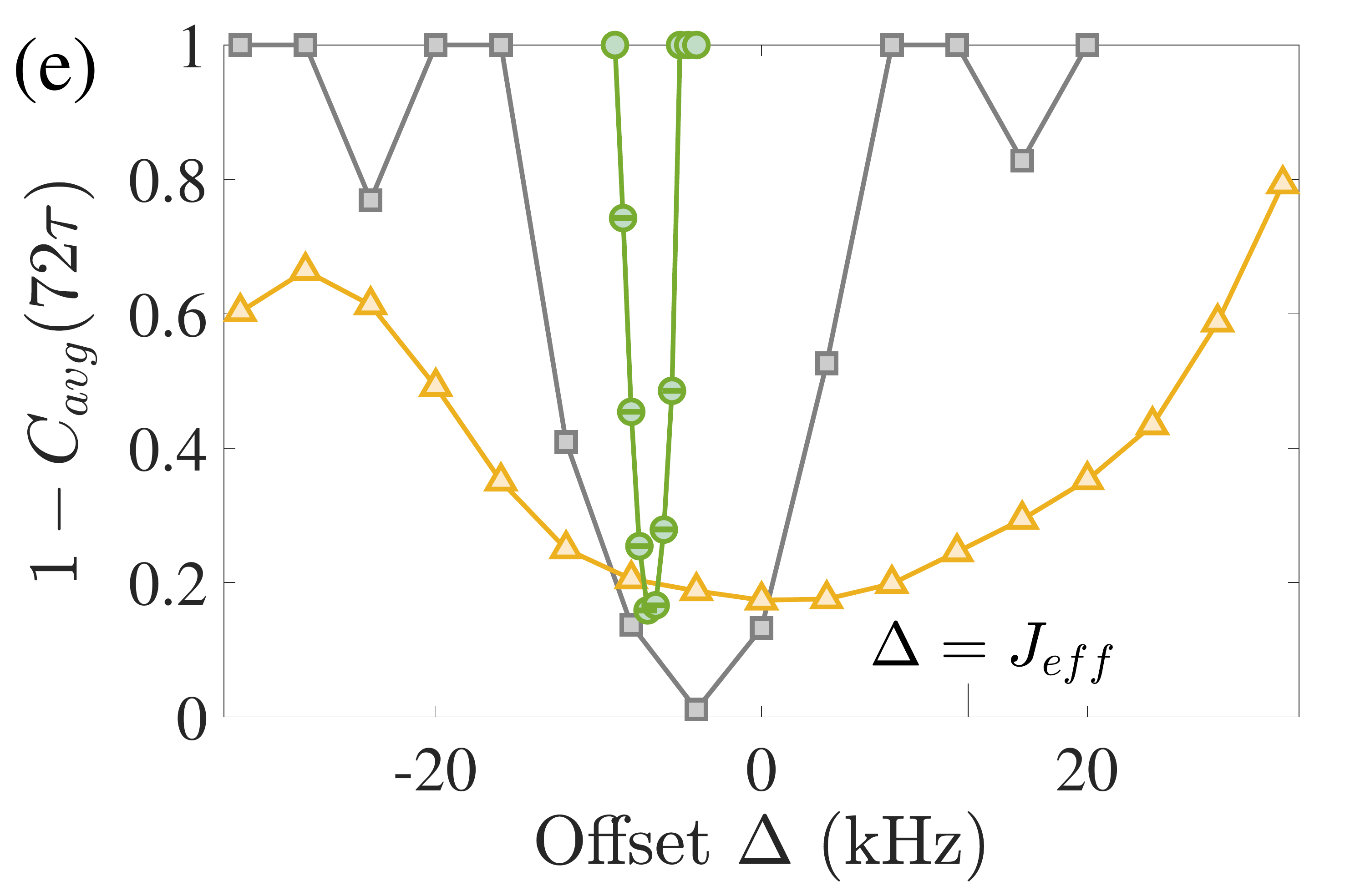}
\includegraphics[width=0.32\textwidth]{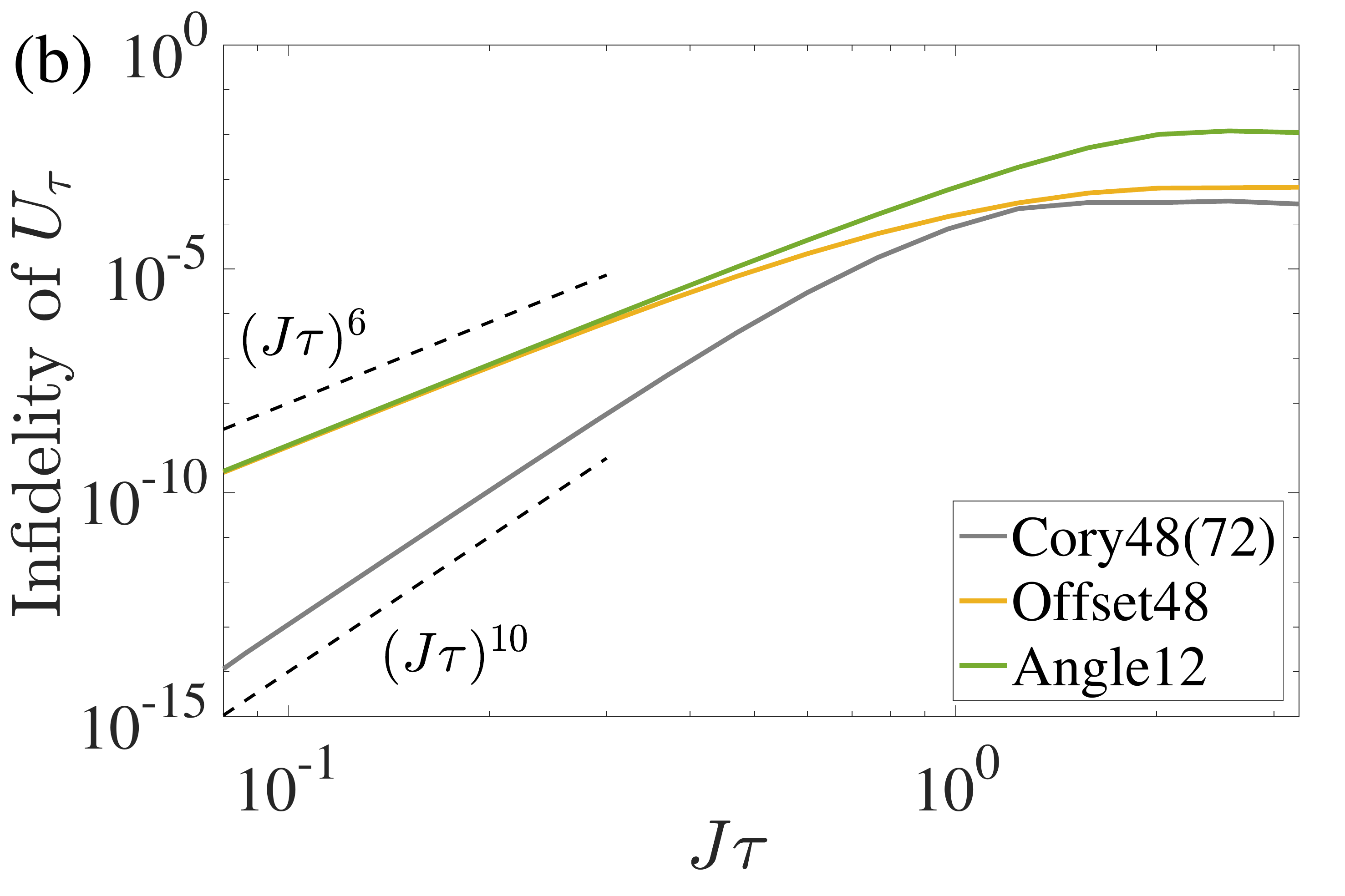}
\includegraphics[width=0.32\textwidth]{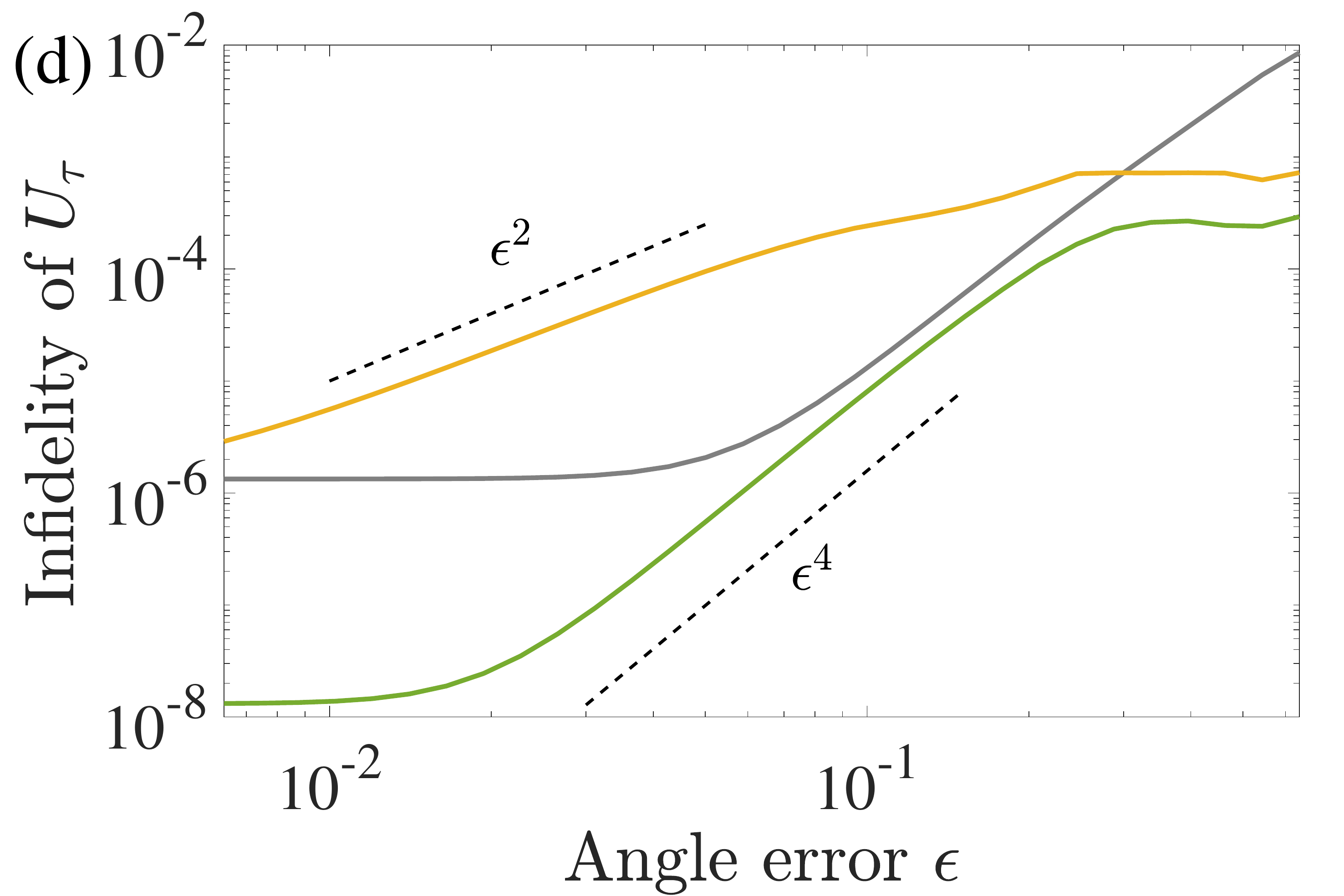}
\includegraphics[width=0.32\textwidth]{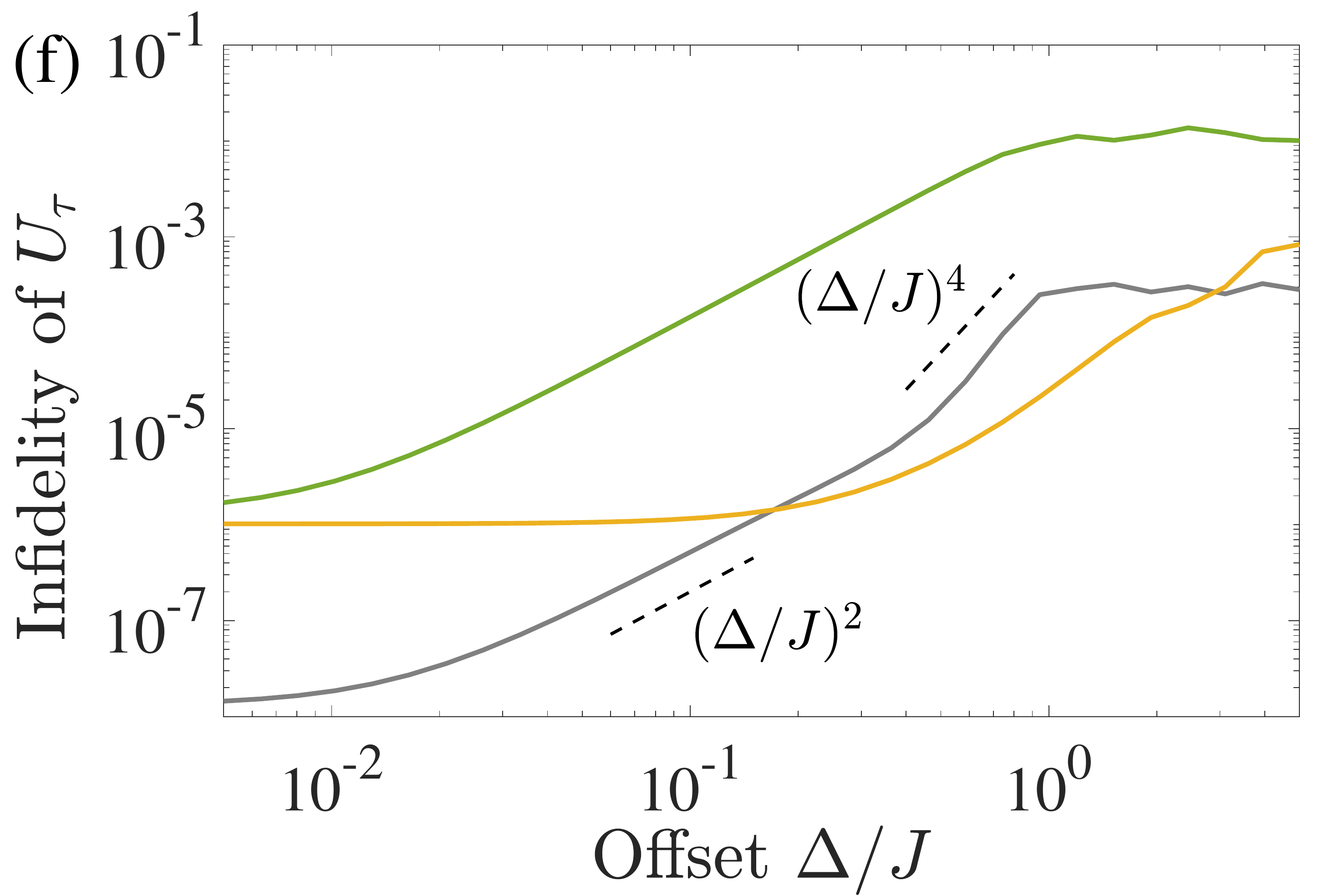}
\caption{\label{fig:RLseq}
Experimentally measured average correlation at $t=72\tau$ on CaF$_2$ (top panels) and numerically simulated propagator infidelity $1-F$ (bottom panels) of DRL sequences and Cory48 for different $\tau$ (a-b), angle error (c-d) and frequency offset (e-f). Dashed lines in (b) (d) and (f) show the scalings specified by nearby expression. We set $C_{avg}=0$ if any of $C_{XX}, C_{YY}, C_{ZZ}$ is smaller than zero. As 72 is not a multiple of 48, $C_{avg}(72\tau)$ of Offset48 is obtained as $[C_{avg}(48\tau)+C_{avg}(96\tau)]/2$. Imperfections are set to zero unless specified by the horizontal axis, with the exception of Angle12 experimental data in (a) and (c), which are taken at the optimal non-zero $\Delta$ due to the presence of phase transient (see \appref{sec:SMpt} for details). 
In (c) and (e), the pulse center-to-center delay is $\tau=5~\mu$s. 
$J_{\textrm{eff}}=79.7$~krad/s for the orientation of our sample~\cite{SM-RL}.
Error bars of Cory48 and Offset48 experimental data are determined from the noise in the free induction decay, which is smaller than the marker size thus not shown. Angle12 has larger error bars in (a) and (c) due to the inaccuracy in finding the optimal $\Delta$. 
In simulations we assume the pulse width is infinitesimal. We use $J=32.7$~krad/s as in FAp, $N=8$, periodic boundary condition and assume nearest-neighbor interactions. In (d) and (f) $\tau=10~\mu$s.
}
\end{figure*}

Control sequences that tackle zero or single sources of imperfections, 
Ideal6, Offset48, Angle12 and PW12, are directly generated by RL without any human input, trained with no error, offset, angle error and finite pulse width, respectively. In the training process, we start with a small $M$ (short sequence) and increase $M$ until we find a high-reward sequence. 
As can be seen from Table~\ref{tab:seq}, all the good sequence lengths are multiple of 6~\cite{SM-RL}. This can be understood via AHT: to cancel the zeroth-order average interaction Hamiltonian and rotate back the toggling frame with the allowed operations, the sequence length must be a multiple of 6. Moreover, when the dipolar interaction dominates, the machine learns to cancel the zeroth-order interaction Hamiltonian as quick as possible, i.e. in Idea6, Angle12 and PW12 the toggling frame Hamiltonian averages to zero every $3\tau$. This coincides with the discovery in Ref.~\cite{Burum79}. The Offset48 sequence, on the other hand, is trained under strong offset, thus does not obey this rule.
We also notice that the no-pulse action is never chosen by DRL, in contrast to celebrated decoupling sequences such as Cory48 and WAHUHA. Although the no pulse action is useful for some applications that requires a long time window between pulses, such as pulsed gradient generation and stroboscopic detection, it is not advantageous for decoupling.
This can also be understood from AHT: it is advantageous to apply pulses as frequently as possible so that higher orders in the average Hamiltonian are suppressed. Previously several attempts were made following this intuition,   adding additional pulses in some of the $2\tau$ window in WAHUHA-like sequences, but they did not lead to better performance~\cite{Burum79,Burum81}. Here we find RL discovers a completely different pattern that applies pulses as frequently as possible, and outperforms celebrated sequences (see next subsection)~\cite{Mansfield71,Haeberlen76}. 
In contrast, one of the most common strategies in conventional sequence design is to first come up with a sequence whose zeroth-order average Hamiltonian is the target Hamiltonian, and then symmetrize the sequence to cancel all odd orders. 
Symmetrization is achieved by following the original sequence by  the same sequence  but in reversed order, and with a $\pi$ phase shift. 
For example, the sequence $x,y,-x,-y$ is symmetrized to $x,y,-x,-y,y,x,-y,-x$. 
Strikingly, many DRL sequences are not symmetric, e.g., Ideal6, Offset48. This suggests that symmetrization is not optimal in some scenarios,  also noted in \cite{Burum79}. 

Figure~\ref{fig:RLseq} shows the experimentally measured average correlation at $t=72\tau$ and numerically simulated infidelity $1-F$ of Angle12 and Offset48, in comparison with Cory48. Note that although Cory48 contains only 48 pulses, its length is $72\tau$ because it also contains 24 no-pulse actions. Therefore, we explicitly denote it as Cory48(72). The experiments are done with CaF$_2$. As our experimental apparatus is not ideal and does not allow varying pulse width over a large range, we cannot provide experimental tests of Ideal6 and PW12, but we provide numerical results in~\cite{SM-RL}. 
$C_{avg}(72\tau)$ of Offset48 is not directly measurable because 72 is not a multiple of 48. Instead, we plot the average of $C_{avg}(48\tau)$ and $C_{avg}(96\tau)$, which is a good approximation of $C_{avg}(72\tau)$ as shown in~\cite{SM-RL}. 
Fig.~\ref{fig:RLseq}(b) shows that the fidelity of both Angle12 and Offset48 have a worse scaling compared with Cory48, because Cory48 cancels the average interaction Hamiltonian to higher order. However, this higher-order effect is not evident in experiments [Fig.~\ref{fig:RLseq}(a)] due to experimental imperfections dominating. 
Although Angle12 is 6 times shorter than Cory48, it shows the same scaling with angle error [Fig.~\ref{fig:RLseq}(d)] and similar robustness in experiments [Fig.~\ref{fig:RLseq}(c)]. As for the offset, the 
scaling of Offset48 is the same as Cory48 in the small $\Delta$ region [Fig.~\ref{fig:RLseq}(f)]. However, when the offset becomes larger, Offset48 outperforms Cory48, as shown both experimentally and numerically [Fig.~\ref{fig:RLseq}(e) and (f)]. This phenomenon is beyond AHT and intrinsically nonperturbative. Not surprisingly, Angle12 is not robust against offset, nor is Offset48 against angle error.

\subsection{Multiple imperfections and \textit{yxx} pattern}
\begin{figure}[!htbp]
\centering
\includegraphics[width=0.98\columnwidth]{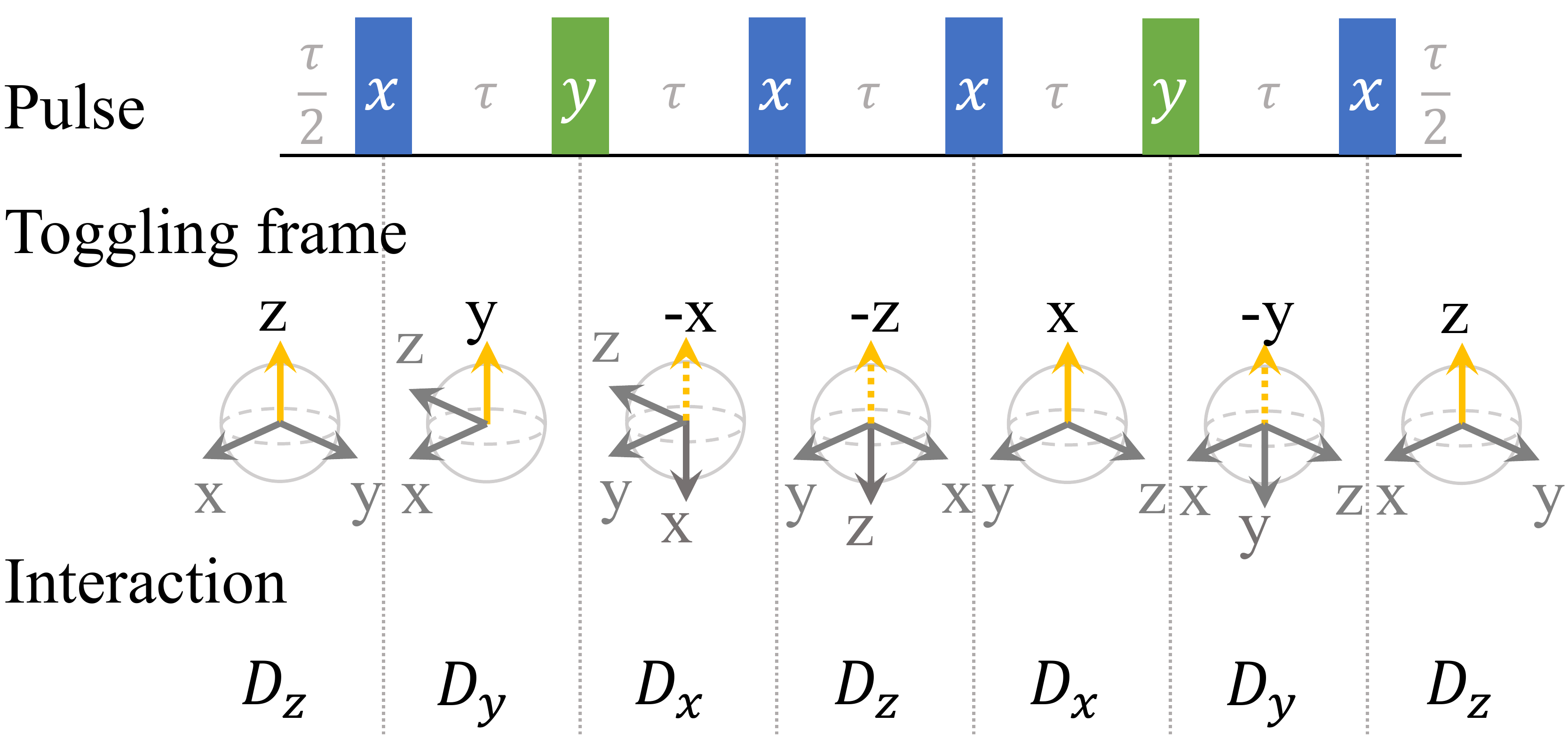}
\caption{
\label{fig:yxx}
yxx-type sequence. Top:  the pulse sequence. Middle: toggling frame transformation at each time. Arrows mark the orientation of toggling frame axis in the lab frame, where the yellow arrow highlights the axis overlapping with lab frame z-axis. Bottom:   the dipolar interaction in the  toggling frame at each time.
}
\end{figure}

\begin{figure*}[!htbp]
\centering
\includegraphics[width=0.32\textwidth]{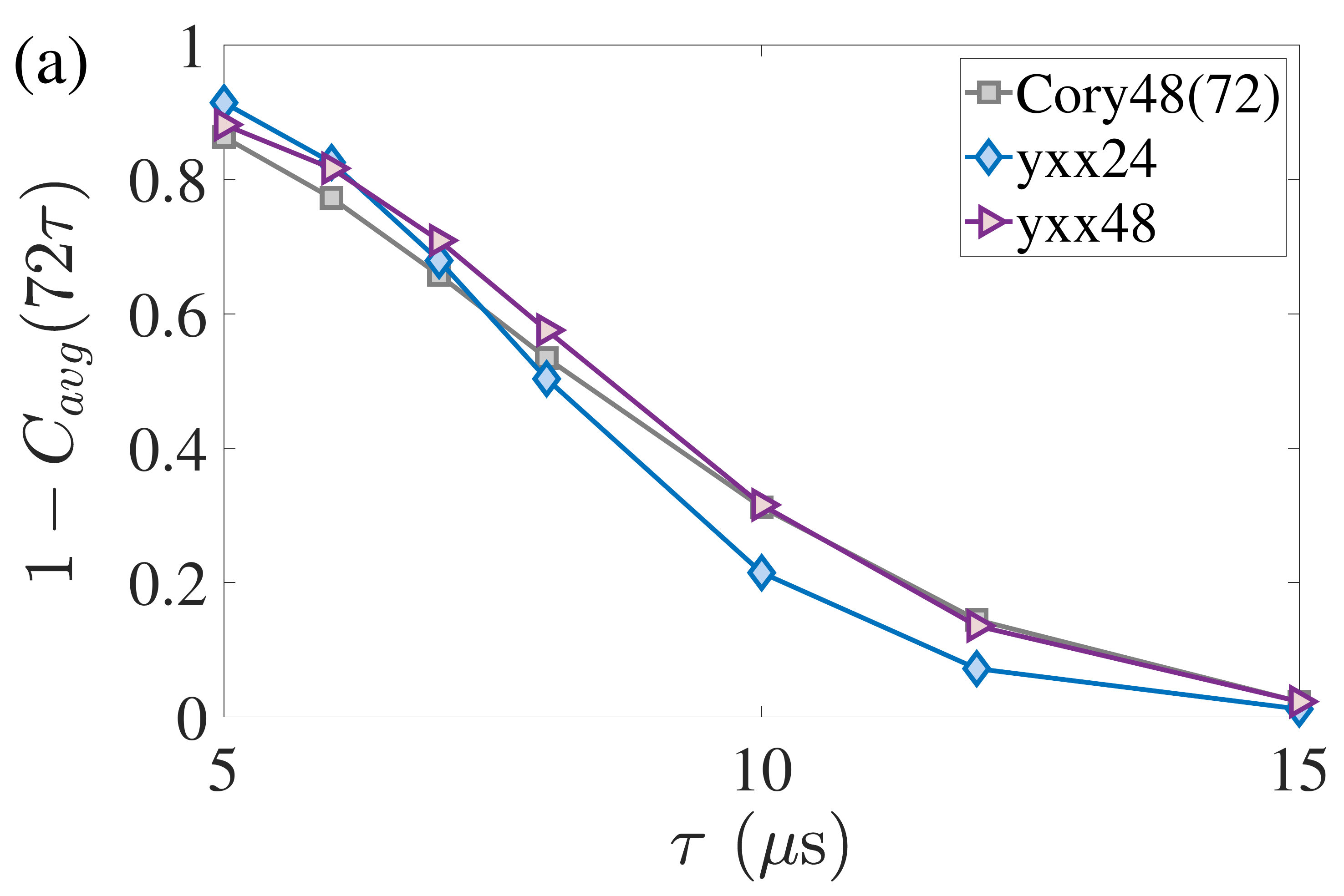}
\includegraphics[width=0.32\textwidth]{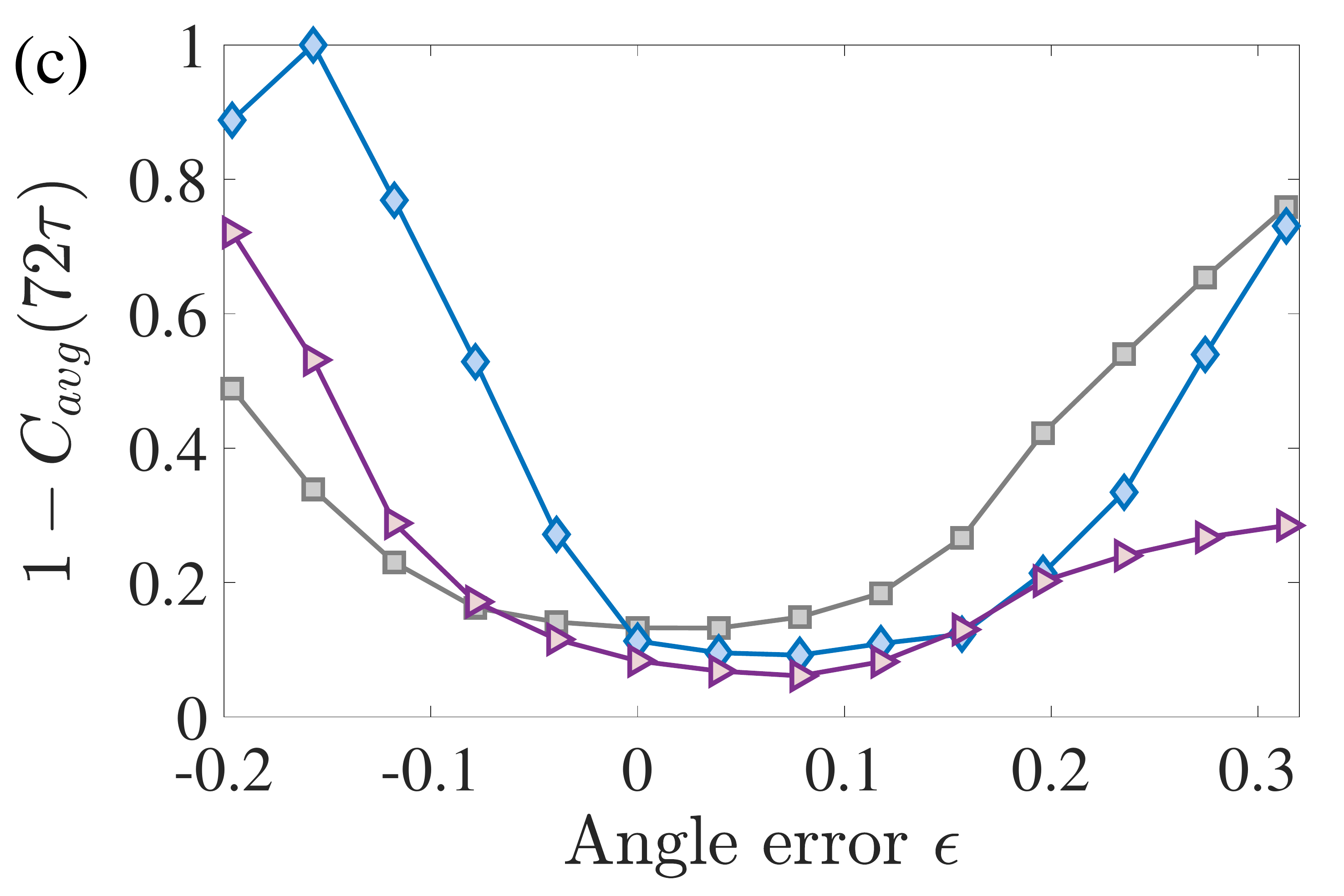}
\includegraphics[width=0.32\textwidth]{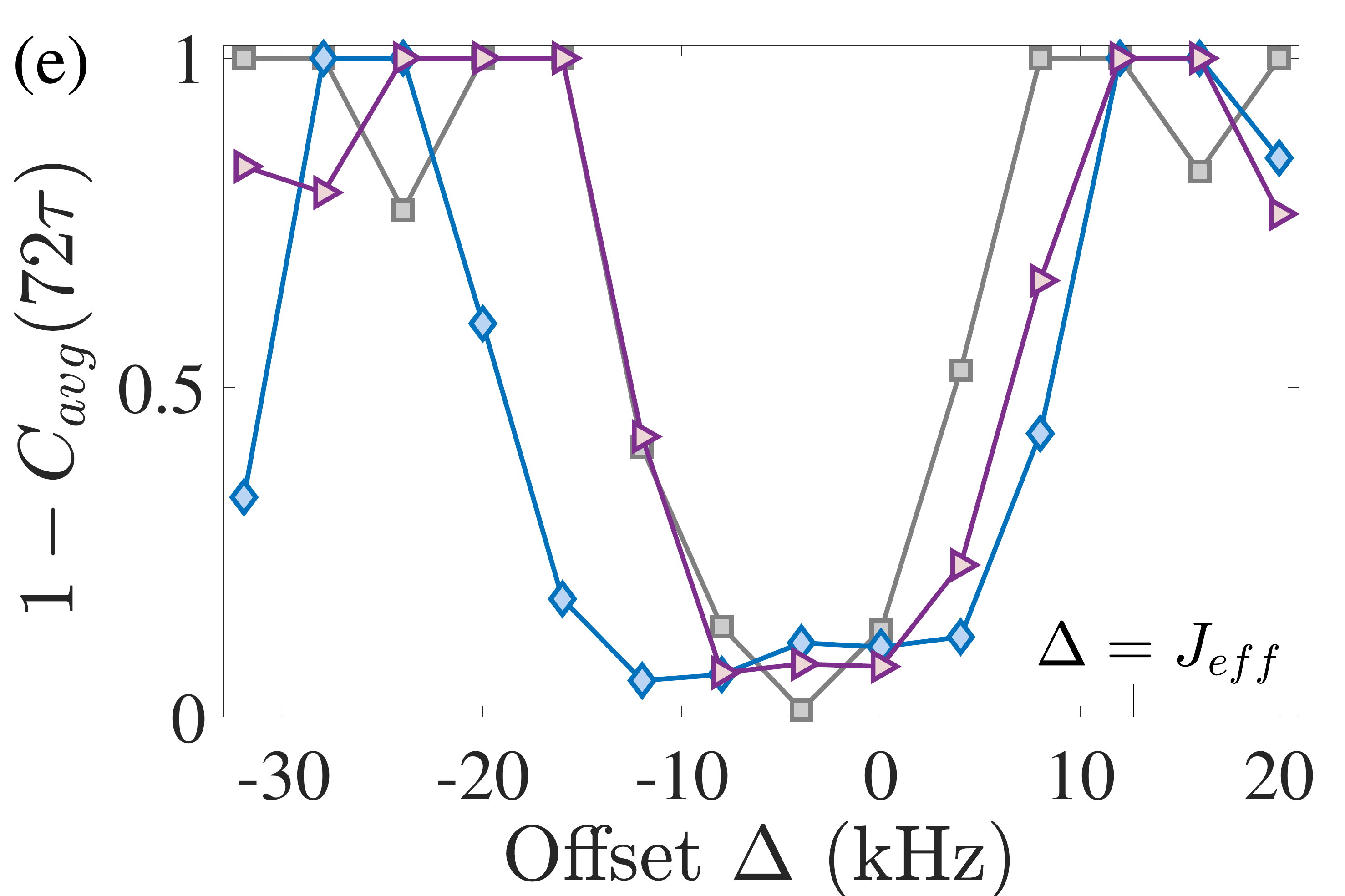}
\vskip 0.2cm
\includegraphics[width=0.32\textwidth]{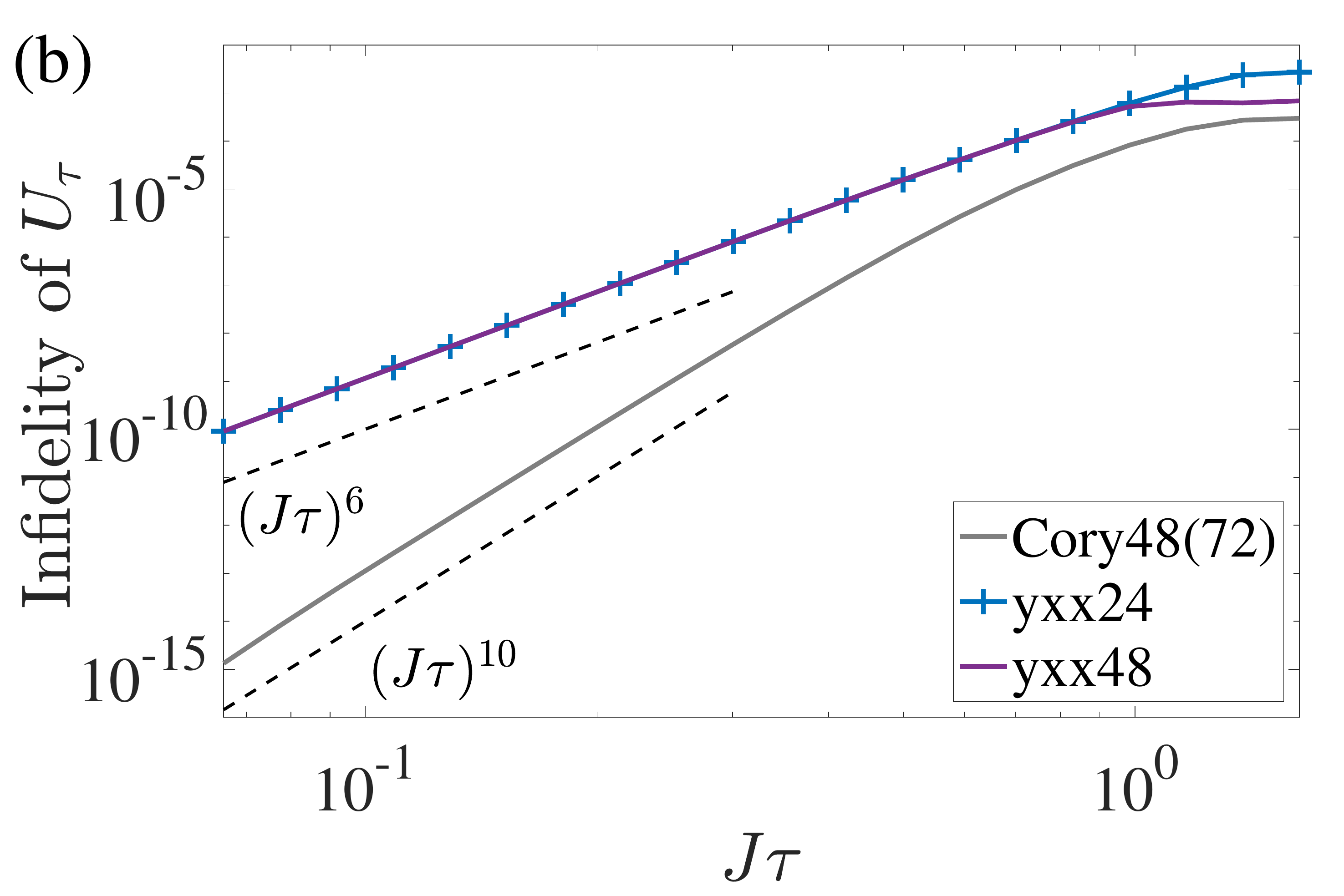}
\includegraphics[width=0.32\textwidth]{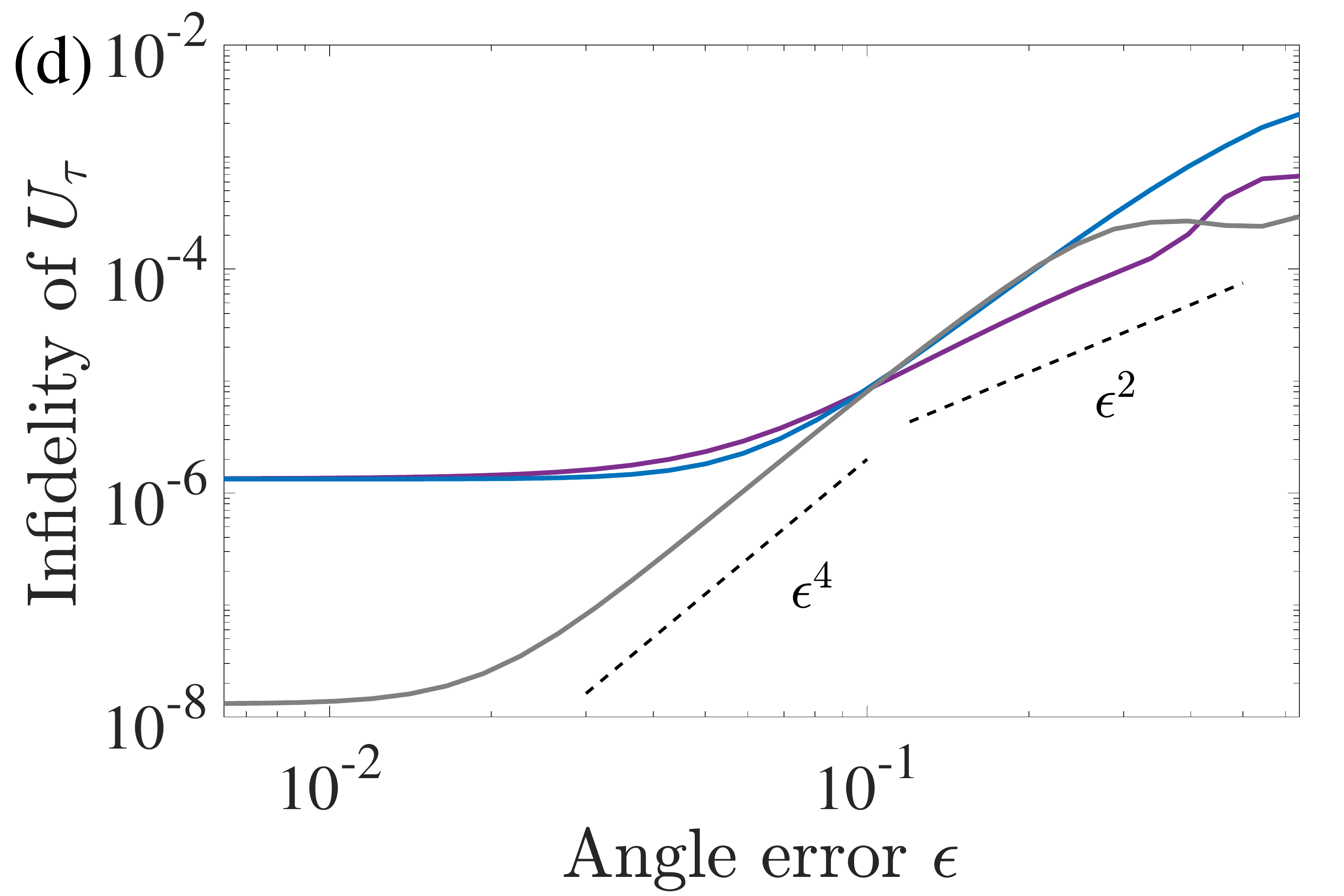}
\includegraphics[width=0.32\textwidth]{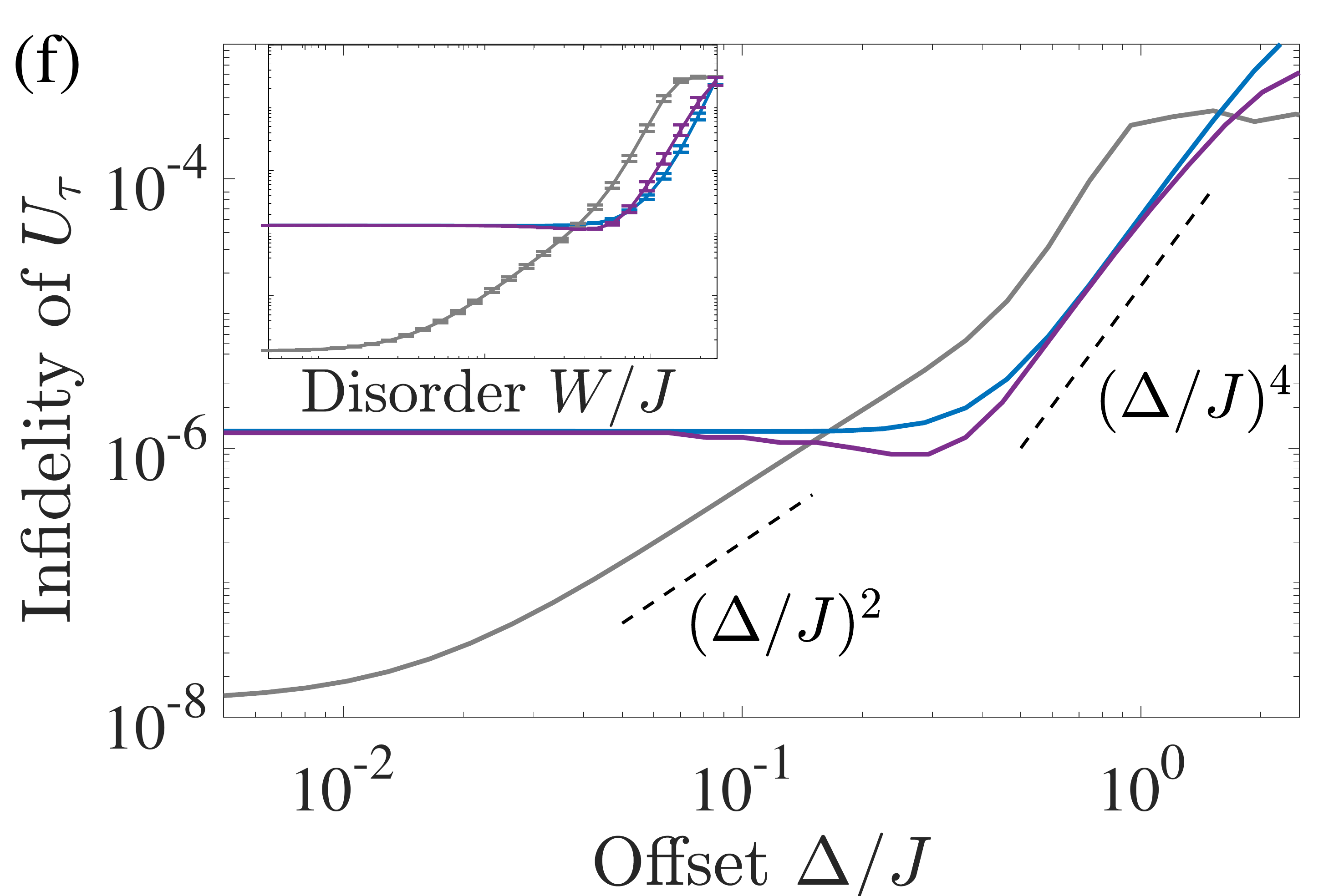}
\caption{\label{fig:yxxseq}
Experimentally measured average correlation at $t=72\tau$ on CaF$_2$ (top panels) and numerically simulated propagator infidelity $1-F$ (bottom panels) of DRL sequences and Cory48 for different $\tau$ (a-b), angle error (c-d) and frequency offset (e-f). 
The inset of (f) shows the propagator infidelity versus disorder strength $W$ averaged over 20 disorder realizations with error bars showing one standard deviation of the average infidelity. The x and y axes ranges of the inset are the same as in the main plot. 
Dashed lines in (b) (d) and (f) show the scalings specified by nearby expression. 
As 72 is not a multiple of 48, $C_{avg}(72\tau)$ of yxx48 is obtained as $[C_{avg}(48\tau)+C_{avg}(96\tau)]/2$.
Imperfections are set to zero unless specified by the horizontal axis. 
Error bars of the experimental data are determined from the noise in the free induction decay which is smaller than the marker size thus not shown. Other parameters are the same as in Fig.~\ref{fig:RLseq}.
}
\end{figure*}

\begin{figure}[!htbp]
\centering
\includegraphics[width=0.98\columnwidth]{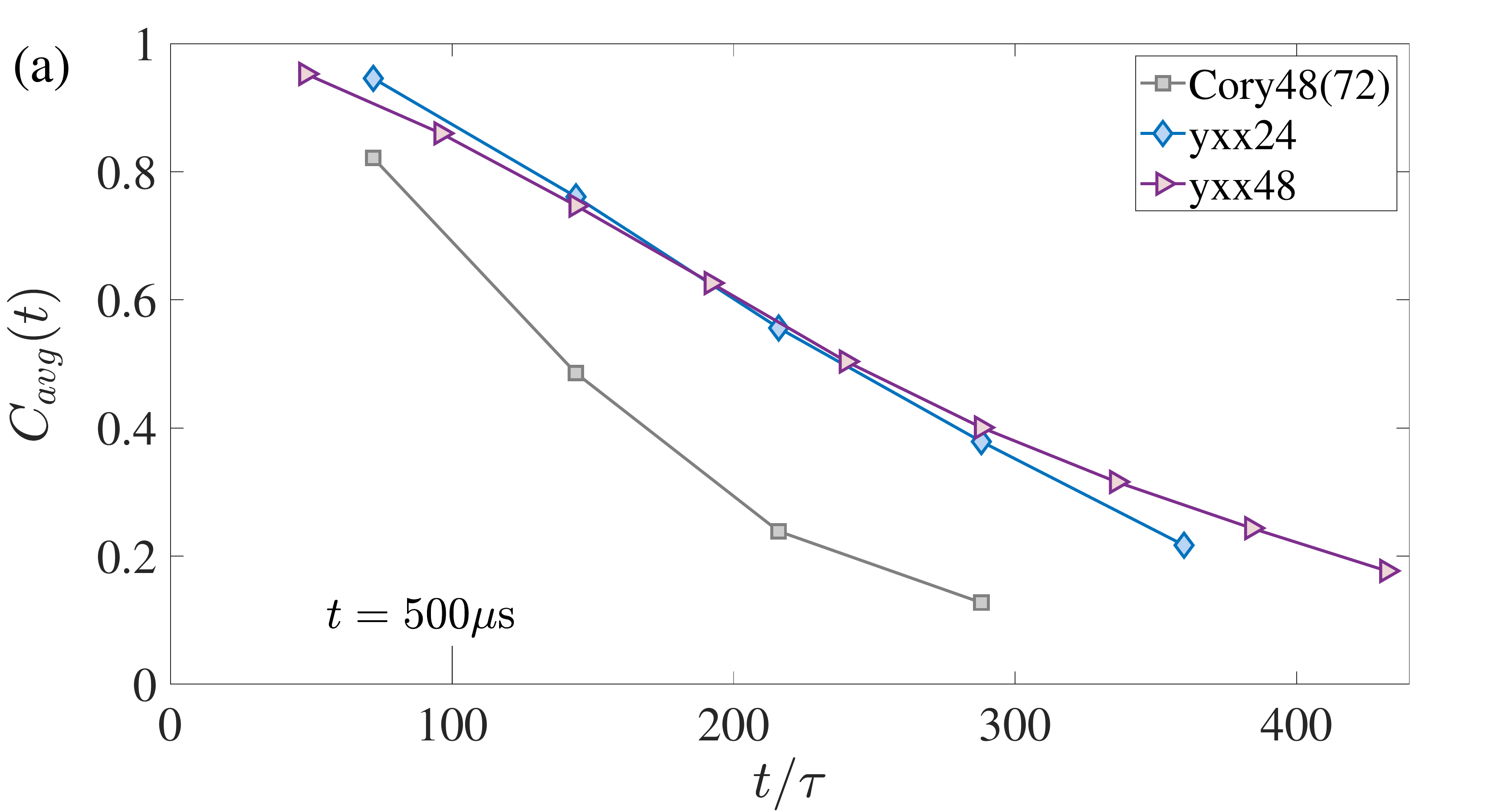}
\includegraphics[width=0.98\columnwidth]{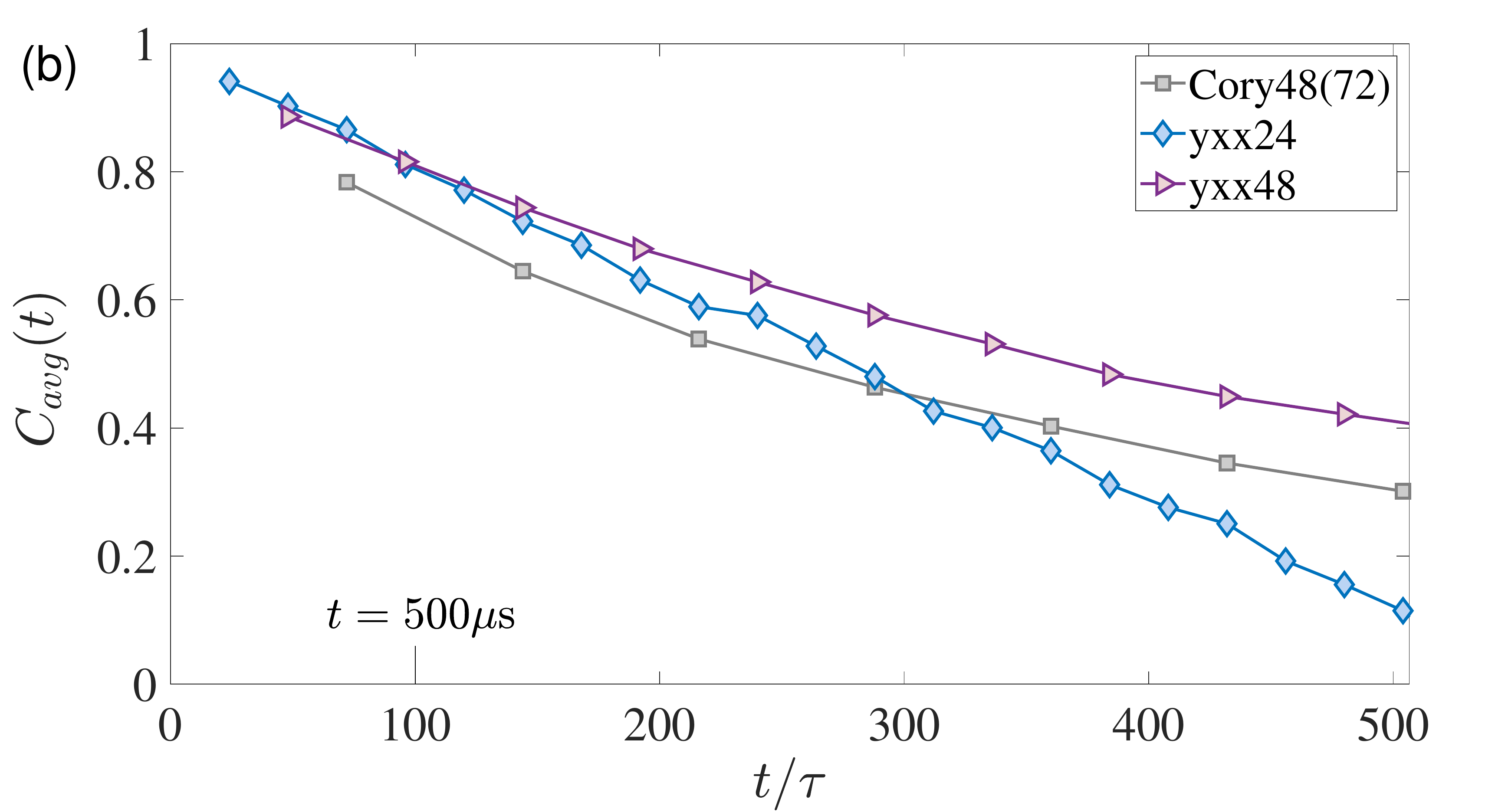}
\caption{\label{fig:yxx_t}
Experimentally measured average correlation at the best calibrated condition of CaF$_2$ (a) and FAp (b) as a function of time. Error bars of the experimental data are determined from the noise in the free induction decay which is smaller than the marker size thus not shown. Other parameters are the same as in Fig.~\ref{fig:RLseq}.
}
\end{figure}

DRL is successful in learning good pulse sequences in the presence of a single imperfection. When two or more imperfections exist simultaneously, the number of satisfactory sequences is significantly reduced and the DRL is unable to find one within reasonable time. Although this may be solved by using more powerful computers or more sophisticated algorithms, we take a physicist's approach. We learn from  sequences DRL found in the presence of a single imperfection and use our understanding to design more powerful sequences. We notice that Ideal6, Angle12 and PW12 are all built from the 3-pulse block $(\pm y \pm x \pm x)$ or its equivalent form $(\pm x \pm y \pm y)$, which we refer to as ``\textit{yxx} pattern''. 

Figure~\ref{fig:yxx} shows the toggling frame configuration and toggling frame Hamiltonian for two consecutive $yxx$ blocks. The last pulse is rotated to the first position for easier analysis, and we note that rotation of pulses unitarily changes the Floquet propagator thus does not change the fidelity~\cite{Bukov15}. We denote the dipolar interaction along the $\alpha$ axis as $D_\alpha$ as defined in Eq.~\ref{eqn:h0}
with $\alpha=x,y,z$. Because $D_x+D_y+D_z=0$, in the ideal case the $xyx$ block cancels the zeroth-order Hamiltonian and is the shortest sequence to do so. Although the length of a solid echo is only $2\tau$~\cite{Powles62}, it does not qualify as a decoupling sequences as defined here, because the average Hamiltonian is $D_y+D_z=-D_x$ and thus the sequence only protects the $X$ state. The shortest known decoupling sequence is WAHUHA, whose length is $6\tau$ though only contains 4 pulses ~\cite{Waugh68}.
The first-order average Hamiltonian of the first $xyx$ block is $-i[D_x,D_y]\tau/6$, which is cancelled by the contribution from the second block. Changing the signs of the pulses will not change the toggling frame interaction Hamiltonian, as the dipolar interaction is invariant under $\pi$ rotation. Therefore, the \textit{yxx} pattern guarantees vanishing zeroth- and first-order average interaction Hamiltonian. 

With this understanding, we adopt two approaches to construct longer sequences that are robust against multiple imperfections. First we can restrict our search to sequences with \textit{yxx} patterns only, so the agent only needs to choose the plus or minus sign instead of five actions. This significantly reduces the dimension of the search space from $5^N$ to $2^N$. In this way we find the yxx48 sequence shown in Table~\ref{tab:seq}. 
A second approach is to directly modify the RL sequences found above in order to cancel the additional imperfections. 
As Angle12 is robust against both angle error and finite pulse width, we double and modify it into the 24 pulse sequence shown as yxx24 in Table~\ref{tab:seq} so that it is also robust against offset (See appendix~\ref{appsec:yxx24intuition} for further details.) 

The performance of these two \textit{yxx} sequences are shown in Fig.~\ref{fig:yxxseq}. Again, the two \textit{yxx} sequences have a worse scaling with $\tau$ than Cory48, but this effect is barely seen in experiments, Fig.~\ref{fig:yxxseq}~(a). Experimentally, the three sequences are all robust against angle error [Fig.~\ref{fig:yxxseq}(c)]. If we only consider angle error, the average Hamiltonian of yxx48 is zero up to first-order; the average Hamiltonians of yxx24 and Cory48 are zero up to second order, since  they all cancel the angle error to first-order, as indicated by AHT. The scaling shown in Fig.~\ref{fig:yxxseq}(d) is the result of the cross commutator between interaction and angle errors from second and higher order average Hamiltonian. For the offset, both \textit{yxx} sequences show a better scaling compared to Cory48. In experiments we observe a plateau at small offset for the \textit{yxx} sequences, but not for Cory48, in agreement with the simulation. The fact that the yxx24 plateau  is wider than that of yxx48 might be a result of unknown experimental imperfections that correlate with frequency offset. 
As the frequency offset has the same form as on-site disorder, we expect any pulse sequence to show similar robustness against the two imperfections. This is confirmed by comparing Fig.~\ref{fig:yxxseq}(f) and its inset (for the \textit{yxx} sequence, disorder  and offset effects are equivalent up to first-order AHT \appref{app:offdis}).

Finally, we compare \textit{yxx} sequences and Cory48 under the best achievable experimental conditions with the two experimental samples in the same apparatus. Results on a different sample and different apparatus can be found in Appendix~\ref{appsec:adamantane}. Data taken with the disorder-free CaF$_2$ sample are in Fig.~\ref{fig:yxx_t}(a). Both yxx24 and yxx48 protect the correlation to significantly longer time than Cory48. In the disordered FAp sample, yxx48 still shows a better performance compared to Cory48, while yxx24 wins for $t<288\tau$ [Figure~\ref{fig:yxx_t}(b)]. This is not surprising because Angle12 is learned at effective $t=\tau$ and yxx24 is built on Angle12. The faster decay of yxx24 coherence at longer time is caused by an unknown field along the z axis, giving rise to  decaying oscillations of $C_{XX}$ and $C_{YY}$, which decrease faster than the exponential decay  of yxx48 and Cory48
~\footnote{We do not know where the z field comes from exactly, because yxx24 has zero average Hamiltonian up to second order for all the experimental imperfections we know and the third order is too complicated to track. But since we observe this field only in the disorder sample FAp, we suspect it is due to high order ($>2$) average Hamiltonian involving even number of disorder Hamiltonians such as $\sum_i [\epsilon \sigma_x^i,[w_i \sigma_x^i,w_i \sigma_z^i]]$ where $w_i^2$ does not average to zero and thus appears as a field. }.

\section{Conclusion and outlook}
\label{sec:conclusion}
We designed robust decoupling sequences using DRL and experimentally demonstrated that they  lead to better performance than the best-known sequence. 
We directly consider $\pi/2$ pulses as actions, enabling discovery of long sequences, and we use a gradient-free optimization method together with deep neural networks to tackle the complex control  landscape. DRL without any human insight is capable of dealing with single imperfections. Surprisingly,  many of the DRL sequences are not symmetric, instead, we observe a \textit{yxx} pattern. Building on our understanding of the \textit{yxx} pattern, we then find sequences that are robust against all dominant imperfections present in our experiments, leading to a better performance than the celebrated Cory48 sequence in two different samples. 
Our work emphasizes the usefulness of both artificial intelligence and human knowledge of the physical system in realistic applications.

We conclude this paper by pointing to some future research directions. (I) Although we focused on decoupling interacting spin-1/2 systems, a task that has applications in building spin-based quantum memories, our method is completely general to other systems and applications, by simply modifying the reward function to engineer the desired Hamiltonian. 
It would be interesting to apply this method to quantum simulation or quantum sensing. 
In this work we train the machine learner in the context of solid-state NMR, where the pulsed controls have been developed and optimized for 50 years, yet RL still shows an advantage. We expect our methods  might yield even more significant improvements in other quantum platforms whose controls are developed more recently, such as color centers in solids, cold atoms, trapped ions, and superconducting circuits.
(II) The DRL training in this work was simply performed on a personal laptop, so there is still large space for improvement on the computational side, e.g. by using a supercomputer with GPU acceleration to tackle more complex control sequences. 
(III) Further improvements could be obtained by a stronger interface between machine learner and the physical system. Here we trained the DRL purely using a classical computer and tested the learned sequences on a quantum simulator. Our method can be readily modified into a hybrid classical-quantum DRL process: the DRL agents on a classical computer generate a sequence, which is then applied in a quantum system; then one use an experimental observable, such as the correlation decay rate, as the reward to train the agents. In this way the the modeling of system Hamiltonian and control imperfections is not required. While in our current learning process simulating the spin-chain environment only takes a small portion of the total CPU time, this could change for different tasks that require simulating a many-body non-integrable system. Then, we expect replacing the classical simulation with quantum experiments will improve the training time and open new avenues for devising quantum control protocols. 

\begin{acknowledgments}
Authors would like to thank H. Zhou and L. Viola for discussion. This work was supported in part by the National Science Foundation under Grants No. PHY1734011, No. PHY1915218, and No. OIA-1921199.
\end{acknowledgments}

\appendix

\section{Phase transient effects on Angle12}
\label{sec:SMpt}
\begin{figure}[!htbp]
\centering
\includegraphics[width=\columnwidth]{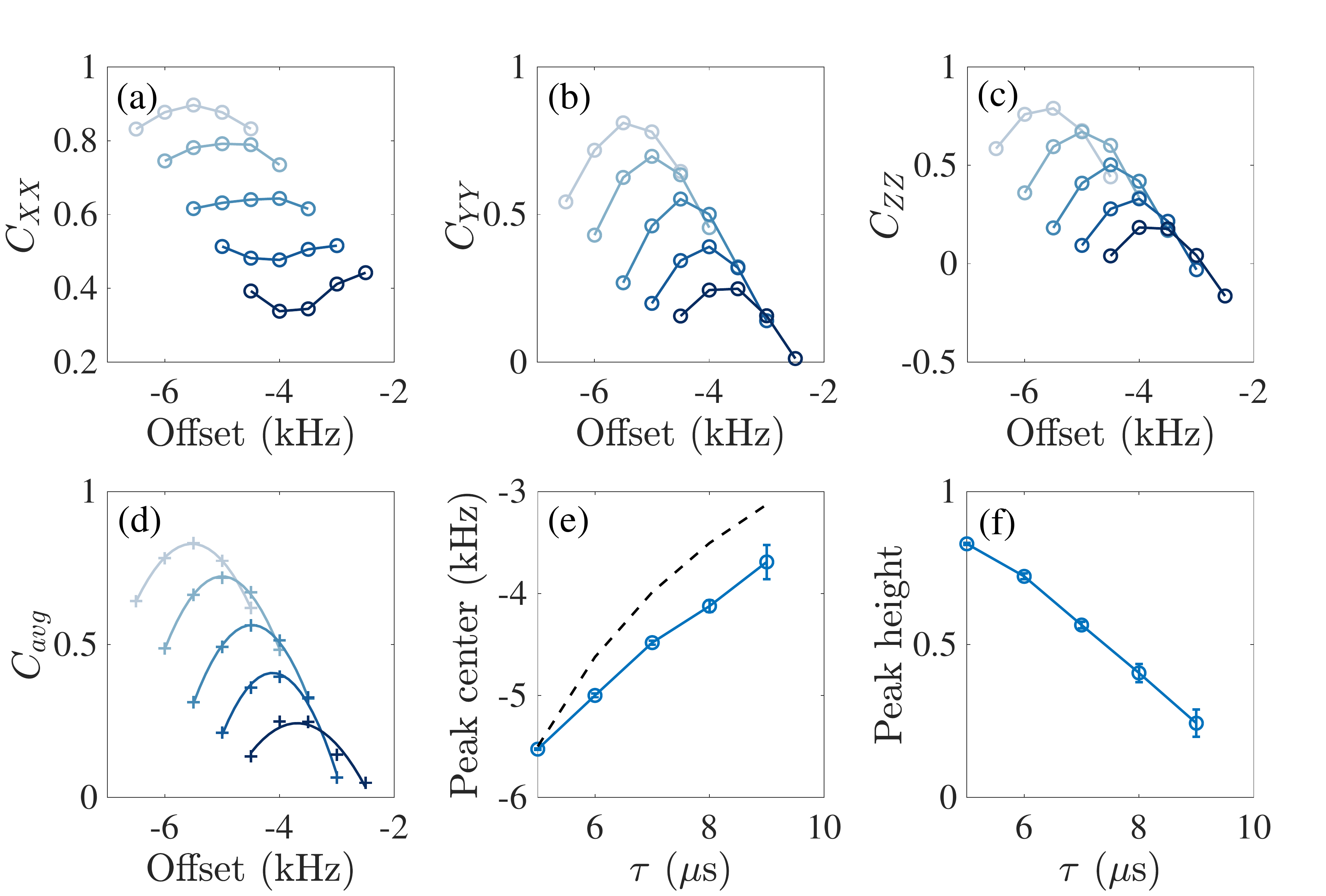}
\caption{\label{fig:pt_tau}
Experimental study of angle-12 for different offset and $\tau$. (a-d) show $C_{XX}$, $C_{YY}$, $C_{ZZ}$ and the average correlation respectively. Different curves are obtained with $\tau$ from 5~$\mu$s to 9~$\mu$s, with a step of 1~$\mu$s and lighter color representing smaller $\tau$. In (d), the plus sign marks the experimental data point and the curve shows the parabolic fitting, whose peak center and height are shown as the blue curve in (e) and (f), respectively. The length of the error bars corresponds to two standard deviation of the fitted results. In (e), the black dashed line shows the peak center expected from first-order AHT.
}
\end{figure}

\begin{figure}[!htbp]
\centering
\includegraphics[width=\columnwidth]{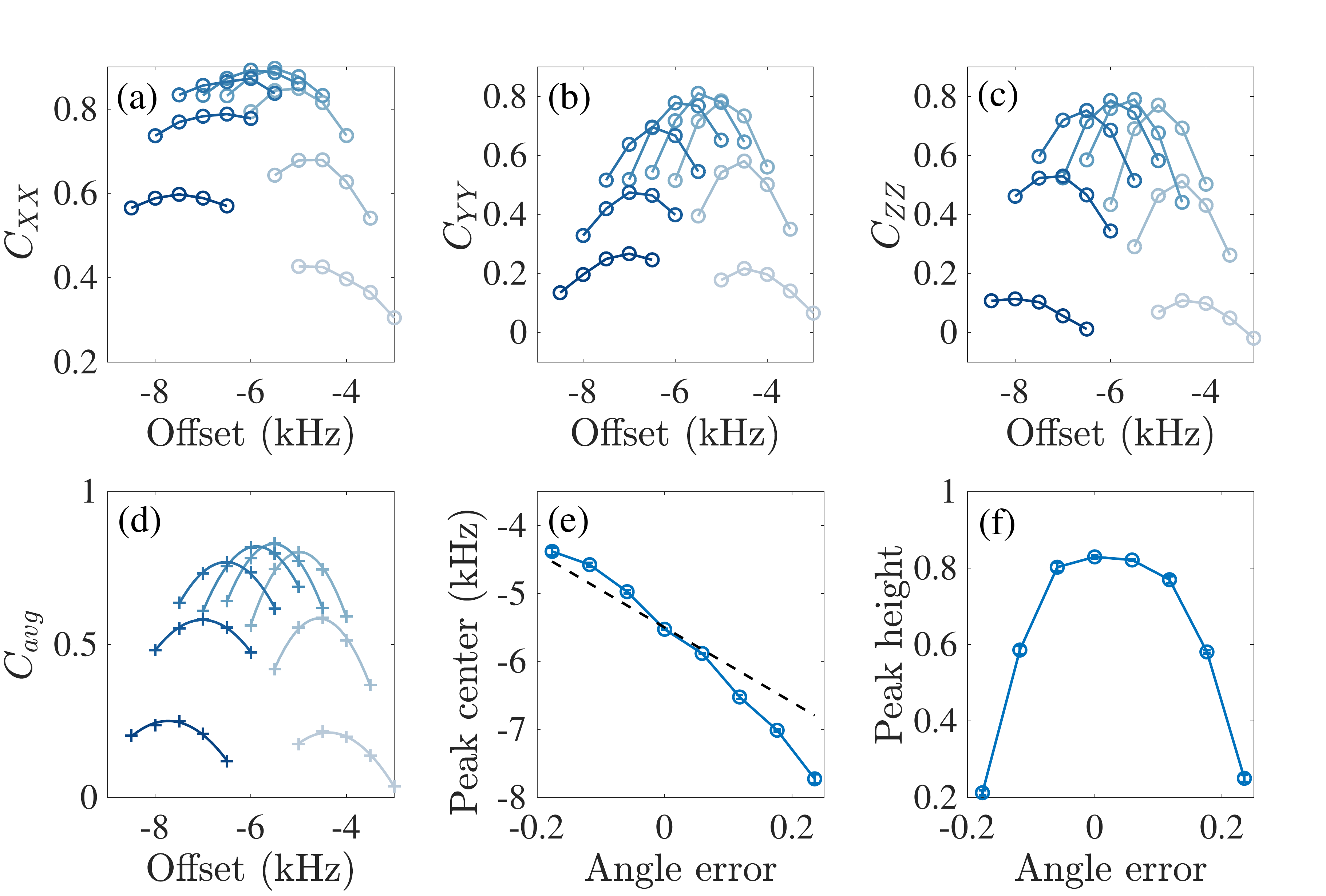}
\caption{\label{fig:pt_ae}
Experimental study of angle-12 for different offset and angle error. (a-d) show $C_{XX}$, $C_{YY}$, $C_{ZZ}$ and the average correlation respectively. Different curves are obtained with angle error from -0.18 to 0.24 with a step of 0.06, with lighter color representing smaller angle error. In (d), the plus sign marks the experimental data point and the curve shows the parabolic fitting, whose peak center and height are shown as the blue curve in (e) and (f), respectively. The length of the error bars corresponds to two standard deviation of the fitted results. In (e), the black dashed line shows the peak center expected from first-order AHT.
}
\end{figure}

The effects of pulse phase transients are typically difficult to quantify, as they introduce difficult to characterize time-dependent Hamiltonian terms. Still, here we show it is still possible to capture their essence using a simple model. In the future, we could even include phase transients into the reward function to design sequences that are robust against them. 

In  Fig.~\ref{fig:RLseq}(e) of the main text, we see that the optimal performance point of the  Angle12 sequence deviates from $\Delta=0$ by a significant amount. Here we show that this is due to the cancellation of phase transient error and offset in Angle12. Since we do not include the phase transient effect as an error source during the training process, we should also minimize this effect in the experimental testing. This can be realized by pinning the offset to $\Delta_0$ for Angle12. Other sequences happen to be sufficiently robust to phase transient that they do not require any special treatment.

We first explain the physics of phase transients. When creating a RF square pulse along the x-axis, the leading and trailing edges are not sharp and they inevitably generate a small y-component~\cite{Haeberlen76}. 
Although the exact description of phase transient is not know, the simple model introduced in Ref.~\cite{Haeberlen76} can qualitatively explain experimental results. An x-pulse with a phase transient is modeled by a propagator $e^{-i\alpha_1 Y} e^{-i(\pi/2) X} e^{-i\alpha_2 Y}$, where $\alpha_1$ and $\alpha_2$ denote the strength of the y-component at the trailing and leading edge, respectively. For pulses along other axes, this model assumes that the additional component is always $+\pi/2$ phase shifted with respect to the main component. 
Using AHT, we find the zeroth-order average phase transient of Angle12 is $(\alpha_1-\alpha_2)(-4X+2Y)/(12\tau)$. As the zeroth-order offset of Angle12 is $\Delta(-4X+2Y)/12$,  the two cancel each other out at the optimal offset $\Delta_0=(\alpha_2-\alpha_1)/\tau$, leading to the non-zero optimum point in Fig.~\ref{fig:RLseq}(e). 

We verify this relation in Fig.~\ref{fig:pt_tau}, where we show the autocorrelations for different offsets and $\tau$. Figure~\ref{fig:pt_tau}(a-c) shows  $C_{XX}, C_{YY}$ and $C_{ZZ}$, where each curve is taken for a given $\tau$ and darker colors denote larger $\tau$. For each $\tau$, there is indeed a peak at $\Delta_0$. When the offset deviates from $\Delta_0$, we see the decrease of $C_{XX}$ is not as significant as that of $C_{YY}$ and $C_{ZZ}$, because this deviation results in an effective magnetic field  $\propto4X-2Y$, which is close to the x direction. When $\tau$ increases, the peak center $\Delta_0$ shifts toward smaller offset (in absolute value) as expected from the AHT analysis above. To quantitatively analyze this trend, we fit $C_{avg}$ at a fixed $\tau$ to a parabolic function, as shown in Fig.~\ref{fig:pt_tau}(d). 
In Fig.~\ref{fig:pt_tau}(e) we plot the peak center $\Delta_0$ as a function of $\tau$ and compare it with the zeroth-order AHT value $\Delta_0\propto 1/\tau$. The two quantities show the same trend, with an imperfect match  due to the simplicity of the model. 
Because of the variation of $\Delta_0$ with $\tau$, it is not reasonable to use the same offset when testing Angle12's robustness against $\tau$, instead, we use the fitted $C_{avg}$ peak height in Fig.~\ref{fig:RLseq}(a). We note that our analysis  does not artificially increase the robustness of Angle12 compared to the ideal case without phase transient. By choosing the optimal $\Delta_0$, we can at most cancel zeroth-order effects of the phase transient, while  higher order terms and cross terms between phase transient and other Hamiltonian components still lead to the  degradation of the autocorrelations. Therefore, we still underestimate the robustness of Angle12  even  when we are using the optimal $\Delta_0$.

When introducing the angle error, the optimal cancellation condition also changes. This can be seen from the first-order AHT. The first-order cross term between angle error and phase transient leads to an additional field $(-2X-2Y+4Z)\epsilon (\alpha_1-\alpha_2)/(12\tau)$; the cross term between angle error and offset leads to a field $(-2X+4Y-4Z)\epsilon\Delta/12$. %
In other words, introducing an angle error dresses the effective fields due to phase transient and offset, and now the two cannot exactly cancel each other. 
Still, there exists an optimal offset $\Delta_0$ where the residual field is the smallest. To first-order in $\epsilon$, $\Delta_0=(\alpha_2-\alpha_1)(1+\epsilon)/\tau$. We experimentally verify this relation in Fig.~\ref{fig:pt_ae}. 
In Fig.~\ref{fig:pt_ae} (a-d) the darker color denotes larger $\epsilon$ and we see that $\Delta_0$ decreases when increasing $\epsilon$ (note that $\Delta_0<0$) as expected from the analysis above. 
Again we fit $C_{avg}$ to a parabolic function to get the peak center and peak height. The peak center as a function of $\epsilon$ is shown in Fig.~\ref{fig:pt_ae}(e) with the dashed line denoting the theoretical value assuming the $\Delta_0$ at $\epsilon=0$ is exact. Our experimental results do agree with the theoretical expectations. The peak height is shown in Fig.~\ref{fig:pt_ae}(f) and also Fig.~\ref{fig:RLseq}(c).

In addition to the two verifications above, we also increase the pulse width and observe $\Delta_0$ decreases (not shown). This is because the phase transient becomes less evident when  using a lower RF power. 
\section{Similarity between offset and disorder}
\label{app:offdis}
We now  consider two types of control imperfection: magnetic field disorder and offset. Both interactions  are  fields along the longitudinal z axis, so many of their properties are similar. Here we formally demonstrate that: (i) any sequence with vanishing zeroth-order offset Hamiltonian must also have vanishing zeroth-order disorder and vice-versa (ii) for \textit{yxx} sequences this is also true to first-order. 

Consider the Hamiltonian $H=D_z+H_z$, where $D_\alpha$ is the dipolar interaction in the $\alpha$ direction and $H_\alpha=\Delta \sum_j  S_\alpha^j$ for the  offset and $H_\alpha=\sum_j w_j S_\alpha^j$ for the  disorder in the $\alpha$ direction. 
A uniform offset  can be viewed as a special disorder realization, so a vanishing average Hamiltonian in the presence of disorder  implies a vanishing average offset Hamiltonian. We then only need to demonstrate the converse in the following. 

The zeroth-order average Hamiltonian for single-body Hamiltonians can be calculated by considering a representative site. Thus the relative strength of the interaction among sites does not matter,  and a  vanishing zeroth-order offset is always  equivalent to a vanishing zeroth-order disorder. 

\begin{figure*}[!t]
\centering
\includegraphics[width=0.98\textwidth]{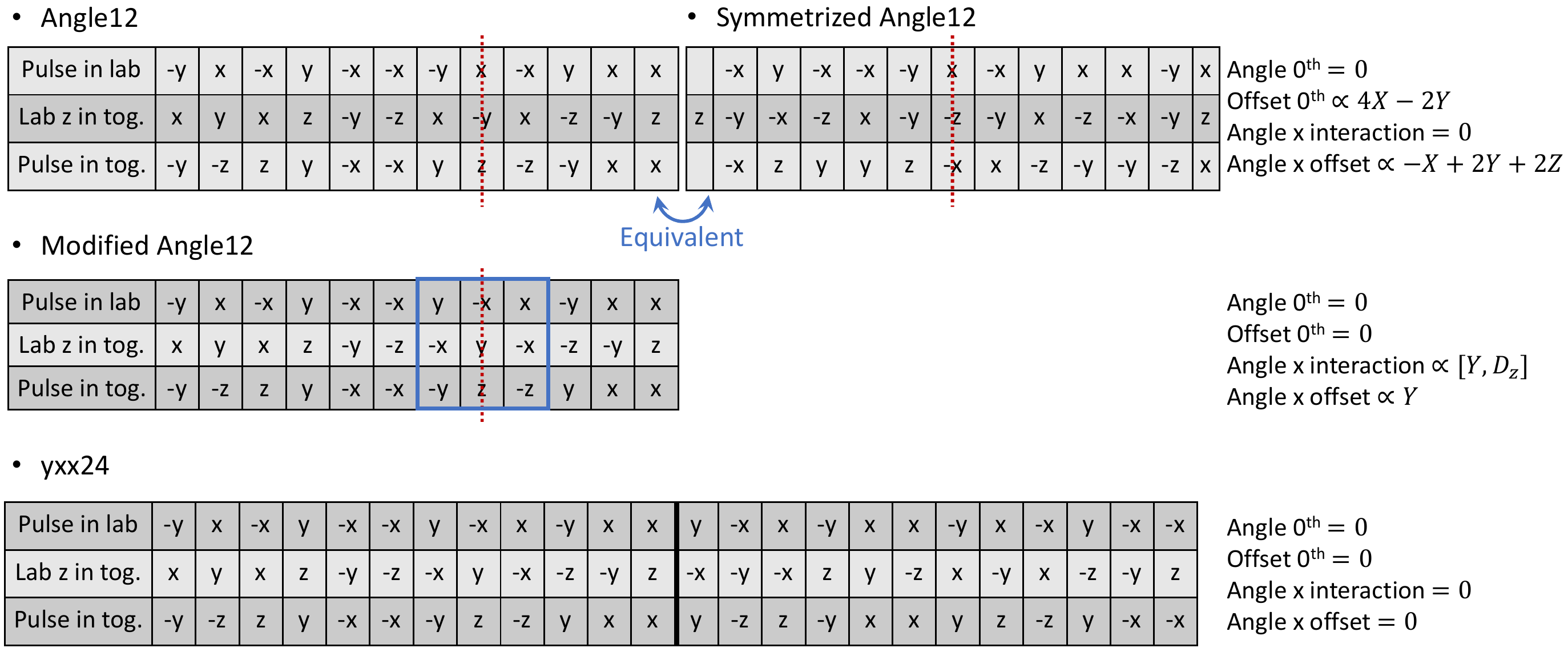}
\caption{\label{fig:yxx24intuition}
Analysis of Angle12, symmetrized Angle12, modified Angle12 and yxx24. For each sequence, we specify three elements: pulse in lab stands for the direction of pulse in lab frame; lab z in tog. stands for the direction of lab frame z axis in toggling frame after the corresponding pulse; pulse in tog. stands for the direction of pulse in toggling frame. Lab z in tog. determines the free evolution Hamiltonian (both interaction and offset) in toggling frame. For example, $-y$ indicates the Hamiltonian is $D_y-
H_y$. Pulse in tog. gives the angle error term. For example, $-y$ indicates the pulse unitary is $e^{i(pi /2+\epsilon)Y}$. AHT analysis results of Angle12 (which is equivalent to symmetrized Angle12), modified Angle12 and yxx24 are shown to the left of the tables. Modified Angle12 is obtained by $\pi$ phase shifting the pulses in the blue box. If we consider only the interaction and the offset, Angle12 and modified Angle12 are symmetric around the red dashed line. yxx24 is obtained from modified Angle12 by appending another modified Angle12 with $\pi$ phase shift.
}
\end{figure*}
The first-order average Hamiltonian contains three parts: the interaction-interaction commutator, which we can ignore for this discussion, the imperfection-imperfection commutator and the interaction-imperfection cross commutator. The  imperfection-imperfection commutator is a single-site operator so it has the same first-order average Hamiltonian for both disorder and offset. Then, the difference between disorder and imperfection lies in the interaction-imperfection cross commutator. 
Indeed, due to symmetries, $[D_\alpha,H_\alpha]=0$ for the offset (with $\alpha=x,y,z$), while this is not true for disorder. 
Still, we can show that for \textit{yxx} sequences,  additional terms arising from the disorder $[D_\alpha,H_\alpha]$ commutator sum up to zero. The detailed discussion is presented in~\cite{SM-RL}, while here we give two key factors. First,   \textit{yxx} sequences can be divided into blocks of 3 intervals of duration $\tau$, and the zeroth-order average interaction within each block is zero. This guarantees that there are no cross terms between different blocks. 
Second, if the zeroth-order offset vanishes, not only we have $\sum_j H_z^j=0$ when we sum over all time intervals, but also if we restrict the sum to the first (or 2$^{\text{nd}}$ and $3^{\text{rd}}$) intervals inside each block. 
In turns, this ensures that when summing over all blocks, commutators of the form $[D_\alpha,H_\alpha]$ add up to zero. Thanks to the similarity between the offset and disorder we were able to use the simpler form of the offset Hamiltonian in the traning algorithm, and still achieve robust sequences against disorder.

\section{Physical intuition for the construction of the yxx24 sequence}\label{appsec:yxx24intuition}
Here we explain how we design the yxx24 sequence starting from the Angle12  sequence, as an example of how human  insight can lead to better control. First we analyze Angle12 using AHT and present the results in \figref{fig:yxx24intuition}. The zeroth-order angle error vanishes, while the zeroth-order offset is proportional to $4X-2Y$ (the zeroth-order interaction term is zero as guaranteed by the \textit{yxx} pattern.) 
We notice that if we consider only the interaction and offset, Angle12 is symmetric, because the second row (``Lab z in tog.'') in \figref{fig:yxx24intuition} is mirror-symmetric around the red dashed line.
In other words, we can ``rotate'' the sequence to make it symmetric. Here by ``rotating'' we mean  shifting $n$ actions from the beginning to the end. For example, we can put the first 2 actions of Angle 12 at the end, so that it becomes $-x,y,-x,-x,-y,x,-x,y,x,x,-y,x$, which is symmetric with respect to the middle point. This sequence, labeled symmetrized Angle12 in~\figref{fig:yxx24intuition}, is equivalent (in terms of fidelity and leading order average Hamiltonian) to the original one found by DRL, if the angle error is ignored. Indeed,  the sequence rotation induces a unitary rotation of the  Floquet Hamiltonian that for decoupling sequences (where the target evolution is the identity) does not change the fidelity~\cite{Bukov15}. 
As the DRL agent only learns from the  propagator fidelity, it does not distinguish symmetric and rotated sequences. This is in contrast to traditional sequence-finding methods that are based on the approximated Floquet-Magnus expansion. Indeed, in the theoretical analysis it is convenient to consider the symmetric case as it reduces the number of nonzero terms in AHT. In particular, the symmetrized Angle 12 sequence has zero first-order average Hamiltonian except for the angle error, and we can thus focus on such terms [as listed in \figref{fig:yxx24intuition}]. 

With the AHT analysis in hand, we first want to modify the sequence such that it has vanishing zeroth-order average Hamiltonian. Notice that under Angle12 the offset gives a  zeroth-order  Hamiltonian  $4X-2Y$. To cancel this contribution we need to change the toggling frame offset Hamiltonian from $X$ to $-X$ in two intervals and from $-Y$ to $Y$ in 1 interval, while keeping the sequence symmetric. The toggling frame offset orientation is shown in the row labeled ``Lab z in tog.'' in~\figref{fig:yxx24intuition}. 
Therefore, we can cancel the zeroth-order offset by adding a $\pi$ phase shift to pulses within the blue box:  we name this new sequence  ``modified Angle12''. The AHT analysis of modified Angle12 is also shown in \figref{fig:yxx24intuition}. Its zeroth-order Hamiltonian is zero. 

Now the dominant nonzero Hamiltonian comes from the first-order cross terms between angle error and interaction, as well as between angle error and offset (since the angle error is not mirror symmetric.) Once the zeroth-order average Hamiltonian is zero, we can double the sequence and use the symmetry to get rid of the first-order corrections. Notice the two first-order corrections are antisymmetric under a $\pi$ rotation along the z-axis and thus can be easily cancelled by combining the  modified Angle12 and another modified Angle12 with a $\pi$ phase shift.  We thus arrive at the yxx24 sequence, whose zeroth-  and first-order average Hamiltonian are all zero.

\section{Experiments on Adamantane}
\label{appsec:adamantane}
In order to ensure that the improved performance of the yxx24 and yxx48 sequences in calcium fluoride and fluorapatite shown in Figure~\ref{fig:yxx_t} were not unique to the spectrometer and probe used there, we also compared the performance of these sequences to the Cory48 sequence in a powdered adamantane sample on a different 300 MHz Bruker DSX spectrometer.

Adamantane (C$_{10}$H$_{16}$) is a plastic solid with a high degree of internal motion.  The proton (hydrogen nuclei) dipolar linewidth is about $(2\pi)13$~krad/s \cite{Cory90} and the system is often used to model a 3D spin system. The T$_1$ relaxation time for the proton spins in adamantane at room temperature was measured to be just under 1 s.

Figure~\ref{fig:yxx_adamantane} shows the comparison between the performance of the Cory48, yxx24 and yxx48 as measured by the average correlation metric introduced earlier.  The collective $\pi/2$ pulses used had a pulse width $t_w = 2$ $\mu$s. The pulse center-to-center delay $\tau=8 \mu$s.  Sequence performance degraded significantly as  this duration was decreased, likely due to finite stabilization times during phase switching and the overlapping of pulse transients.

It should be noted that while the decay of the average correlation metric resembles the results of a single line-narrowing experiment, care should be taken while comparing them directly.  The bi-exponential behavior of the decays for the Cory48 sequence gives rise to effective linewidths of  823 Hz and 88 Hz respectively.  The effective linewidth is significantly broader than the 3.5 Hz obtained in \cite{Cory90} probably due to the longer $\pi/2$ pulse and $\tau$ values used here.  Similarly, the fidelity of the sequences shown here are seen to be slightly lower than those obtained in Figure~\ref{fig:yxx_t}.  Note that the plot shows the data acquired after an even number of cycles, with a maximum of 128 cycles.

The yxx24 and Cory48 sequences show almost identical behavior at all timescales in these experiments.  However, while the fidelity of the yxx48 sequence is initially lower than that of the other two, the performance at longer timescales matches that of the other two.

\begin{figure}[!htbp]
\centering
\includegraphics[width=0.9\columnwidth]{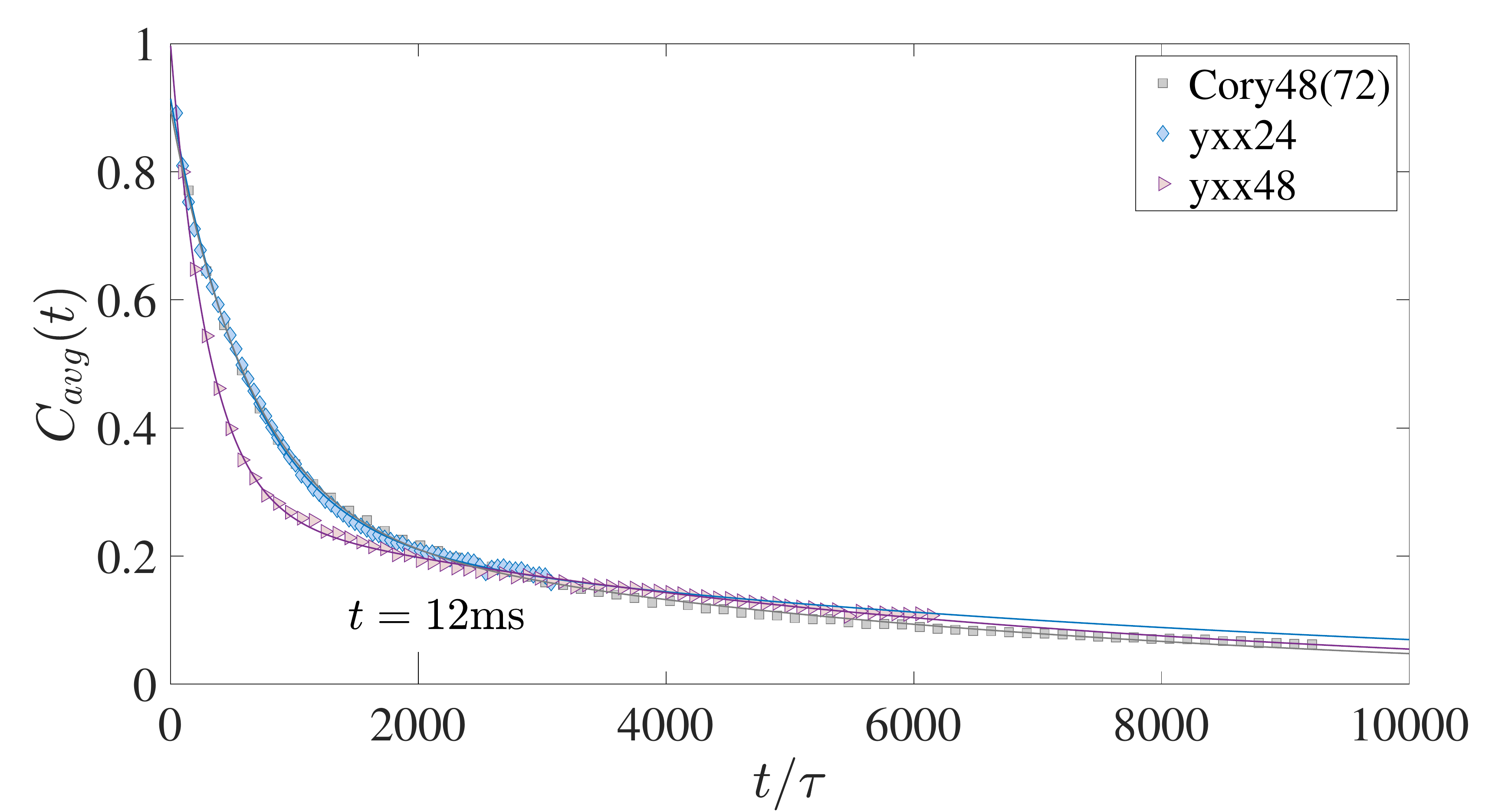}
\caption{\label{fig:yxx_adamantane}
Experimentally measured average correlation at the best calibrated condition for adamantane as a function of time. The solid curves are biexponential fits to the data. 
}
\end{figure}

\bibliography{Biblio}% Produces the bibliography via BibTeX.

\section*{Supplemental Material}

\section{Experimental system}
In the main text, we present  experimental results  from CaF$_2$ and flourapitite (FAp). The description of  our FAp sample can be found in the supplementary information of Ref.~\cite{Peng21}. Here we provide details of the CaF$_2$ sample. We use a single crystal of CaF$_2$ where the $^{19}$F atoms form a simple cubic lattice with nearest-neighbor distance $d=2.7315$~\r{A}~\cite{Canters69}. The sample is of millimeter scale. The experiments are performed at room temperature inside a magnetic field $B=7T$ using a 300MHz Bruker spectrometer. The total Hamiltonian is 
\begin{equation}
H_{tot}=\omega_F \sum_k S_z^k + \sum_{j<k} \frac{\hbar \gamma_F^2}{|\vec{r}_{jk}|^3}\left(\vec{S}^j \cdot\vec{S}^k-\frac{3\vec{S}^j\cdot \vec{r}_{jk}\vec{S}^k\cdot \vec{r}_{jk}}{|\vec{r}_{jk}|^2}\right),
\end{equation}
where $\omega_F=\gamma_F B\approx(2\pi)282.4$~MHz is the Zeeman frequency of $^{19}$F, $\gamma_F\approx241.67\times10^6$~rad$\cdot$s$^{-1}$ $\cdot$T$^{-1}$ is the gyromagnetic ratio~\cite{Canters69}, $\vec{r}_{jk}$ is the displacement between the $j$ and $k$ spins. Because the Zeeman frequency is much larger than other energy scales in experiments, it is convenient to write the Hamiltonian in the rotating frame set by the Zeeman interaction, and neglect the fast rotating parts of the Hamiltonian (rotating wave approximation). We then obtain the secular dipolar Hamiltonian
\begin{equation}
D_z=\frac{1}{2}\sum_{j<k}^N J_{jk}\left(3S_z^j S_z^k - \vec{S}^j\!\cdot\!\vec{S}^k \right),
\end{equation}
with $J_{jk}=\hbar\gamma_F^2\frac{3\cos(\theta_{jk})^2-1}{|\vec r_{jk}|^3}$, where $\theta_{jk}$ is the angle between $\vec r_{jk}$ and the magnetic field (aligned with the $z$-axis). If $\theta_{jk}=0$, the nearest-neighbor coupling strength reaches its maximum value $J_M=65.8$~krad/s. As we see, the coupling strength depends on the orientation of the sample with respect to the magnetic field. The orientation is difficult to measure precisely, however, it can be inferred from the free induction decay (FID) of $^{19}$F spins.

\begin{figure}[h]
	\centering
	\includegraphics[width=\columnwidth]{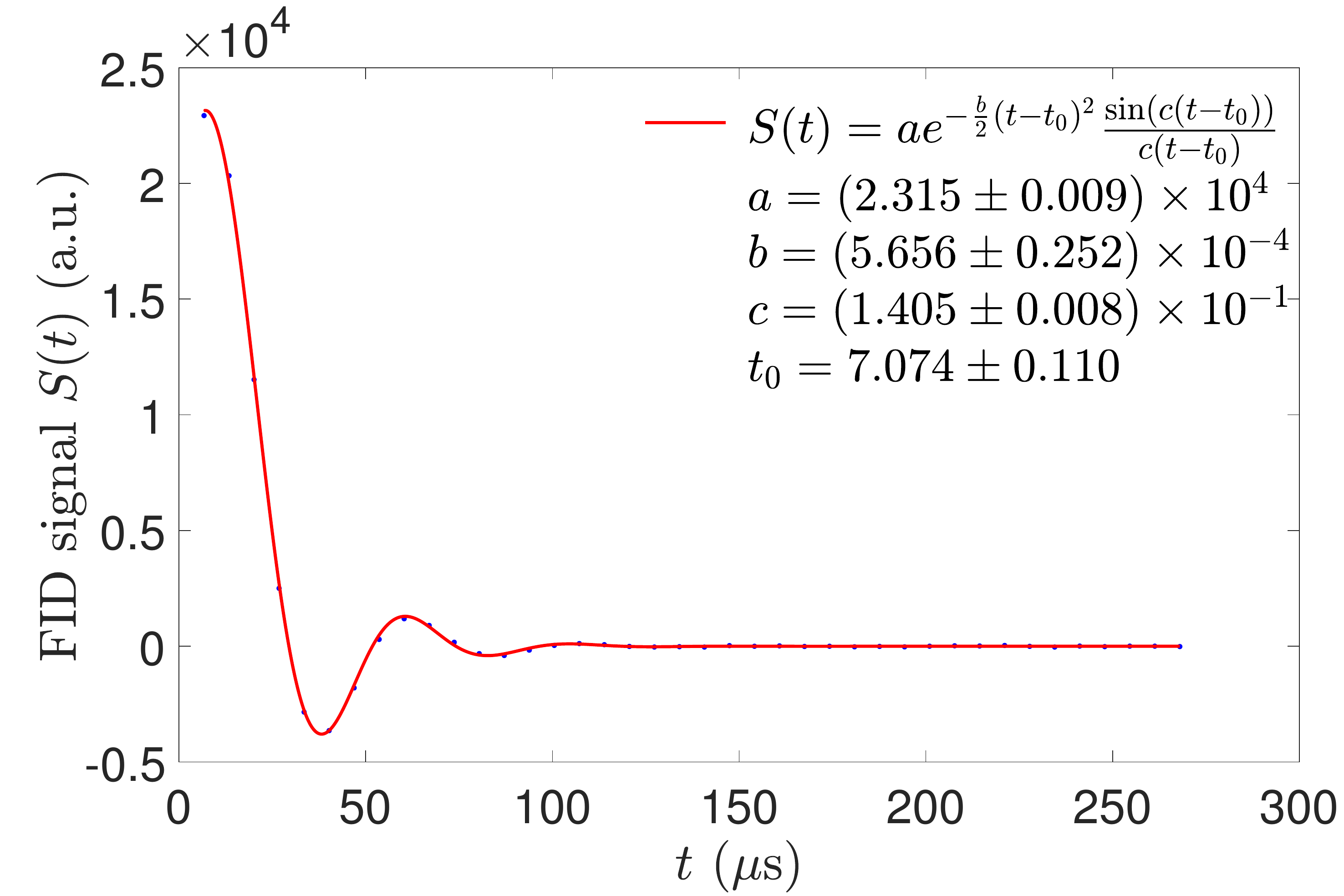}
	\caption{\label{fig:fitFID}
		Experimental FID data (blue dots) and fitted curve (red curve). The fitted coefficients are shown in the panel, with the error being one standard deviation. Fitting $R^2=0.9997$.
	}
\end{figure}

The FID reveals the dephasing process of the transverse polarization under dipolar interaction $S(t)=\mathrm{Tr}[X(t)X]$, with $X=\sum_k S_x^k$ and $X(t)=e^{-iD_zt}Xe^{iD_zt}$. Fourier transform of the FID signal $S(t)$ gives the nuclear magnetic resonance (NMR) spectrum. Reference~\cite{Canters69} reports that the second moment of the NMR spectrum ($M_2=S''(0)/S(0)$) is related to the orientation of the sample via $M_2=\frac{9}{16}\sum_k J_{jk}^2$ for spin-1/2. For a simple cubic lattice, the 2nd moment reduces to $M_2=\frac{1}{16}[c_1+c_2(\lambda_1^4+\lambda_2^4+\lambda_3^4)]J_M^2$, with $c_1=-7.2722$, $c_2=37.3260$ and $\lambda_1, \lambda_2, \lambda_3$  the direction cosines of the magnetic field with respect to the crystal axes. The maximum $M_2^{1/2}$ is 90.2~krad/s at the [100] direction, while the minimum is 37.4~krad/s for the [111] direction. We measured $M_2^{1/2}=84.5$~krad/s so the orientation is closer to [100] (details explained later). To compare the experiments (3D $1/r^3$ interaction) with simulations (1D nearest neighbor interaction), we define an effective interaction strength by equating the $M_2$ for the two cases. For 1D nearest neighbor dipolar interaction with strength $J$, $M_2^{1/2}=3\sqrt{2}J^2/4$. Plugging in $M_2^{1/2}=84.5$~krad/s from our experiments, we get the corresponding effective 1D interaction strength $J_{\textrm{eff}}=79.7$~krad/s. We note that even though the $M_2$ is matched between experiments and simulation with interaction strength $J_{\textrm{eff}}$, for other experiments and observables they are not supposed match exactly, but should exhibit qualitatively similar behavior. For comparison, FAp, a crystal with a large anisotropy that gives rise to a quasi-1D behavior, has $J=32.7$~krad/s when aligned with its c-axis along z. 

Figure~\ref{fig:fitFID} shows the FID signal for our CaF$_2$ sample, where a solid echo is used to overcome the dead time of the receiver~\cite{Powles62}. Due to the limitation of time step, it is difficult to directly calculate $M_2$ using $S''(0)/S(0)$. Instead, we fit $S(t)$ to an empirical function $S(t)=ae^{-\frac{b}{2}(t-t_0)^2}\frac{\sin(c (t-t_0))}{c (t-t_0)}$~\cite{Abragam61}, so that we can evaluate $M_2$ more accurate by utilizing more data points. The fitting result is shown in Fig.~\ref{fig:fitFID}. Then we calculate the second derivative of the empirical function and obtain $M_2^{1/2}=\sqrt{b+c^2/3}=84.5\pm0.5$~krad/s.

\section{Numerical verification}
\subsection{Average correlation and infidelity}\label{sec:SM48}
In the main text, we use the fidelity of the propagator as the reward for RL, but fidelity is not  measurable in our experiments (indeed it would require process tomography.) 
\begin{figure}[h]
	\centering
	\includegraphics[width=\columnwidth]{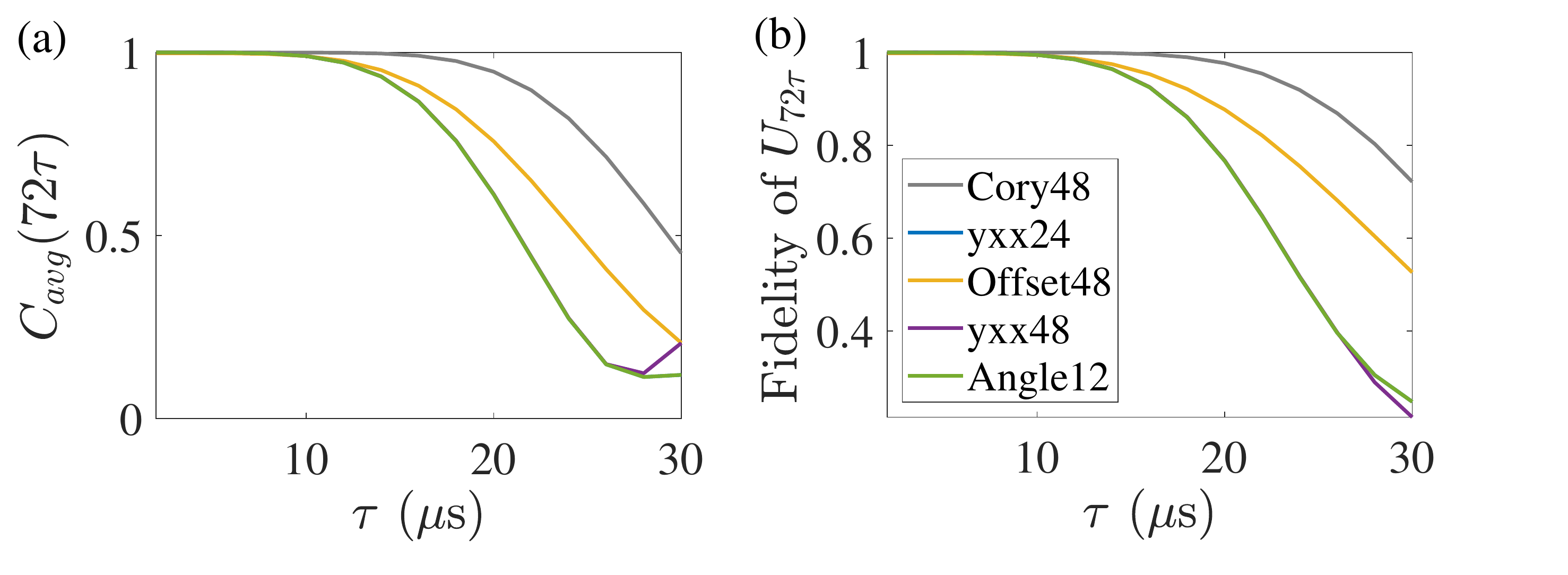}
	\includegraphics[width=\columnwidth]{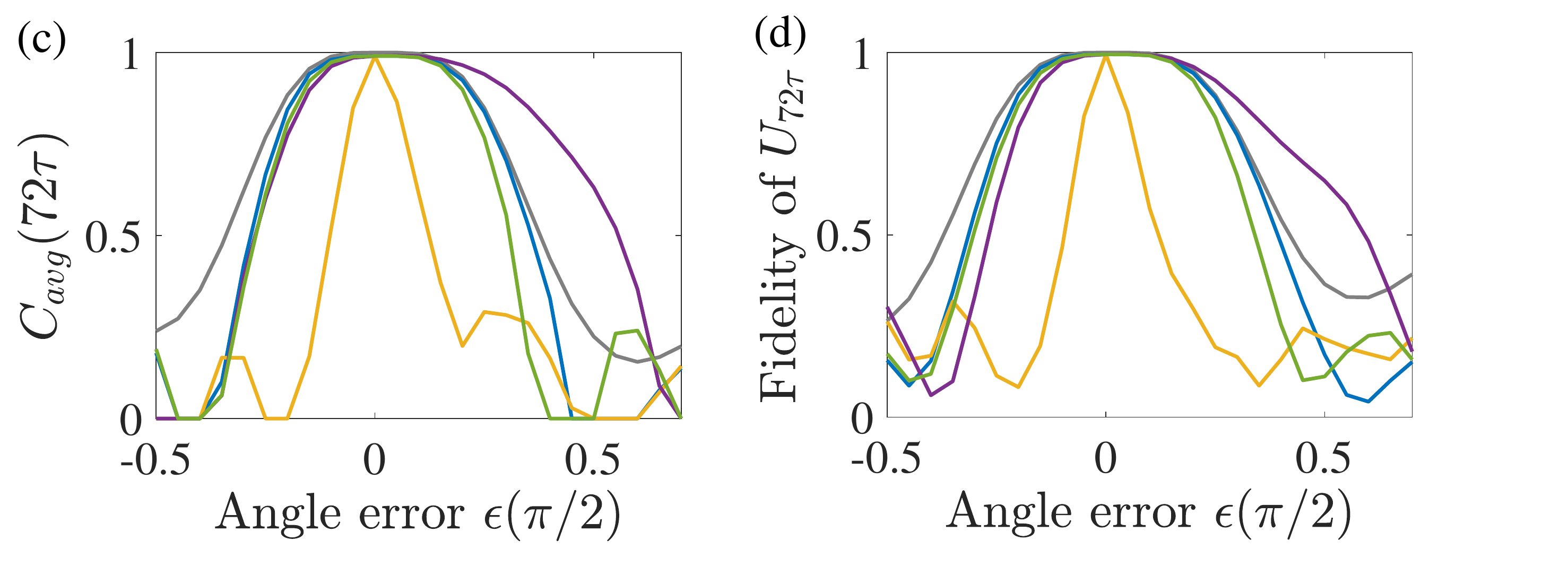}
	\includegraphics[width=\columnwidth]{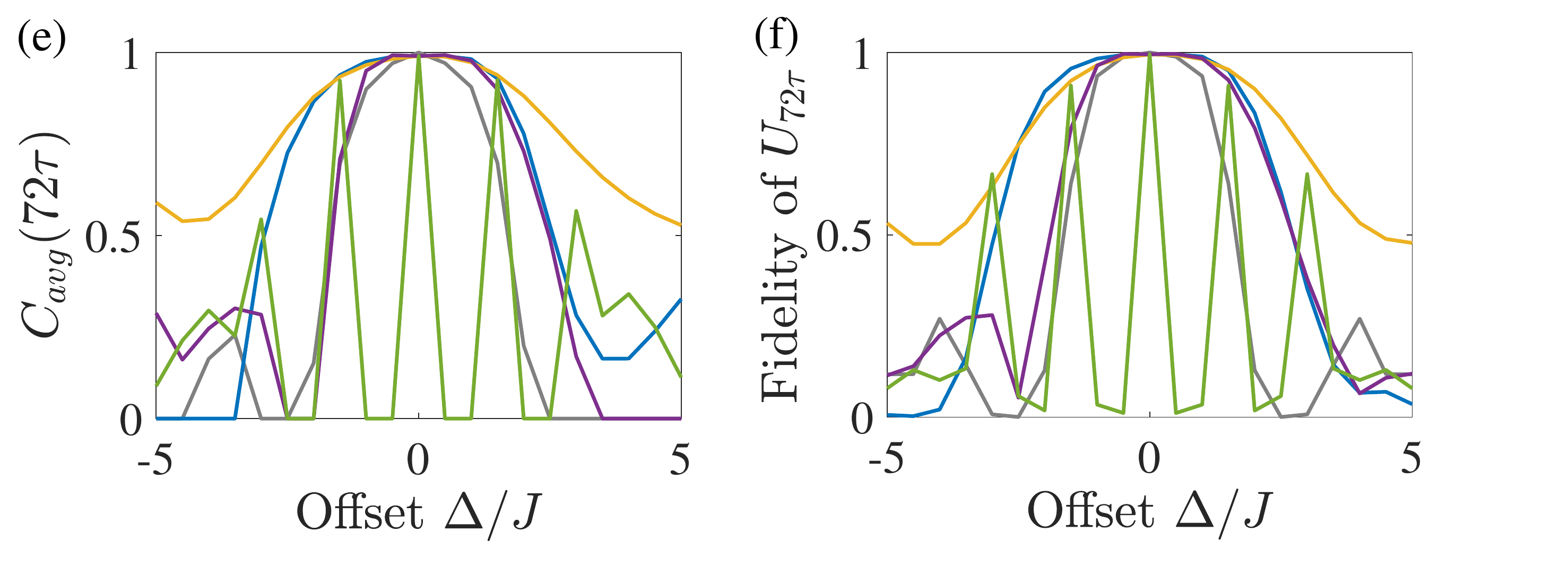}
	\caption{\label{fig:PU}
		Numerical comparison between  average correlation $C_{avg}=(C_{XX}C_{YY}C_{ZZ})^{1/3}$ and fidelity of the effective propagator at $72\tau$, for different $\tau$ (a-b), offset $\Delta$ (c-d) and angle error $\epsilon$ (e-f). The average correlation quantitatively reflects the fidelity. Here we use $t_w=1~\mu$s to mimic experimental conditions, system size $L=6$ and other parameters  as in Fig.~3 of the main text.
	}
\end{figure}
Here we verify that the experimentally measurable average correlation can quantitatively reflect fidelity. Similar to experimental results in Fig.~3 and Fig.~5 in the main text, we plot the simulated $C_{avg}(72\tau)$ and fidelity in~Fig.~\ref{fig:PU}. 

Comparing fidelity and average correlations as a function of  all the parameters we explore, we find that the average correlation closely resembles the propagator fidelity. 
%When plotting for the longer time $72\tau$, new features arise. 
We also reveal new features in the Angle12 sequence when considering  a broader range of offsets in simulations. 
The multiple peaks found for Angle12  (Fig.~\ref{fig:PU}e-f) as a function of the offset indicate that the sequence is effective only for some particular values of the ratio $\Delta/J$. Indeed, the zeroth order average offset of Angle12 forms an effective magnetic field whose strength is $\sqrt{20}\Delta/12$. Then, the peaks arise when $6\sqrt{20}\Delta\tau$ is a multiple of $2\pi$. Therefore, the peak center distance in unit of $\Delta/F$is $\pi/(6\sqrt{5}J\tau)\approx0.72$ in agreement with Fig.~\ref{fig:PU}(e-f).

\subsection{Comparing performance of different-length sequences}
Experimentally measured average correlation DRL sequences, yxx sequences and Cory48 for different $\tau$ (a), angle error (b) and frequency offset (c).
\begin{figure}[h]
	\centering
	\includegraphics[width=\linewidth]{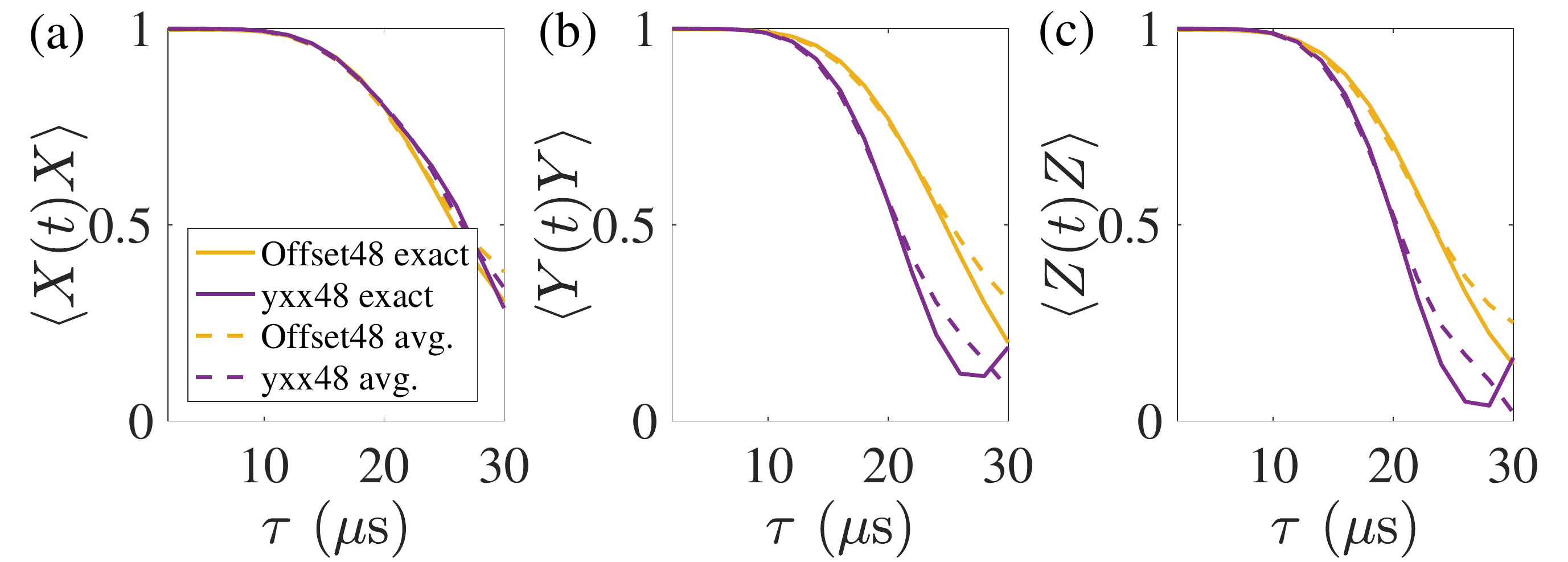}
	\includegraphics[width=\linewidth]{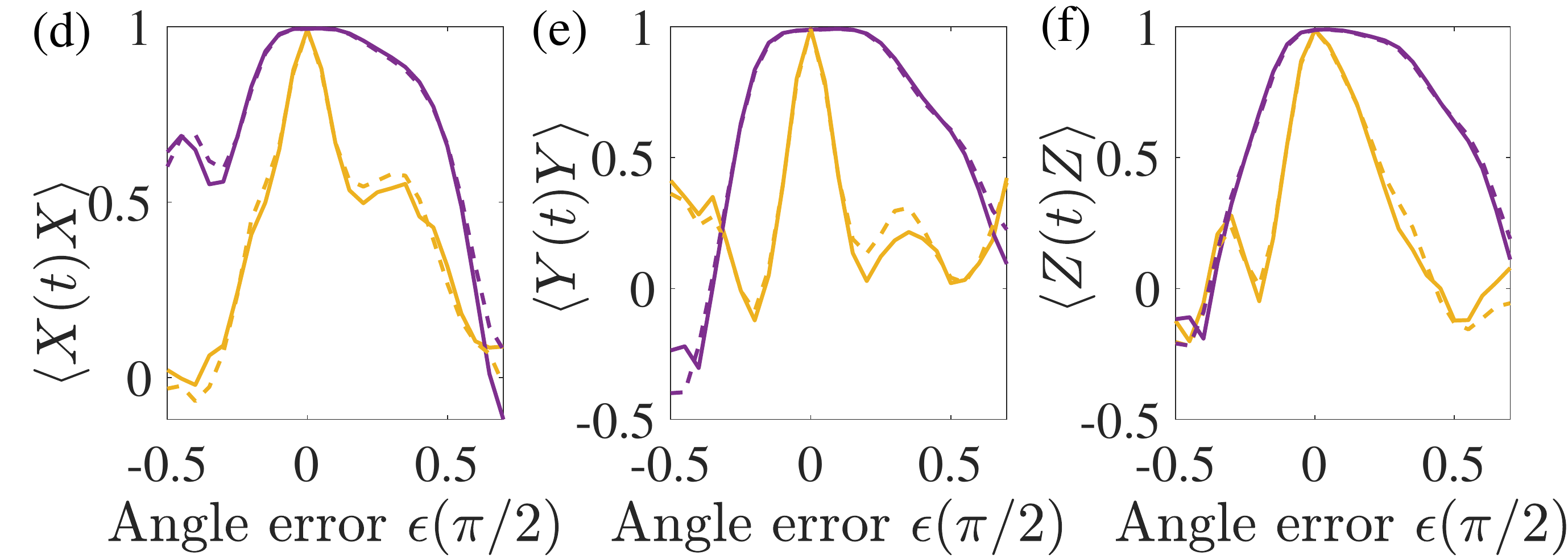}
	\includegraphics[width=\linewidth]{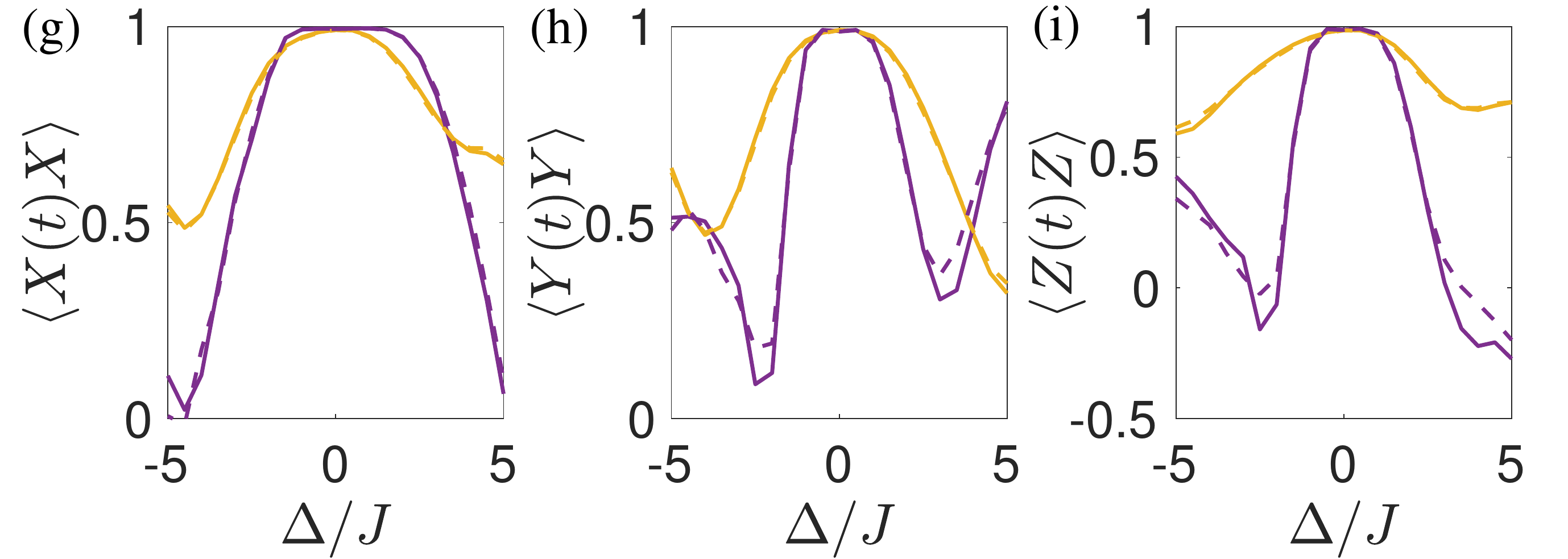}
	\caption{
		\label{fig:ave_exact}
		Comparison of average signal and the exact signal using numerical simulation. Exact results (solid curves) are obtained by first calculating the Floquet Hamiltonian $H_F$ and then evolve the initial state under $H_F$ for time $t=72\tau$; averaged results (dashed curves) are calculate as $(\langle \mathcal{O}(48\tau)\mathcal{O} \rangle +\langle \mathcal{O}(96\tau)\mathcal{O} \rangle)/2$, with $\mathcal{O}=X,Y,Z$. 
		Parameters are the same as in Fig.~\ref{fig:PU}.
	}
\end{figure}

\begin{figure*}[!htbp]
	\centering
	\includegraphics[width=0.32\linewidth]{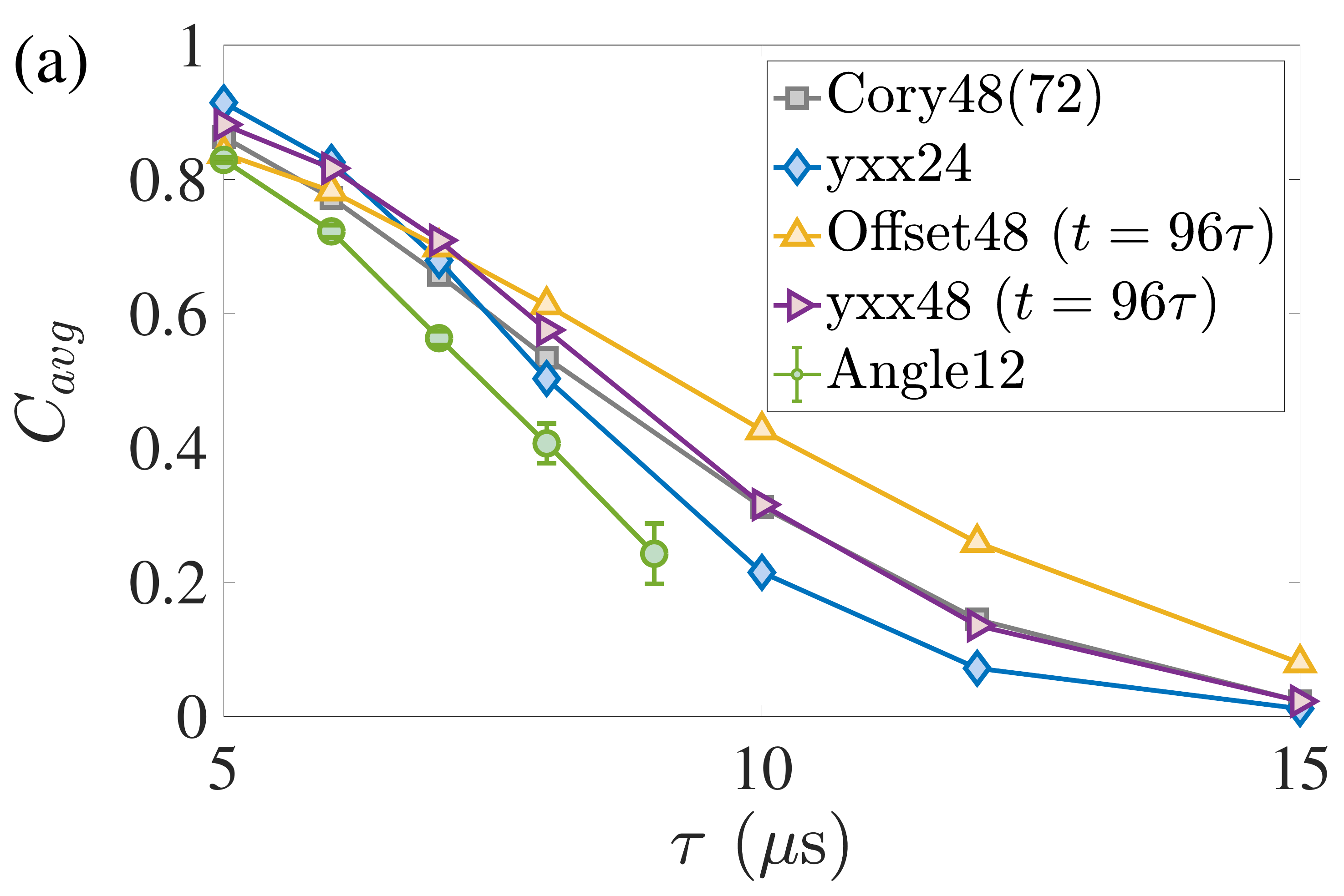}
	\includegraphics[width=0.32\linewidth]{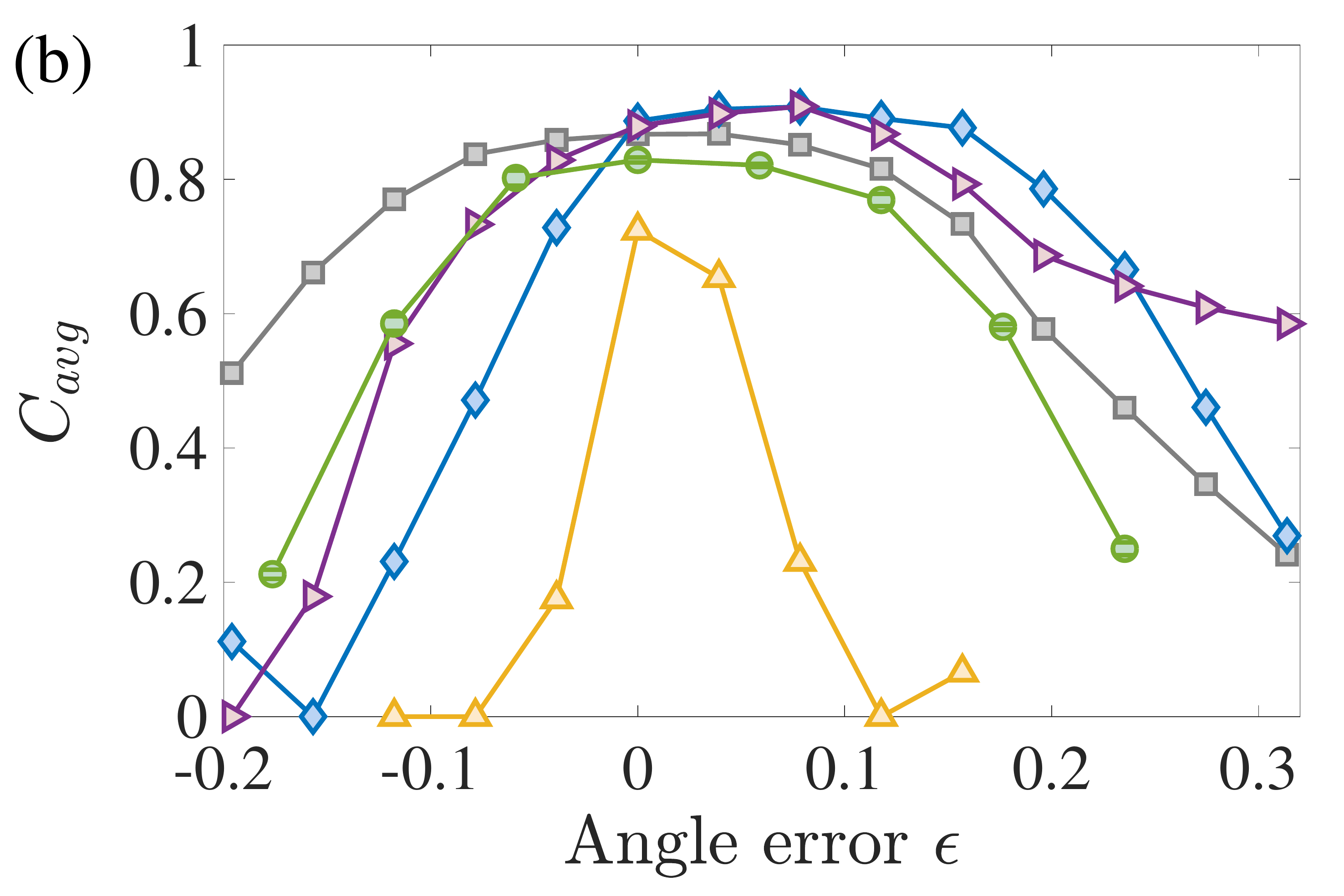}
	\includegraphics[width=0.32\linewidth]{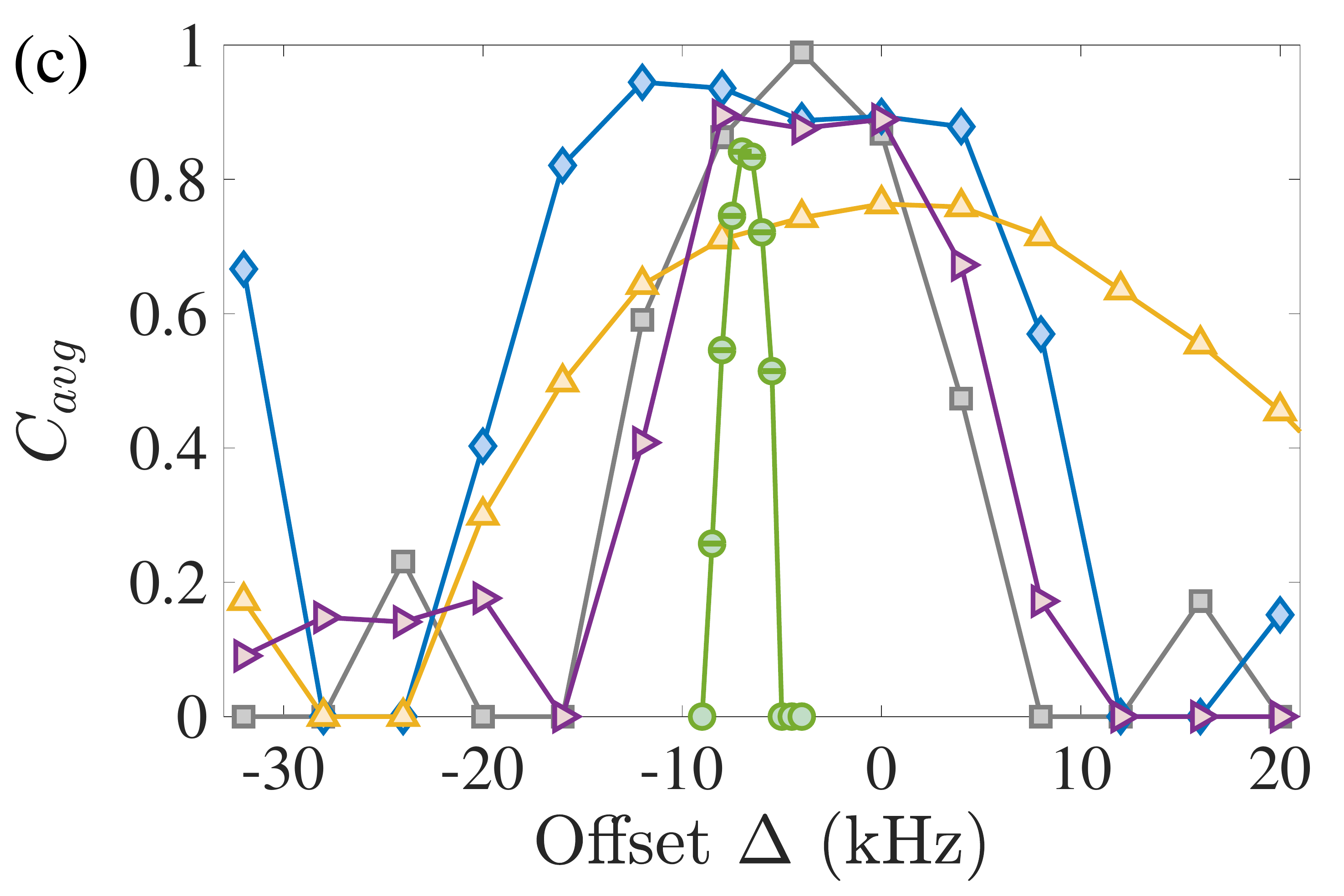}
	\caption{
		\label{fig:Exp96}
		Autocorrelations of yxx48 and Offset48 are evaluated at $t=96\tau$, while others are at $t=72\tau$. Imperfections are set to zero unless specified by the horizontal axis, with the exception of Angle12 experimental data in (a) and (b), which are taken at optimal $\Delta$ unequal to zero due to the presence of phase transient (see Appendix for details). 
		In (b) and (c), the pulse center-to-center delay is $\tau=5~\mu$s. 
		Error bars of Cory48 and Offset48 experimental data are determined from the noise in the free induction decay, which is smaller than the marker size thus not shown. Angle12 has larger error bars in (a) and (b) due to the inaccuracy in finding the optimal $\Delta$. 
	}
\end{figure*}

Since we study sequences with different length (12$\tau$, 24$\tau$,  48$\tau$, 72$\tau$), it is necessary to have a fair metric for comparison. As the longest sequence is Cory48 (72$\tau$), we choose to evaluate every sequence at 72$\tau$. However, this poses a problem for sequences of length 48$\tau$, that do not allow directly measuring  the signal at 72$\tau$. 
Although this problem could be in principle solved by comparing the signal at 144$\tau$, at such long time the unitary fidelity is low. 
As a result, we choose to approximate the signal at 72$\tau$ with the average of signal at 48$\tau$ and 96$\tau$. 
This approximation is numerically verified in Fig.~\ref{fig:ave_exact}, where the exact signal at 72$\tau$ is obtained by first calculating the Floquet Hamiltonian $H_F$ and then evolving the initial state under $H_F$ for  a time $t=72\tau$. The approximation is very good, especially in the high-fidelity region. 
In combination  with Fig.~\ref{fig:PU}, which compares fidelity and average autocorrelation at 72$\tau$, we can  conclude that the average of autocorrelations at 48$\tau$ and 96$\tau$ is a good approximation to the propagator fidelity.
To further verify that the advantages of yxx48 and Offset48, we plot their autocorrelations at $t=96\tau$ and compare with other sequences at $t=72\tau$ in~\figref{fig:Exp96}. Though the comparison underestimates the performance of yxx48 and Offset48, we still can see their expected robustness.

\section{Additional data}
\subsection{Individual autocorrelations}

In the main text we use the geometric average of $C_{XX}, C_{YY}$ and $C_{ZZ}$ as a experimental metric to compare different sequences. Here we present $C_{XX}, C_{YY}$ and $C_{ZZ}$ of Cory48(72), yxx24, Offset48 and yxx48 individually in Fig.~\ref{fig:indiv}. The individual correlations of Angle12 is presented in Appendix.

\subsection{Numerical results of Ideal6 and PW12}
\begin{figure}[b]
	\centering
	\includegraphics[width=\linewidth]{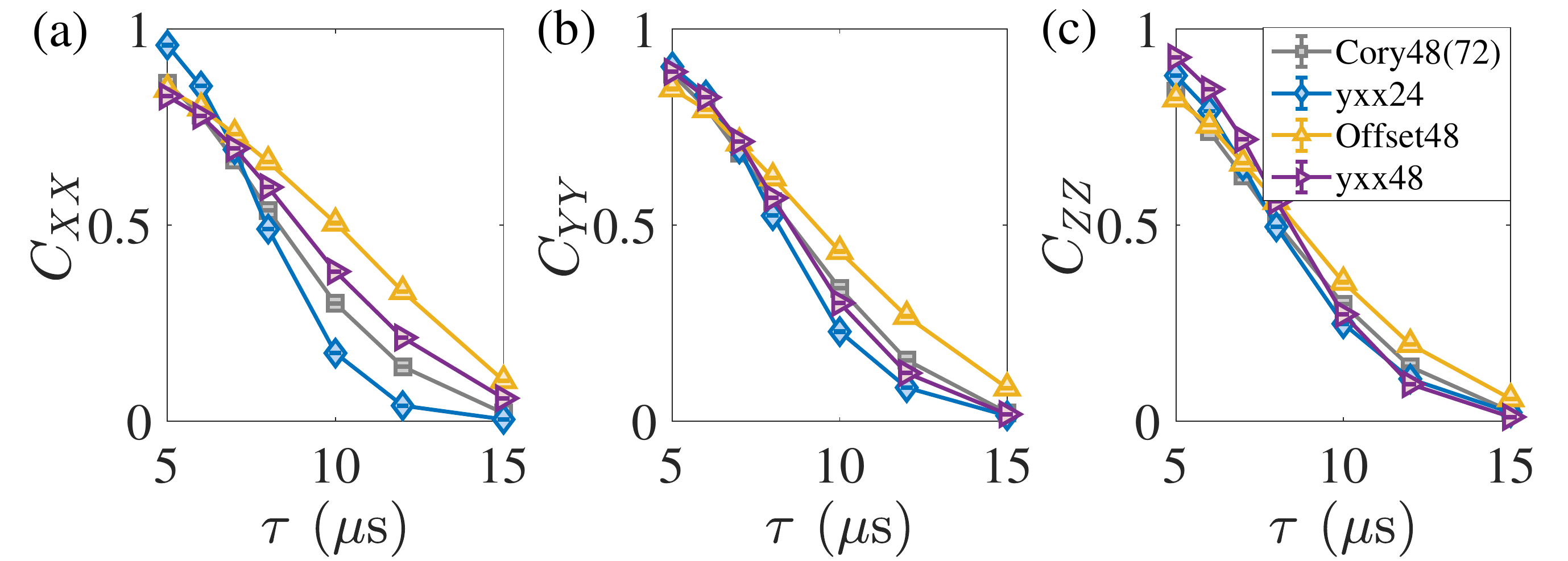}
	\includegraphics[width=\linewidth]{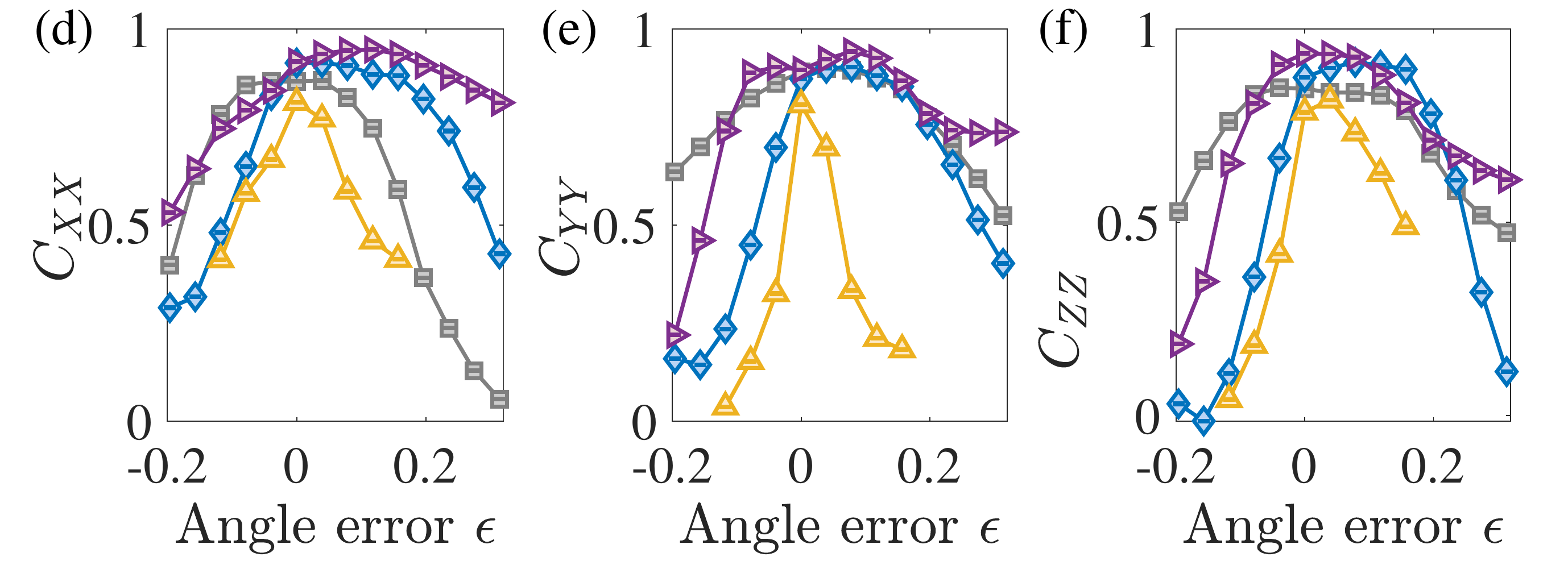}
	\includegraphics[width=\linewidth]{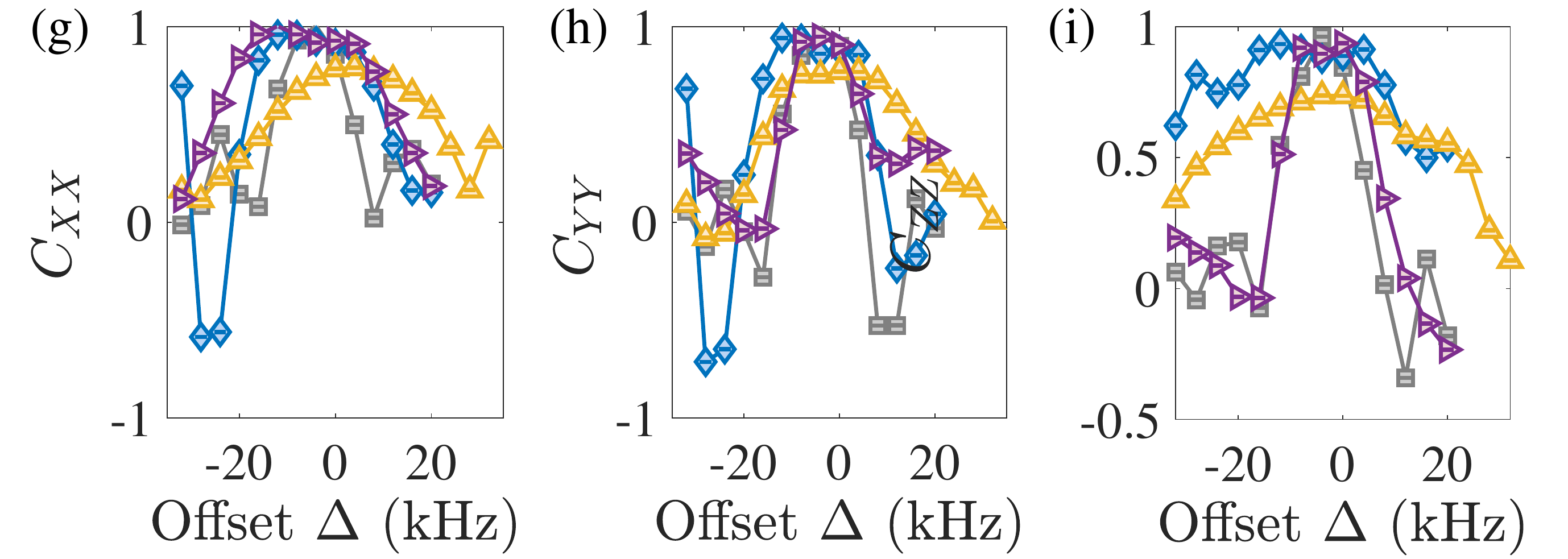}
	\caption{\label{fig:indiv}
		Experimental correlations $C_{XX}, C_{YY}$ and $C_{ZZ}$  Cory48(72), yxx24, Offset48 and yxx48. 
	}
\end{figure}
\begin{figure}[h]
	\centering
	\includegraphics[width=0.48\linewidth]{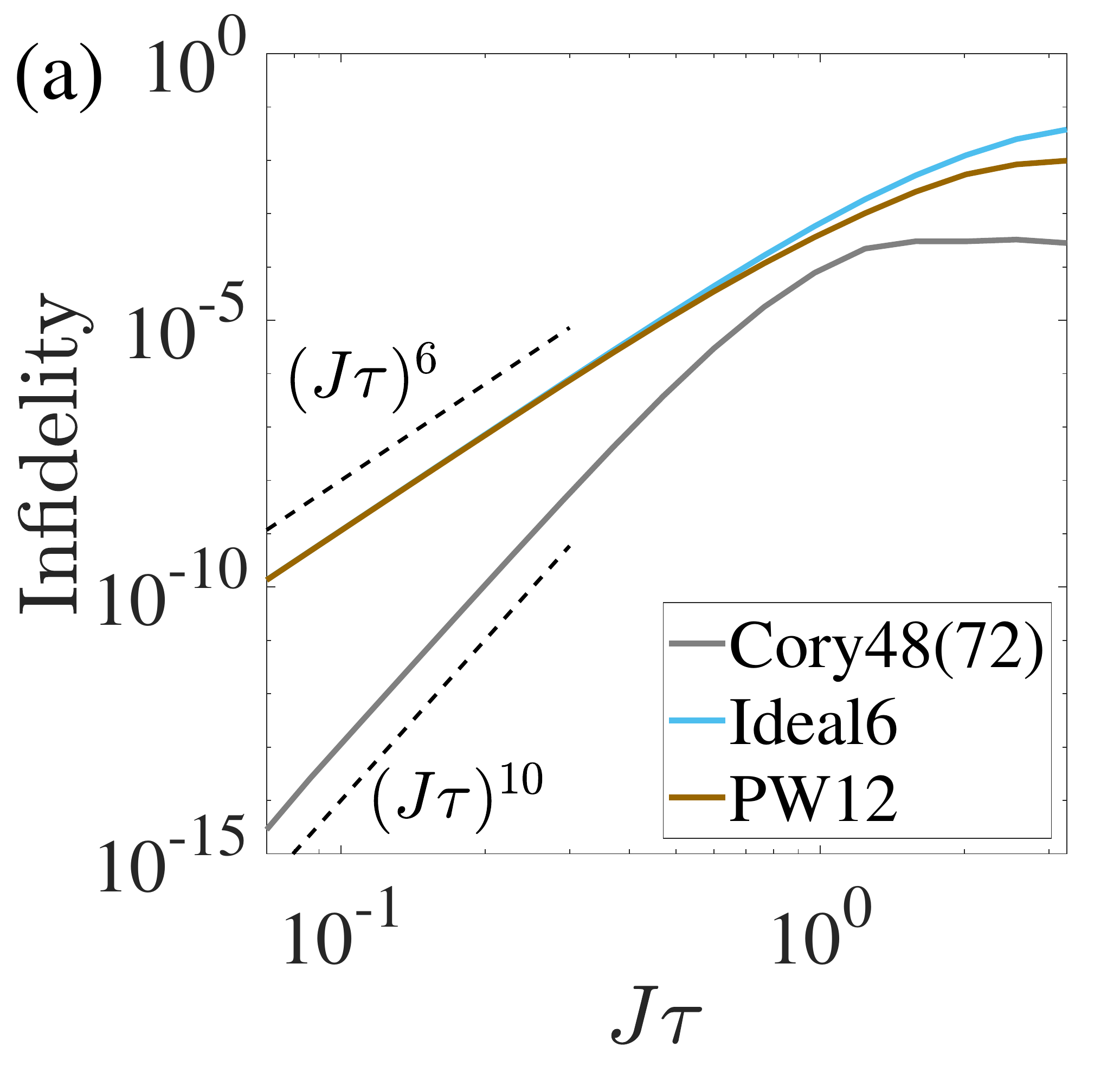}
	\includegraphics[width=0.48\linewidth]{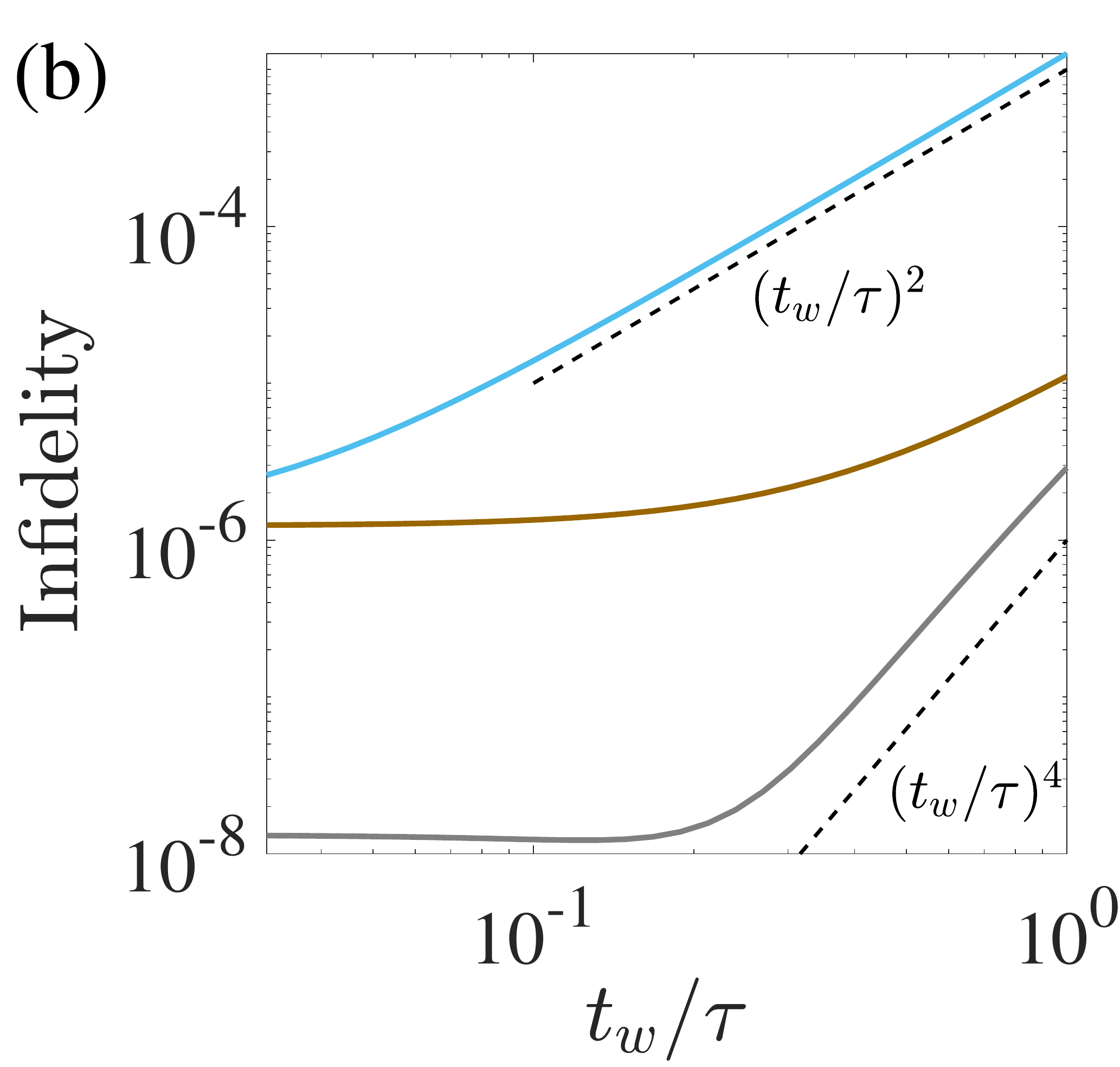}
	\caption{\label{fig:ideal-PW}
		Numerically simulated propagator infidelity $1-F$ of Ideal6 (cyan) and PW12 (brown) sequences for different $J\tau$ (a) and pulse width $t_w$ (b). Results for Cory48 (gray) are also presented for comparison. Dashed lines show the scalings specified by nearby expression. Imperfections are set to zero unless specified by the horizontal axis.
		In (a) we assume the pulse width is infinitesimal. In (b) $\tau=10~\mu$s. We use $J=32.7$~krad/s as in FAp, $N=8$, periodic boundary condition and assume nearest neighbor interaction. 
	}
\end{figure}
Here we present  numerical results for the Ideal6 and PW12 sequences (\figref{fig:ideal-PW}.) Both sequences cancel the interaction up to first-order average Hamiltonian assuming ideal conditions (Ideal6) or finite-width pulses (PW12.) Because of these assumptions, they cannot be tested experimentally, where other intrinsic errors are present. 
Although they are not as good as Cory48, which cancels up to third order coupling effects, their scaling is the same as yxx24 and yxx48 (Fig. 5 in the main text), despite being much shorter. 
PW12 shows a smaller increase in infidelity when the pulse width $t_w$ is increased  compared to Cory48 and Ideal6. PW12 cancels pulse width up to first order, although the scaling for PW12 is not evident here due to the practical constraint  $t_w<\tau$.

\subsection{RL hyperparameters}
%todo: SM?
We list RL hyperparameters in Table~\ref{tab:hyper}, and we initialize the DNN with Gaussian random weight.

\begin{table}[hbtp]
	\begin{tabular}{|c|c|}
		\hline
		\textbf{Name}  & \textbf{Hyperparameters}  \\ \hline 
		Ideal6 & $N=201, P=11$, NN:-64-64- \\ \hline
		Offset48 & $N=3001, P=31$, NN:-512-64-  \\ \hline
		Angle12  & $N=801, P=21$, NN:-128-64- \\ \hline
		PW12 & $N=801, P=21$, NN:-128-64- \\ \hline
		yxx48  & $ N=801, P=21$, NN:-128-64-  \\ \hline
		yxx24  & Built from Angle12  \\ \hline 
	\end{tabular}
	\caption{\label{tab:hyper}
		Hyperparameters of RL algorithm. $N$ and $P$ are the population size and number of selected parents. NN stands for neural network and the numbers following are the number of neurons in two hidden layers. }
\end{table}

\section{Analytical results}\label{app:thm}
\subsection{AHT analysis of learned sequences}
Typically,   sequences are designed by matching the average Hamiltonian order by order.
In RL, although the machine is not aware of the analytical perturbation tools,  many of the learned sequences do have vanishing low-order average Hamiltonians, as they lead to good fidelity.
We show the results of AHT for the DRL pulse sequences in \tabref{tab:seqAHT}. The scalings of unitary propagator fidelities shown in Fig. 3 and Fig. 5 of the main text are in agreement with the AHT results here.

\begin{table*}[!htbp]
	\begin{tabular}{|c|c|c|c|c|}
		\hline
		\textbf{Name}  & \textbf{interaction $J$} & \textbf{offset $\Delta$} &  \textbf{pulse width $t_w$} & \textbf{angle error $\epsilon$}  \\ \hline 
		Ideal6 & $(J\tau)^2$ & $\Delta$ & $Jt_w$ & $\epsilon$ \\ \hline
		Offset48 &  $(J\tau)^2$ & $\Delta^2, \Delta(J\tau)$ & $Jt_w$ & $\epsilon$ \\ \hline
		Angle12  &  $(J\tau)^2$ & $\Delta$ &  $Jt_w O((J\tau)^2)$ & $O(\epsilon^4), \epsilon^2(J\tau)$ \\ \hline
		PW12 &  $(J\tau)^2$ & $\Delta$ &  $Jt_w O((J\tau)^2)$ & $O(\epsilon^4), \epsilon(J\tau)$ \\ \hline
		yxx48  &  $(J\tau)^2$ & $O(\Delta^4), \Delta^2(J\tau)$ &  $Jt_w O((J\tau)^2)$  &  $\epsilon^3, \epsilon^2(J\tau)$   \\ \hline
		yxx24  & $(J\tau)^2$ & $O(\Delta^4), \Delta^2(J\tau)$ & $Jt_w O((J\tau)^2)$  &  $O(\epsilon^4), \epsilon^2(J\tau)$   \\ \hline 
		Cory48(72) & $(J\tau)^4$ & $\Delta^3, \Delta(J\tau)^2$ & $Jt_w O((J\tau)^2)$ & $O(\epsilon^4), \epsilon^2(J\tau)$ \\ \hline 
	\end{tabular}
	\caption{\label{tab:seqAHT}
		The leading non-vanishing average Hamiltonian of all pulse sequences. $O(\cdot)$ means the order is lower bounded by $(\cdot)$. For example, $O(\epsilon^4)$ means the term could be $\propto \epsilon^4$ or  $\propto\epsilon^5$ or higher orders.}
\end{table*}

\subsection{Decoupling sequence length}
We provide one theorem about the length of decoupling sequences, which helps reducing the searching space of our DRL algorithm.
\begin{AHT}\label{thm:length}
	To cancel the dipolar interaction $D_z$ to 1st order in AHT with collective rotations, we need the length of the pulse sequence to be $L=6n,n\in \mathbb N^+$
\end{AHT}
\begin{proof}
	To cancel the 0th order AHT arising from the secular dipolar Hamiltonian, we require an equal number of $D_x,D_y$ and $D_z$ in the toggling frame, such that $D_x+D_y+D_z=0$. This means $L=3m,~m\in \mathbb N^+$. To set the 1st order AHT to zero generally requires to satisfy two conditions: $[D_x,D_x]=0$ and $[D_x,D_y]+[D_x,D_z]=-[D_x,D_x]=0$. The first one is a trivial requirement, while the second one requires an even number of commutators. For sequences of length $L=3m$, the total number of commutators, when subtracting the trivial one,  is
	\be
	N_{c}=\frac{L(L-1)}{2}-3\frac{m(m-1)}{2}=3m^2
	\ee
	To have an even $N_{c}$,  $m$ must be even. Hence, we have $L=6n,~n\in \mathbb N^+$.
\end{proof}
Although the length of a solid echo is only $2\tau$, it does not qualify as a decoupling sequences as defined here, because the average Hamiltonian is $D_y+D_z=-D_x$ and thus the sequence only protects the $X$ state. The shortest known decoupling sequence is WAHUHA, whose length is indeed $6\tau$, though only contains 4 pulses.

\subsection{Offset and disorder in \textit{yxx} sequences}
Here we show in detail that a yxx sequence, with vanishing zeroth and first order average Hamiltonian in the presence of offset, must have vanishing zeroth and first order average Hamiltonian in the presence of disorder, if there are no other imperfections. 

The first order average Hamiltonian is
\begin{equation}\label{eq:1stAH}
\begin{aligned}
H_A^{(1)}&=-\frac{i}{2M\tau}\int_0^{M\tau}dt_1\int_0^{t_1} dt_2\left\{[H_0(t_1),H_0(t_2)]\right.\\
&+[H_0(t_1),H_z(t_2)]+[H_z(t_1),H_0(t_2)]\\
&\left.+[H_z(t_1),H_z(t_2)]\right\},
\end{aligned}
\end{equation}
where $H_0(t)$ and $H_z(t)$ represents the interaction and imperfection Hamiltonian in the toggling frame at time $t$.
The first commutator integrates to zero as guaranteed by the \textit{yxx} pattern; the last commutator is a single-site operator so the argument for zeroth order average Hamiltonian (see Appendix) also applies here; the two commutators in the middle line integrates to the same value so we need to consider only one of them
\begin{equation}\label{eqn:appidentity}
\begin{aligned}
&\quad\int_0^{M\tau}\ud t_1\int_0^{t_1}\ud t_2[H_0(t_1),H_z(t_2)]\\
&=-\int_0^{M\tau}\ud t_1\int_{t_1}^{M\tau}\ud t_2[H_0(t_1),H_z(t_2)]\\
%=\int\ud t_1\int_{t_1}\ud t_2[B(t_2),A(t_1)]
&=\int_0^{M\tau}\ud t_2\int^{t_2}_0\ud t_1[H_z(t_2),H_0(t_1)]\\
&=\int_0^{M\tau}\ud t_1\int^{t_1}_0\ud t_2[H_z(t_1),H_0(t_2)],
\end{aligned}
\end{equation}
where in the first equality we use the fact that zeroth order average Hamiltonian vanishes $\int\ud t H_z(t)=0$. Later on we consider only the first term $[H_0(t_1),H_z(t_2)]$ with $t_1>t_2$. For easier discussion, we divide the entire sequence into xyx blocks as in Fig. 4 in the main text and each block contains 3 time interval $\tau$ (here we do not rotate the first half interval to the end, so that within each interval the toggling frame Hamiltonian does not change). We use $j=0,1,\cdots,(M/3-1)$ to label the blocks and $k=0,1,2$ to label the intervals within the block, therefore denoting the piece-wise constant toggling frame Hamiltonian as $H_0^{j,k}=H_0(t)$, $H_z^{j,k}=H_z(t)$ with $3j\tau+k\tau<t<3j\tau+k\tau+\tau$. Then we can rewrite the integration as summation 
\begin{equation}\label{eq:sum}
\begin{aligned}
&\quad\int_0^{M\tau}\ud t_1\int^{t_1}_0\ud t_2[H_0(t_1),H_z(t_2)]\\
&=\sum_{\substack{j,k,j',k'\\3j+k>3j'+k'}}[H_0^{j,k},H_z^{j',k'}]\tau^2.
\end{aligned}
\end{equation}
For even $j$, $H_0^{j,0}=D_z,H_0^{j,1}=D_y,H_0^{j,2}=D_x$ and $H_z^{j,0}=\pm H_z, H_z^{j,1}=\pm H_y,H_z^{j,2}=\pm H_x$; for odd $j$, $H_0^{j,0}=D_z,H_0^{j,1}=D_x,H_0^{j,2}=D_y$ and $H_z^{j,0}=\pm H_z, H_z^{j,1}=\pm H_x,H_z^{j,2}=\pm H_y$. 
The inter-block commutators sum up to zero, because $\sum_{k}[H_0^{j,k},H_z^{j',k'}]=[D_x+D_y+D_z,H_z^{j',k'}]=[0,H_z^{j',k'}]=0$. Therefore we can further simplify \eqnref{eq:sum} as 
\begin{equation}\label{eq:sumblock}
\int_0^{M\tau}\ud t_1\int^{t_1}_0\ud t_2[H_0(t_1),H_z(t_2)]=\sum_{\substack{k',j,k\\k>k'}}[H_0^{j,k},H_z^{j,k'}]\tau^2.
\end{equation}
For this commutator, the only difference between disorder and offset is that $[D_\alpha, H_\alpha]\neq0$ for disorder, while for offset $[D_\alpha,H_\alpha]=0$, with $\alpha=x,y,z$. As a result, the commutator of the form $[D_\alpha,H_\alpha]$ give rise to the difference between disorder and offset. This term can only appear as $[H_0^{j,1}+H_0^{j,2},H_1^{j,0}]=[-D_z,\pm H_z]$. When sum over $j$, $\pm H_z$ sums to zero because zeroth order average Hamiltonian vanish, therefore $\sum_j H_0^{j,1}+H_0^{j,2},H_1^{j,0}]=0$. In conclusion, disorder may induce more nonzero first order average Hamiltonian than offset, but they sum up to zero, therefore a \textit{yxx} sequence with vanishing zeroth and first order average Hamiltonian for the offset case must also have vanishing zeroth and first order average Hamiltonian for the disorder case. 

\end{document}